\documentclass[11pt]{article}
\usepackage[margin=1in]{geometry}
\usepackage{amsmath}
\usepackage{graphicx}
\usepackage{mathtools}
\usepackage{comment}
\usepackage{enumerate}

\usepackage{url} 
\newcommand{\blind}{1}

\usepackage{amsmath,amsthm, amssymb, bbm}
\usepackage{blindtext}
\usepackage{enumerate}
\usepackage{listings}
\usepackage{color, colortbl}
\usepackage{xcolor}
\usepackage[nottoc]{tocbibind}
\usepackage[utf8]{inputenc}
\usepackage{caption}
\usepackage{tabularx}
\usepackage{array}
\usepackage{rotating}
\usepackage{wrapfig}
\usepackage{graphicx}
\usepackage{subcaption}
\definecolor {processblue}{cmyk}{0.96,0,0,0}
\usepackage{mathrsfs}
\usepackage{hyperref}
\hypersetup{
	colorlinks   = true,
	linkcolor = red,
	citecolor    = blue
}
\usepackage{fancyvrb}
\usepackage{enumitem}
\definecolor{LightCyan}{rgb}{0.88,1,1}
\usepackage{parskip}
\usepackage{float}
\usepackage{titling}
\usepackage{varwidth}
\usepackage{multirow}
\usepackage{breqn}
\usepackage{bm}
\usepackage{titlesec}
\titleformat{\subsection}    
{\bfseries\itshape}{\thesubsection}{1em}{}
\usepackage{xr}
\usepackage{setspace}

\newtheorem{lemma}{Lemma}
\newtheorem{proposition}{Proposition}

\theoremstyle{remark}
\newtheorem{definition}{Definition}

\newtheorem{remark}{Remark}
\usepackage{amssymb,mathtools}

\makeatletter
\newcommand*{\addFileDependency}[1]{
	\typeout{(#1)}
	\@addtofilelist{#1}
	\IfFileExists{#1}{}{\typeout{No file #1.}}
}
\makeatother

\def\indist{\rightsquigarrow}
\def\ind{\perp\!\!\!\perp}

\newcommand{\var}{\text{var}}
\newcommand{\DR}{\text{DR}}

\newcommand{\cov}{\text{cov}}
\newcommand{\Pb}{\mathbb{P}}

\newcommand{\Pn}{\mathbb{P}_n}

\newcommand{\E}{\mathbb{E}}
\newcommand{\R}{\mathbb{R}}
\newcommand{\N}{\mathbb{N}}

\def\logit{\text{logit}}

\newcommand{\one}{\mathbbm{1}}

\usepackage[backend=bibtex, style=nature, doi=true, url=true, citestyle=authoryear-comp, intitle=true, maxbibnames=99]{biblatex}
\bibliography{ref.bib}

\begin{document}

	\def\spacingset#1{\renewcommand{\baselinestretch}%
		{#1}\small\normalsize} \spacingset{1}

	
	\if1\blind
	{
		\title{\bf Doubly-robust inference and optimality in structure-agnostic models with smoothness}
		\author{Matteo Bonvini\thanks{Department of Statistics, Rutgers, the State University of New Jersey. Email: mb1662@stat.rutgers.edu.} \and Edward H. Kennedy\thanks{Department of Statistics \& Data Science, Carnegie Mellon University. Email: edward@stat.cmu.edu.} \and Oliver Dukes\thanks{Department of Applied Mathematics, Computer Science and Statistics, Ghent University. Email: Oliver.Dukes@ugent.be} \and Sivaraman Balakrishnan\thanks{Department of Statistics \& Data Science, Carnegie Mellon University. Email: siva@stat.cmu.edu.}}
		\date{ \today \\ \medskip \textit{Preliminary. Comments welcome.}}
		\maketitle
	} \fi
	
	\if0\blind
	{
		\bigskip
		\bigskip
		\bigskip
		\begin{center}
			{\bf Title}
		\end{center}
		\medskip
	} \fi
	
	\begin{abstract}
    We study the problem of constructing an estimator of the average treatment effect (ATE) with observational data. The celebrated doubly-robust, augmented-IPW (AIPW) estimator generally requires consistent estimation of both nuisance functions for standard root-n inference, and moreover that the product of the errors of the nuisances should shrink at a rate faster than $n^{-1/2}$. A recent strand of research has aimed to understand the extent to which the AIPW estimator can be improved upon (in a minimax sense). Under structural assumptions on the nuisance functions, the AIPW estimator is typically not minimax-optimal, and improvements can be made using higher-order influence functions (Robins et al, 2017). Conversely, without any assumptions on the nuisances beyond the mean-square-error rates at which they can be estimated, the rate achieved by the AIPW estimator is already optimal (Balakrishnan et al, 2023; Jin and Syrgkanis, 2024).
    
    We make three main contributions. First, we propose a new hybrid class of distributions that combine structural agnosticism regarding the nuisance function space with additional smoothness constraints. Second, we calculate minimax lower bounds for estimating the ATE in the new class, as well as in the pure structure-agnostic one. Third, we propose a new estimator of the ATE that enjoys doubly-robust asymptotic linearity; it can yield  asymptotically valid Wald-type confidence intervals even when the propensity score or the outcome model is inconsistently estimated, or estimated at a slow rate. Under certain conditions, we show that its rate of convergence in the new class can be much faster than that achieved by the AIPW estimator and, in particular, matches the minimax lower bound rate, thereby establishing its optimality. Finally, we complement our theoretical findings with simulations.
	\end{abstract}
	\section{Introduction}
	The effect of a binary random variable $A$ on an outcome $Y$ is often measured by the average treatment effect (ATE). Letting $Y^a$ denote the potential outcome that one would observe had treatment been set to $A = a$, the ATE is defined as $\E(Y^1 - Y^0)$. Suppose a sufficiently rich set of covariates $X \in \R^p$ is collected. Under consistency ($A = a \implies Y^a = Y$), positivity ($\Pb(A = 1 \mid X) \in (0, 1)$) and no-unmeasured-confounding ($A \ind Y^a \mid X$), we have 
	\begin{align*}
		\E(Y^1 - Y^0) = \E\{\E(Y \mid A= 1, X) - \E(Y \mid A = 0, X)\}.
	\end{align*}
	In this work, we consider the problem of estimating $\E(Y^1 - Y^0)$ identified as above. For simplicity, however, we focus on $\psi \equiv \E\{\E(Y \mid A = 1, X)\}$, with the understanding that the same results apply to $\E\{\E(Y \mid A = 0, X)\}$. The parameter $\psi$ captures the average outcome if every unit in the population takes $A = 1$; it can also be interpreted as the population mean outcome under missingness at random \parencite{rubin1976inference}. Here, we consider estimation of $\psi$ regardless of its interpretation.
	
	One key feature of $\psi$ is that it can be estimated at $n^{-1/2}$ rates even in nonparametric models where the best rate of convergence for estimating $\mu(X) = \E(Y \mid A= 1, X)$ is slower than $n^{-1/2}$. One way to see this is to consider a randomized trial whereby the probability of receiving treatment given the measured covariates, $\Pb(A = 1 \mid X)$, is known. Because $\psi$ can be expressed as $\psi = \E[AY\omega(X)]$, for $\omega(X) = \{\Pb(A = 1 \mid X)\}^{-1}$, by Chebyshev's inequality, it can be estimated at $n^{-1/2}$ rates simply as $n^{-1}\sum_{i = 1}^n A_i Y_i\omega(X_i)$. In non-randomized studies, however, both nuisance functions $\omega(X)$ and $\mu(X)$ are unknown. In this sense, the convergence rate of an estimator of $\psi$ will typically depend on how accurately these nuisance parameters can be estimated. A vast literature has thus focused on weakening the dependence of the estimator on the nuisance functions' estimation error.
	
	In the context of nonparametric modeling, many widely adopted estimators of $\psi$ rely, in some form, on estimation of both $\mu(X)$ and $\omega(X)$. The augmented-inverse-probability-weighted (AIPW) estimator \parencite{robins1994estimation} and those based on Targeted-Maximum-Likelihood Estimation (TMLE) \parencite{van2011targeted, van2018targeted} are two prominent examples. In particular, these two approaches, based on the first-order influence function of $\psi$, are agnostic with respect to how the nuisances are estimated. Other approaches consider more specific, clever estimators of these nuisances, but can still be considered as variants of the two approaches above. Finally, the list of available estimators of $\psi$ also includes other strategies, such as those based on matching (see, e.g., \cite{imbens2004nonparametric}).
	
	In non-randomized studies, where both $\mu(X)$ and $\omega(X)$ are unknown, the state-of-the-art to conduct inference on $\psi$ is to assume that both nuisances are estimated well enough. A standard requirement is that the product of the root-mean-square-errors is asymptotically negligible, i.e., it converges in probability to 0 when scaled by $\sqrt{n}$. In this case, an asymptotically valid confidence interval is simply the Wald interval. 
	
	When it is not possible to estimate the nuisances with the accuracy required for the validity of the Wald interval, one option is to attempt to estimate and take into account the bias of the original estimator. This idea is connected to the development of the general theory of Higher Order Influence Functions (HOIFs) \parencite{robins2008higher, robins2017higher, robins2009quadratic, robins2009semiparametric, van2018higher}, as well as the discovery of estimators that can be $\sqrt{n}$-consistent even if the model for $\mu(X)$ or $\omega(X)$ is misspecified \parencite{van2014targeted, benkeser2017doubly, dukes2021doubly}. In this work, we build upon these two streams of literature and propose a novel estimator of $\psi$, which remains $\sqrt{n}$-consistent even if one of the two nuisances is not consistently estimated (or estimated at a slow rate). A more detailed list of our contributions, including minimax lower bounds, appears in Section \ref{sec:main_contributions}, after introducing notation and the basic problem statement. 
	
	\subsection{Notation}
	We assume that one observes $n$ iid copies $O^n = O_1, \ldots, O_n$, where $O = (Y, A, X) \in \mathcal{Y} \times \{0, 1\} \times \R^d$. Let $f(x)$ denote the density of $X$ and
	\begin{align*}
		& \pi(X) = \Pb(A = 1 \mid X), \quad \omega(X) = 1/\pi(X),  \quad \mu(X) = \E(Y \mid A = 1, X), \quad \text{and} \quad g(X) = \pi(X)f(X).
	\end{align*}
	The parametrization of the density in terms of $\omega(x)$ instead of $\pi(x)$ is natural and convenient when deriving the lower bounds results in Section \ref{sec:lower_bounds}, but it is not essential for deriving the properties of the estimators described in Section \ref{sec:upper_bounds}. With this notation, we can write
	\begin{align*}
		\psi & = \E\{\mu(X)\} = \E\left\{ \frac{AY}{\pi(X)}\right\} = \E\left\{ \omega(X)AY \right\} = \int \omega(x) \mu(x) g(x) dx.
	\end{align*}
	Let us use the notation $\Pn f = n^{-1} \sum_{i = 1}^n f(O_i)$, $\Pb f = \int f(o) d\Pb(o)$ and $\|f\|^2 = \Pb f^2$. We also let $a \lesssim b$ denote $a \leq C b$ for some constant $C$ that is independent of the sample size $n$, and let $a \land b$ denote $\min(a, b)$. To lighten the notation, when it does not induce confusion, we will also use $\omega = \omega(X)$, $\mu = \mu(X)$, $\omega_i = \omega(X_i)$ and $\mu_i = \mu(X_i)$. A similar notation will be used for the estimators $\widehat\omega(x)$ and $\widehat\mu(x)$. We also use the notation $\E_{X_1 \mid X_2}\{f(X_1)\} = \E\{f(X_1) \mid X_2\}$. Throughout, we assume that $\widehat\omega(x)$ and $\widehat\mu(x)$ are computed on a sample $D^n$ that is independent of $O^n$. This can be accomplished by splitting $O^n$ in subsamples and then swapping the roles of the samples for training the nuisance estimators. Further, we assume that all observations and nuisance functions are bounded; in particular, we assume that $\widehat\omega(x), \omega(x) \in [1, M_a]$, $|\mu(x)| \leq M_\mu$ and $|\widehat\mu(x)| \leq M_\mu$, for some constants $M_\omega$ and $M_\mu$.
	
	Finally, recall the definition of a H\"{o}lder smooth function. We consider this function class when we introduce our estimators in Sections \ref{sec:max_estimators} and \ref{sec:main_estimator}.
	\begin{definition}\label{Holder_def}
		Let $\alpha \in [0, 1]$ and $C$ be a positive constant. A function $f(x)$ is H\"{o}lder of order $\alpha$ if
		\begin{align*}
			|f(x) - f(y)| \ \leq C \|x - y\|^\alpha \text{ for every } x, y \text{ in the domain of } f.
		\end{align*}
	\end{definition}
	Because our subsequent results only pertain to the vanishing rate of the mean-square-errors of our estimators as a function of the sample size, we will not keep track of constants. In this light, we will say that a function is smooth of order $\alpha$ if it satisfies Definition \ref{Holder_def} for some constant $C$. Finally, letting $\lfloor \alpha \rfloor$ denote the greatest integer less than $\alpha$,  we say a function $f$ is H\"{o}lder of order $\alpha > 1$ if $f$ is $\lfloor \alpha \rfloor$-times differentiable with $\lfloor \alpha \rfloor^{\text{th}}$ derivative H\"{o}lder of order $\alpha - \lfloor \alpha \rfloor$ in the sense of Definition \ref{Holder_def}, and if $f$ has all derivatives up to order $\lfloor \alpha \rfloor$ bounded above by some constant.
	
	\subsection{Problem statement} 
	The AIPW estimator, also known as the doubly-robust (DR) estimator, is defined as
	\begin{align}\label{eq:aipw}
		& \widehat\psi_{\DR} = \frac{1}{n} \sum_{i = 1}^n A_i\widehat\omega(X_i)\{Y_i - \widehat\mu(X_i)\} + \widehat\mu(X_i) \equiv \Pn \widehat\varphi,
	\end{align}
	where $\varphi(O) = A\omega(X)\{Y - \mu(X)\} + \mu(X)$ is the (uncentered) influence function of $\psi$. The variance $\var(\varphi)$ is the semiparametric efficiency bound for estimating $\psi$ in nonparametric models, i.e. the smallest variance any regular estimator of $\psi$ can achieve \parencite{tsiatis2006semiparametric, kennedy2022semiparametric, newey1990semiparametric}. Under certain conditions, $\widehat\psi_\DR$ achieves this semiparametric bound and it is thus efficient. This can be seen from the following decomposition. Let $\overline\omega$ and $\overline\mu$ denote the limits, as $n \to \infty$, of $\widehat\omega$ and $\widehat\mu$. In this sense, let us assume, without essential loss of generality, that $\|\widehat\omega - \overline\omega\| = o_\Pb(1)$ and $\|\widehat\mu - \overline\mu\| = o_\Pb(1)$. Letting $\overline\varphi(O) = A\overline\omega(X)\{Y - \overline\mu(X)\} + \overline\mu(X)$, we have
	\begin{align*}
		\widehat\psi_\DR - \psi = (\Pn - \Pb) (\widehat\varphi - \overline\varphi) + (\Pn - \Pb) \overline\varphi + \Pb(\widehat\varphi - \varphi).
	\end{align*}
	The central term, when scaled by $\sqrt{n}$, converges to $N(0, \var(\overline\varphi))$ by the central limit theorem. Thus, by Slutsky's theorem, $\sqrt{n}(\widehat\psi - \psi) \indist N(0, \var(\overline\varphi))$ as long as the first and last terms on the right-hand-side are $o_\Pb(n^{-1/2})$. Notice that $\var(\overline\varphi) = \var(\varphi)$ if $\overline\omega = \omega$ and $\overline\mu = \mu$. Next, in light of Lemma 2 in \cite{kennedy2018sharp} (restated in the Appendix as Lemma \ref{lemma:edward}) and the fact that $\widehat\omega$ and $\widehat\mu$ are trained on a separate sample, the first term is $o_\Pb(n^{-1/2})$. The most difficult term to control is the third one, which equals
	\begin{align}\label{eq:Rn}
		R_n \equiv \Pb(\widehat\varphi - \varphi) = \int \{\widehat\omega(x) - \omega(x)\}\{\mu(x) - \widehat\mu(x)\} g(x)dx.
	\end{align}
	In this light, we rewrite the decomposition $\widehat\psi_\DR - \psi$ as
	\begin{align}\label{eq:main_decomposition}
		\widehat\psi_\DR - \psi = (\Pn - \Pb)\overline\varphi + R_n + o_\Pb(n^{-1/2}).
	\end{align}
	For inference, we assume throughout that equation \eqref{eq:main_decomposition} holds. The quantity $R_n$ is the conditional bias of $\widehat\psi_\DR$ given the training sample $D^n$. By the Cauchy-Schwarz inequality, 
	\begin{align*}
		|R_n| \ \lesssim \|\widehat\omega - \omega\| \|\widehat\mu - \mu\| \implies \widehat\psi_\DR - \psi = O_\Pb\left(n^{-1/2} + \E\{\|\widehat\omega - \omega\| \|\widehat\mu - \mu\|\}\right).
	\end{align*}
	Thus, in general, asymptotic negligibility of this term is guaranteed if 
	\begin{align*}
		\E\{\|\widehat\omega - \omega\| \|\widehat\mu - \mu\|\} \leq \{\E(\|\widehat\omega - \omega\|^2)\}^{1/2}\cdot \{\E(\|\widehat\mu - \mu\|^2)\}^{1/2} = o(n^{-1/2}),
	\end{align*}
	leading to the standard $n^{-1/4}$-rate requirement on the nuisance root-mean-square-errors. This is remarkable because $n^{-1/4}$-rates are achievable in nonparametric function classes under appropriate structural conditions, e.g. sufficient smoothness or sparsity. For example, if $\omega$ and $\mu$ are H\"{o}lder smooth of order $s$ and are estimated by rate-optimal estimators (see, e.g., Section 1.6 in \cite{tsybakov2008introduction} for a discussion on local polynomials), then $\E\{\|\widehat\omega - \omega\| \|\widehat\mu - \mu\|\} = o(n^{-1/2})$ follows if $s > d / 2$. More generally, for some sequences of constants $\epsilon_n$ and $\delta_n$, the Cauchy-Schwarz bound yields that $\E|\widehat\psi_\DR - \psi| \ \lesssim n^{-1/2} + \delta_n \epsilon_n$ whenever $\|\widehat\omega - \omega\| \ \leq \epsilon_n$ and $\|\widehat\mu - \mu\| \ \leq \delta_n$.
	
	When the condition ensuring $\epsilon_n \delta_n = o(n^{-1/2})$ fails, the Cauchy-Schwarz bound does not guarantee that the bias of $\widehat\psi_\DR$ is vanishing at a rate faster than $n^{-1/2}$, the order of the standard error. The foundational work on Higher Order Influence Functions (HOIFs) by Robins and co-authors offers a general, principled way to carry out functional estimation optimally. This approach has led to new estimators of $\psi$ based on higher-order $U$-statistics, which have been shown to be optimal in certain nonparametric models \parencite{robins2008higher, robins2009quadratic, robins2017minimax, liu2017semiparametric}. This general theory considerably expands the use of U-statistics for optimal functional estimation, which has a long history in statistics; see, e.g., the literature on estimating integral functionals of a density \parencite{bickel1988estimating, laurent1996efficient, laurent1997estimation, birge1995estimation}. With respect to the settings considered here, higher order corrections for estimating $\psi$ can be viewed as effectively estimating $R_n$ and subtracting it off from $\widehat\psi_\DR$, leading to better bias-variance trade-offs. Higher order corrections can also be done in the context of TMLE \parencite{van2018higher, diaz2016second, van2021higher}. With a more direct focus on inference, \cite{van2014targeted} and \cite{benkeser2017doubly} have discovered TMLE-based estimators of $\psi$ that are $\sqrt{n}$-consistent and asymptotically normal even if either $\widehat\omega$ or $\widehat\mu$ is not consistent. This represents an improvement upon $\widehat\psi_\DR$ because, at best, the rate of convergence for $\widehat\omega$ and $\widehat\mu$ is of order $n^{-1/2}$ (corresponding to correctly specified parametric models), so that, if the Cauchy-Schwarz bound on $|R_n|$ is employed, $R_n$ would \textit{not} be asymptotically negligible if $\widehat\omega$ or $\widehat\mu$ are inconsistent. We remark that their estimators do not directly build on the theory of HOIFs and, in particular, do not employ U-statistics. Finally, a growing body of literature (e.g., \cite{bradic2019sparsity}, \cite{zhang2021dynamic}, \cite{celentano2023challenges}, \cite{liu2023root}, \cite{bradic2024high} and references therein) considers the challenging task of estimating $\psi$ (and related functionals) under potential mis-specification of high-dimensional models for the nuisance functions. Our work is distinct from this stream of literature as we do not posit high dimensional models for $\omega$ and $\mu$. As we strive to be agnostic with respect to the analyst's choice for $\widehat\omega(x)$ and $\widehat\mu(x)$, such models are allowed within the framework we consider, but our analysis does not leverage any additional structure that they may induce on the problem.
	
	\subsection{Main contributions}\label{sec:main_contributions}
	
	In this work, we make at least three main contributions. First, we propose and develop a new function class, which is data-dependent and can be viewed as a hybrid between the completely structure-agnostic ones considered in \cite{balakrishnan2023fundamental} and more traditional ones based on smoothness conditions. The structure-agnostic class of distributions considered in \cite{balakrishnan2023fundamental} is defined as the set of all distributions $\mathcal{P}(\epsilon_n, \delta_n)$ for which it holds that $\|\widehat\omega - \omega\| \ \leq \epsilon_n$ and $\|\widehat\mu - \mu\| \ \leq \delta_n$, for some rates $\epsilon_n$ and $\delta_n$. We study a subset of this class that additionally impose some smoothness constraints on certain regression functions for which the estimators $\widehat\omega(X)$ and $\widehat\mu(X)$ enter as covariates. We find that the convergence rates admitted over this more regular class can be potentially much faster than those holding over the completely structure-agnostic ones studied in \cite{balakrishnan2023fundamental}. As our proposed function class does not directly impose regularity restrictions on $\omega$ and $\mu$ and yet can potentially admit fast rates of convergence, it can provide a nice middle ground between complete agnosticism at one extreme and much more structured smoothness, say encoded in H\"{o}lder smoothness restrictions on $\omega$ and $\mu$, at the other.  
	
	Our second contribution is to provide minimax lower bounds for estimating $\psi$ in the new hybrid class proposed, as well as in the pure structure agnostic one. We find that, over $\mathcal{P}(\epsilon_n, \delta_n)$, i.e., if the rate-condition for estimating $\omega$ and $\mu$ is the only information available (together with mild boundedness regularity conditions), the rate of convergence $n^{-1/2} + (\epsilon_n \delta_n)$ for $\psi$ is not improvable in a minimax sense. This shows that, in this framework, $\widehat\psi_\DR$ is not improvable without the introduction of additional assumptions. Our current construction only covers the case $\epsilon_n \leq \delta_n$. However, in a work concurrent to ours and developed independently, \cite{jin2024structure} show that the bound $\epsilon_n \delta_n$ also holds when $\epsilon_n > \delta_n$ (at least when $\epsilon_n = o(1)$ and $\delta_n = o(1)$); their proof is conceptually similar to ours but employs a different parametrization of the data generating process. 
	
	On the contrary, the lower bound derived for estimating $\psi$ in the proposed hybrid model is of order $n^{-1/2} + (\epsilon_n^2 \land \delta^2_n)$. In this sense, imposing additional smoothness constraints on $\mathcal{P}(\epsilon_n, \delta_n)$ allows for potentially much faster rates of convergence. The rate $\epsilon_n^2 \land \delta_n^2$  does not rule out the existence of valid confidence intervals shrinking at the rate $n^{-1/2}$ even if one of the two nuisance functions is not consistently estimated or the rate of convergence is too slow, i.e., even if one between $\epsilon_n$ or $\delta_n$ does not converge to zero (fast enough). It thus allows for the possibility of conducting doubly-robust root-$n$ inference. As mentioned above, in virtue of the lower bound rate $\epsilon_n \delta_n$, the purely structure-agnostic class of distributions studied in \cite{balakrishnan2023fundamental} and \cite{jin2024structure} does not allow for doubly-robust root-$n$ inference in nonparametric nuisance function classes. 
	
	Our third contribution is the construction and analysis of a new estimator that achieves the rate $n^{-1/2} + (\epsilon_n^2 \land \delta_n^2)$ under certain conditions. As this rate matches our minimax lower bound rate for estimating $\psi$ in the proposed hybrid function class, this new estimator is a minimax optimal one. We view our new estimator as a hybrid between the estimator of the ATE based on the approximate second-order influence function \parencite{robins2009quadratic, diaz2016second} and the estimator(s) considered in \cite{van2014targeted}, \cite{benkeser2017doubly} and \cite{dukes2021doubly} that are specifically designed to conduct doubly-robust inference. Our estimator can be used for doubly-robust inference under arguably more transparent conditions than the ones previously considered in the literature, which did not posit a hybrid smoothness/structure-agnostic model like we do here. In addition, our constructions can be easily adjusted to conduct doubly-robust inference in settings where this is not readily feasible using currently available estimators, such as to estimate the parameters in a partially linear logistic model (see Appendix \ref{appendix:partially_linear_logit}). Since we posted the first version of this work on arXiv, a new pre-print by \cite{van2024automatic} has been uploaded; tackling a similar problem as ours, the authors present an intriguing approach based on isotonic calibration. In future work, it would be interesting and helpful to compare the different methods. Finally, we evaluate the performance of the newly derived estimator against that of $\widehat\psi_\DR$ and that of the estimator described in \cite{benkeser2017doubly} and implemented in the \texttt{R} package \texttt{drtmle} \parencite{benkeser2023doubly}. The code to replicate the simulations can be found at \url{https://github.com/matteobonvini/dr_inference}. 
	
	\section{Structure-agnostic viewpoint \label{sec:lower_bounds}}
	\subsection{Purely structure-agnostic class of distributions}
	In this section, we describe the optimality viewpoint introduced in \cite{balakrishnan2023fundamental}. They consider the problem of functional estimation with nuisance parameters when all that is known are convergence rates for estimating the nuisance components. In our setting, we aim to derive the best possible rate for estimating $\psi$ when the only information available are bounds of the form $\|\widehat\omega - \omega\| \ \leq \epsilon_n$ and $\| \widehat\mu - \mu\| \ \leq \delta_n$. This framework is particularly helpful for understanding how precisely one can estimate $\psi$ without imposing structural assumptions on the data generating process; e.g., smoothness or sparsity on $\omega(x)$ and $\mu(x)$. In deriving the lower bounds, we assume that $Y$ is binary, $X$ is supported in $[0, 1]^d$ and $c \leq \widehat\mu(x) \leq 1- c$ for some constant $c > 0$ and all $x$.
	
	Given two arbitrary estimators $\widehat\omega(x)$ and $\widehat\mu(x)$, the class $\mathcal{P}(\epsilon_n, \delta_n)$ consists of all densities such that $\|\widehat\omega - \omega\| \ \leq \epsilon_n$ and $\|\widehat\mu - \mu\| \ \leq \delta_n$. Our first result is that the rate of convergence of any estimator of $\psi$ over this class cannot be faster than $n^{-1/2} + \epsilon_n \cdot \delta_n$. Our current proof breaks for data generating processes where $\epsilon_n > \delta_n$. However, both cases are covered by the concurrent work by \cite{jin2024structure}. The class $\mathcal{P}(\epsilon_n, \delta_n)$ aims to describe the setting where one constructs $\widehat\omega(\cdot)$ and $\widehat\mu(\cdot)$ on a separate independent training sample. The information available would then be encoded in the form of high-probability bounds for $\|\widehat\omega - \omega\| \ \leq \epsilon_n$ and $\|\widehat\mu - \mu\| \ \leq \delta_n$. Without essential loss of generality, in what follows, we assume these bounds hold exactly (not just with high probability) and refer to Section 3.1.1 in \cite{balakrishnan2023fundamental} for further discussion. 
	
	\begin{proposition}\label{prop:pure_sa}
		Let $\mathcal{P}(\epsilon_n, \delta_n)$ denote the class of all densities such that $\sup_{p \in \mathcal{P}(\epsilon_n, \delta_n)} \| \omega_p - \widehat\omega\|_2 \ \leq \epsilon_n$ and $\sup_{p \in \mathcal{P}(\epsilon_n, \delta_n)} \|\mu_p - \widehat\mu\|_2 \ \leq \delta_n$.  Then, provided that $\epsilon_n \leq \delta_n$,
		\begin{align*}
			\inf_{T_n} \sup_{p \in \mathcal{P}(\epsilon_n, \delta_n)} \E|T_n - \psi_p| \ \gtrsim \ \epsilon_n \delta_n.
		\end{align*}
	\end{proposition}
	\begin{proof}
		See \cite{jin2024structure} for a full proof that also covers the case when $\epsilon_n > \delta_n$. Our proof relying on $\epsilon_n \leq \delta_n$ is reported in Appendix \ref{appendix:prop_pure_sa}.
	\end{proof}
	
	Proposition \ref{prop:pure_sa} establishes that, if the only assumption on the data generating process consists of error bounds for estimating $\omega$ and $\mu$, then $\psi$ cannot be estimated at a rate faster than the product of these bounds. This means that, to improve upon the AIPW estimator $\widehat\psi_\DR$, which achieves this rate under mild conditions, one needs to introduce other assumptions in addition to rate conditions on the nuisance components. Furthermore, if either $\omega$ or $\mu$ is inconsistently estimated, so that either $\epsilon_n$ or $\delta_n$ do not vanish as $n$ goes to infinity, then $\psi$ cannot be estimated at a rate faster than $\epsilon_n \land \delta_n$, without introducing other assumptions. This implies that nonparametric, doubly-robust root-$n$ inference is possible only if one relies on additional conditions. Note that this is meant as a clarifying technical statement; we do not claim that relying on such additional conditions should necessarily be avoided in practice. 
    
    \begin{wrapfigure}{r}{0.5\textwidth}
        \centering
    \includegraphics[width=0.28\textwidth]{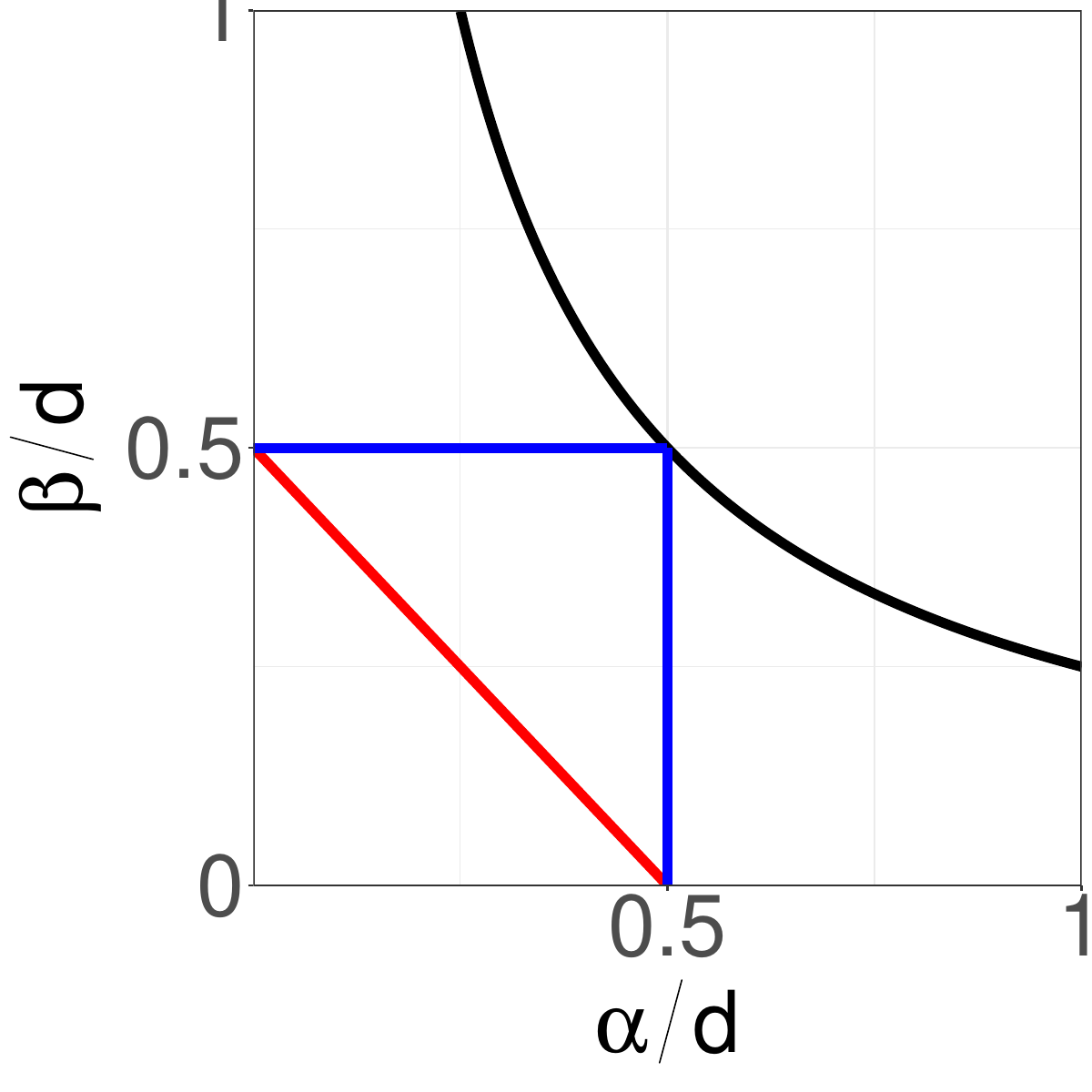}
    \caption{Possible frontiers for estimating $\psi$ at root-$n$ rates in the H\"{o}lder-smoothness model. The black line is the frontier obtained by the AIPW via the Cauchy-Schwarz bound, the red line refers to that obtained by HOIFs estimators fully exploiting the structural smoothness, while the blue line refers to the frontier obtained by estimators converging at a rate $\|\widehat\mu - \mu\|^2 \land \|\widehat\omega - \omega\|^2$. \label{fig:rates}}
\end{wrapfigure}
	We conclude this section with a remark highlighting a connection between the result from Proposition \ref{prop:pure_sa} and the seminal results on optimal estimation of functionals indexed by H\"{o}lder nuisance components.
	\begin{remark}
		Suppose that $\omega$ and $\mu$ are H\"{o}lder smooth functions on a $d$-dimensional domain of orders $\alpha$ and $\beta$. Then, there exist rate-optimal estimators in these classes that satisfy, with high probability, $\|\widehat\omega - \omega\| \ \lesssim n^{-\alpha/(2\alpha + d)}$ and $\|\widehat\mu - \mu\| \ \lesssim n^{-\beta / (2\beta + d)}$, i.e., $\epsilon_n = C_a n^{-\alpha/(2\alpha + d)}$ and $\delta_n = C_b n^{-\beta /(2\beta + d)}$ for some constants $C_a$ and $C_b$. Suppose that the good event $\|\widehat\omega - \omega\| \leq \epsilon_n$ and $\|\widehat\mu - \mu\| \leq \delta_n$ holds. Proposition \ref{prop:pure_sa} (in the case of $\alpha \geq \beta$) or more generally Theorem 2.1 in \cite{jin2024structure} yield that no estimator of $\psi$ can achieve a rate faster than $n^{-1/2} + \epsilon_n \delta_n$. This rate is slower than the minimax rate for estimating $\psi$ in this model, which is of the order $n^{-1/2}  + n^{-4s / (4s + d)}$, where $s = (\alpha + \beta) / 2$ \parencite{robins2008higher, robins2009semiparametric}. This does not contradict Proposition \ref{prop:pure_sa}: the class $\mathcal{P}(n^{-\alpha/(2\alpha + d)}, n^{-\beta/(2\beta + d)})$ is larger than the class of densities such that $\omega$ and $\mu$ are $\alpha$- and $\beta$- H\"{o}lder smooth. In fact, the worst-case construction used to prove Proposition \ref{prop:pure_sa} relies on nuisance functions that are not necessarily H\"{o}lder smooth. In this sense, Proposition \ref{prop:pure_sa} suggests that imposing smoothness assumptions on $\omega$ and $\mu$ induces regularity on $\psi$ that is in addition to the regularity induced on $\omega$ and $\mu$. This then results in a model where $\widehat\psi_\DR$ is no longer optimal and one has to rely on additional corrections for optimal estimation, such as those based on HOIFs. This intuition aligns with the result that in a H\"{o}lder model where $\alpha + \beta > d / 2$, $\psi$ can be estimated at $n^{-1/2}$ rates without the need for consistent estimators of $\omega$ and $\mu$. However, if the nuisance functions' estimators are consistent, then the HOIFs-based estimator is also semiparametric efficient in the sense of having the smallest asymptotic variance among all regular estimators. See \cite{liu2017semiparametric}, particularly Corollary 5 and Remark 7. In Figure \ref{fig:rates}, we plot possible frontiers for root-$n$ consistency in the H\"{o}lder smoothness model.
	\end{remark}
	
	\subsection{Hybrid structure-agnostic class of distributions with smoothness}
	Next, we introduce our main hybrid class, as well as two other hybrid models that are tailored to settings where it known whether $\omega$ or $\mu$ is easier to estimate. In defining these classes, we depart from the structure agnostic framework introduced in \cite{balakrishnan2023fundamental}, encoded in $\mathcal{P}(\epsilon_n, \delta_n)$, to introduce certain smoothness conditions. In particular, we consider $\mathcal{P}_\omega(\epsilon_n)$ denoting the collection of all densities such that $\|\widehat\omega - \omega\| \ \leq \epsilon_n$ and $\E \{\mu(X) \mid A = 1, \widehat\omega(X) = t_1, \omega(X) = t_2, D^n\}$ is smooth. The reason why this class can be of interest in the context of estimating $\psi$ is that one can write $R_n$ as
	\begin{align*}
		R_n & = \E\left[ A \{\widehat\omega(X) - \omega(X)\} \{\mu(X) - \widehat\mu(X)\} \mid D^n\right] \\
		& = \E\left[ A \{\widehat\omega(X) - \omega(X)\}\E \{\mu(X) \mid A = 1, \widehat\omega(X), \omega(X), D^n\} \mid D^n \right] - \E\left[A \{\widehat\omega(X) - \omega(X)\}\widehat\mu(X) \mid D^n\right] \\
		& = \E\left[\{A\widehat\omega(X) - 1\}\E \{\mu(X) \mid A = 1, \widehat\omega(X), \omega(X), D^n\} \mid D^n \right] - \E\left[\{A\widehat\omega(X) - 1\}\widehat\mu(X) \mid D^n\right].
	\end{align*}
	If $\E \{\mu(X) \mid A = 1, \widehat\omega(X), \omega(X), D^n\}$ was known, $R_n$ could be estimated by a sample average with accuracy of order $n^{-1/2}$. The hope is that if this additional nuisance function is unknown, as it would be in practice, but smooth enough so that it can be estimated well, one may still estimate $\psi$ with $n^{-1/2}$-accuracy. A more structure-agnostic way to define $\mathcal{P}_\omega(\epsilon_n)$ would be to simply impose a rate condition on the accuracy with which $\E \{\mu(X) \mid A = 1, \widehat\omega(X), \omega(X), D^n\}$ can be estimated, as opposed to assuming this function is smooth. We leave this refinement for future work. 
	
	Motivated by writing $R_n$ as
	\begin{align*}
		R_n & =   \E\left[ A\{Y - \widehat\mu(X)\}\widehat\omega(X) \mid D^n\right] - \E\left[A\{Y - \widehat\mu(X)\} \E \{\omega(X) \mid A = 1, \widehat\mu(X), \mu(X), D^n\} \mid D^n \right],
	\end{align*}
	we also consider the class of densities $\mathcal{P}_\mu(\delta_n)$ such that $\|\widehat\mu - \mu\| \ \leq \delta_n$ and $\E \{\omega(X) \mid A = 1, \widehat\mu(X) = t_1, \mu(X) = t_2, D^n\}$ is smooth. 
	
	Finally, we consider $\mathcal{P}_{\omega\mu}(\epsilon_n, \delta_n)$, the main hybrid class that we propose. It restricts $\mathcal{P}(\epsilon_n, \delta_n)$ to include only densities for which $\E \{\omega(X) \mid A = 1, \widehat\mu(X) = t_1, \mu(X) = t_2, \widehat\omega(X) = t_3, D^n\}$ and $\E \{\mu(X) \mid A = 1, \widehat\omega(X) = t_1, \omega(X) = t_2, \widehat\mu(X) = t_3, D^n\}$ are smooth. 
	In the following proposition, we derive a minimax lower bound for each of the three hybrid classes considered.
	\begin{proposition}\label{siva_prop} We consider three cases:
		\begin{enumerate}
			\item Let $\mathcal{P}_\omega(\epsilon_n)$ denote the class of all densities such that $\sup_{p \in \mathcal{P}_\omega(\epsilon_n)} \| \omega_p - \widehat\omega\| \ \leq \epsilon_n$ and $\E \{\mu(X) \mid A = 1, \omega(X) = t_1, \widehat\omega(X) = t_2, D^n\}$ is arbitrarily smooth. Then,
			\begin{align*}
				\inf_{T_n} \sup_{p \in \mathcal{P}_\omega(\epsilon_n)} \E|T_n - \psi_p| \ \gtrsim \ \epsilon^2_n.
			\end{align*}
			\item  Let $\mathcal{P}_\mu(\delta_n)$ denote the class of all densities such that $\sup_{p \in \mathcal{P}_\mu(\delta_n)} \| \mu_p - \widehat\mu\| \ \leq \delta_n$ and $\E \{\omega(X) \mid A = 1, \widehat\mu(X) = t_1, \mu(X) = t_2, D^n\}$ is arbitrarily smooth. Then,
			\begin{align*}
				\inf_{T_n} \sup_{p \in \mathcal{P}_\mu(\delta_n)} \E|T_n - \psi_p| \ \gtrsim \ \delta^2_n.
			\end{align*}
			\item Let $\mathcal{P}_{\omega\mu}(\epsilon_n, \delta_n)$ be the class of all densities such that 1) $\sup_{p \in \mathcal{P}_{\omega\mu}} \|\omega_p - \widehat\omega\| \leq \epsilon_n$ and $\sup_{p \in \mathcal{P}_{\omega\mu}} \|\mu_p - \widehat\mu\| \leq \delta_n$, 2) $\E \{\mu(X) \mid A = 1, \widehat\omega(X) = t_1, \omega(X) = t_2, \widehat\mu(X) = t_3; D^n\}$ and $\E \{\omega(X) \mid A = 1, \widehat\mu(X) = t_1, \mu(X) = t_2, \widehat\omega(X) = t_3; D^n\}$ are arbitrarily smooth. Then, 
			\begin{align*}
				\inf_{T_n} \sup_{p \in \mathcal{P}_{\omega\mu}(\epsilon_n, \delta_n)} \E|T_n - \psi_p| \ \gtrsim \ \delta^2_n \land \epsilon_n^2.
			\end{align*}
		\end{enumerate}
	\end{proposition}
	As exemplified in the third claim, the smaller class $\mathcal{P}_{\omega\mu}(\epsilon_n, \delta_n)$ is an example of a collection of densities for which our lower bound allows for the possibility of doubly-robust, root-$n$ inference. In fact, $\delta_n^2 \land \epsilon_n^2 = o(n^{-1/2})$ if either $\|\widehat\omega - \omega\| = o(n^{-1/4})$ or $\|\widehat\mu - \mu\|  = o(n^{-1/4})$. In Section \ref{sec:main_estimator}, we construct an estimator achieving this rate, under certain conditions. This estimator can then be used for conducting doubly-robust inference via a standard Wald confidence interval. 
	
	\begin{remark}
		All the lower bounds from Propositions \ref{prop:pure_sa} and \ref{siva_prop} can be strengthened by taking the maximum between the rates shown and the parametric rate $n^{-1/2}$. Even if the nuisance functions $\omega(x)$ and $\mu(x)$ were known exactly, one would not typically be able to estimate $\psi$ are a rate faster than $n^{-1/2}$. The resulting lower bounds can be derived by a standard argument; see, for example, Section B.3 (Case 1) in \cite{balakrishnan2023fundamental}. 
	\end{remark}
	\section{New estimators of $\psi$ \label{sec:upper_bounds}}
	\subsection{Preliminaries and overview}
	In this section, we consider three new estimators of $\psi$ and derive upper bounds on their risk. Each estimator will be written as the doubly-robust estimator $\widehat\psi_\DR$ \eqref{eq:aipw} minus a term $T_n$ taking a different form depending on the assumptions invoked. That is,
	\begin{align*}
		\widehat\psi = \widehat\psi_{\DR} - T_n = \Pn \widehat\varphi - T_n, \text{ where } 
		\varphi(O) = A\omega(X)\{Y - \mu(X)\} + \mu(X).
	\end{align*}
	We view $T_n$ as an estimator of $R_n$ (from \eqref{eq:Rn}) that stems from merging ideas from the theory of HOIFs and some observations previously made in the doubly-robust inference literature (see Sections \ref{sec:connection_to_HOIFs} and \ref{sec:connection_to_DRAL}). 
	
	The bias $R_n$ can be written in different ways other than as in equation \eqref{eq:Rn}. Similar to \cite{van2014targeted} and \cite{benkeser2017doubly}, we will consider additional nuisance functions taking the form of regressions with both estimated outcomes and covariates. Define
	\begin{align*}
		& s_\omega(t_1, t_2; D^n) =  \E \{\mu(X) - \widehat\mu(X) \mid A = 1, \widehat\omega(X) = t_1, \omega(X) = t_2, D^n\}, \\
		& s_\mu(t_1, t_2; D^n) = \E\{\widehat\omega(X) - \omega(X) \mid A = 1, \widehat\mu(X) = t_1, \mu(X) = t_2, D^n\}, \\
		& f_\omega(t_1, t_2, t_3; D^n) = \E \{\mu(X) - \widehat\mu(X) \mid A = 1, \widehat\omega(X) = t_1, \omega(X) = t_2, \widehat\mu(X) = t_3, D^n\}, \\
		& f_\mu(t_1, t_2, t_3; D^n) = \E\{\widehat\omega(X) - \omega(X) \mid A = 1, \widehat\mu(X) = t_1, \mu(X) = t_2, \widehat\omega(X) = t_3, D^n\}.
	\end{align*}
	Notice that
	\begin{align*}
		& s_\omega(t_1, t_2; D^n) = \E\{Y - \widehat\mu(X) \mid A = 1, \widehat\omega(X) = t_1, \omega(X) = t_2, D^n\}, \\
		& f_\omega(t_1, t_2, t_3; D^n) = \E\{Y - \widehat\mu(X) \mid A = 1, \widehat\omega(X) = t_1, \omega(X) = t_2, \widehat\mu(X) = t_3, D^n\},
	\end{align*}
	but $s_\mu$ and $f_\mu$ cannot directly be written as regressions of observed outcomes on partially observed covariates. 
	
	By the law of total expectation, $R_n$ can be written as an expectation of an observed random variable times one of the four functions above:
	\begin{align*}
		R_n & =  \E\{(A\widehat\omega - 1)s_\omega(\widehat\omega, \omega; D^n) \mid D^n\} = \E\{(A\widehat\omega - 1)f_\omega(\widehat\omega, \omega, \widehat\mu; D^n) \mid D^n\} \\
		& = \E\{A(Y - \widehat\mu)s_\mu(\widehat\mu, \mu; D^n)\mid D^n\} = \E\{A(Y - \widehat\mu)f_\mu(\widehat\mu, \mu, \widehat\omega; D^n)\mid D^n\},
	\end{align*}
	where, for shorthand notation, $\widehat\omega = \widehat\omega(X)$ and $\widehat\mu = \widehat\mu(X)$. If either $s_\omega$, $s_\mu$, $f_\omega$ or $f_\mu$ were known, then $R_n$ could be estimated efficiently by a sample average. In Section \ref{sec:max_estimators}, we derive estimators tailored to models where $s_\omega$ or $s_\mu$ are H\"{o}lder functions. These estimators will improve upon $\widehat\psi_\DR$, and match the lower bound rates from Claims 1 and 2 in Proposition \ref{siva_prop}, when it is known whether $s_\omega$ or $s_\mu$ is smoother and thus easier to estimate. Without such knowledge, they would perform, in terms of mean-square-error, as well as $\widehat\psi_\DR$ but not necessarily better. To remedy this, in Section \ref{sec:main_estimator}, we construct an estimator that, under certain conditions, can improve upon $\widehat\psi_\DR$ without the knowledge of which nuisance function is easier to estimate. Further, it is shown to achieve the lower bound rate from Claim 3 in Proposition \ref{siva_prop}. 
	
	All these nuisance functions depend on estimated outcomes and covariates. The problem of nonparametrically estimating a regression function when some covariates need to be estimated in a first step has been considered, for example, by \cite{mammen2012nonparametric}, \cite{sperlich2009note} and \cite{dukes2021doubly}. In non-randomized experiments, regression adjustments based on the propensity score, e.g. via matching \parencite{abadie2016matching} or ordinary least squares \parencite{robins1992estimating, vansteelandt2014structural}, represent instances of such a general estimation problem. To understand one of the main challenges in this problem, due to its intrinsic non-smoothness, consider estimating $\E(Y \mid f(X) = t)$, for some function $f(X)$, estimated by $\widehat{f}(X)$, and observable random variables $Y$ and $X$. Regressing $Y$ onto $\widehat{f}(X)$ via some local method that depends on a vanishing bandwidth $h$ will need to ensure that $h$ does not shrink faster that the error $\widehat{f} - f$ or else the localization would be ``misplaced." This, in turns, leads to difficulties in choosing the correct order of $h$ as well as to potentially a dramatically slow rate of convergence since one may expect the error $\widehat{f} - f$ to be inflated by multiplication by $h^{-1}$ (see, e.g., Theorem 1 in \cite{mammen2012nonparametric} and Theorem 3 in \cite{dukes2021doubly}).
	  
    For this reason, we impose smoothness conditions directly on $s_\omega(t_1, t_2; D^n)$ rather than bounding the error in estimating $\E \{\mu(X) - \widehat\mu(X) \mid A = 1, \omega(X), D^n\}$ \parencite{benkeser2017doubly, dukes2021doubly}. A similar condition appears in \cite{van2024automatic}. 

    We conclude this section by providing some intuition for when one may expect that $s_\omega$, $s_\mu$, $f_\omega$ and $f_\mu$ possess some smoothness. However, a formal investigation that takes into account specific estimators of $\omega(X)$ and $\mu(X)$ is left for future work. 
	\begin{remark}
		Our proofs of Propositions \ref{prop:max_squared_a}, \ref{prop:max_squared_b} and \ref{prop:min_squared} below require control only for $t_1$ and $t_2$ of the form $t_1 = \widehat\omega(x_0)$ and $t_2 = \omega(x_0)$ for some arbitrary $x_0$ in the support of $X$. In this light, in this remark, we take $s_\omega(\widehat\omega(x_0), \omega(x_0); D^n)$ as the example and consider two cases where one may expect this function to be smooth. We leave the conditioning on $A = 1$ and $D^n$ implicit. 
		
		Perhaps the easiest case to consider is when $s_\omega(\widehat\omega(X), \omega(X); D^n) = \E \{\mu(X) - \widehat\mu(X) \mid \omega(X), D^n\}$, i.e., $\mu(X) - \widehat\mu(X) = m(\omega(X)) + \epsilon$, where $\epsilon$ is mean-zero given both $\omega(X)$ and $\widehat\omega(X)$. If this is the case, then one needs to assume that $m(\cdot)$ possesses some smoothness, which is arguably a more standard requirement as it does not involve a generated regressor.
		
		Next, we consider the case where, given $D^n$, $(\mu(X) - \widehat\mu(X), \omega(X), \widehat\omega(X))$ is jointly Gaussian\footnote{Strictly speaking, our boundedness assumption on $\mu$, $\widehat\mu$, $\omega$, and $\widehat\omega$ excludes the possibility that the vector is jointly Gaussian. The elementary calculations reported here are simply meant to capture some idealized sufficient conditions that might justify the assumption that the extra nuisance functions are smooth.}:
		\begin{align*}
			\begin{bmatrix}
				\mu(X) - \widehat\mu(X) \\
				\omega(X) \\
				\widehat\omega(X) 
			\end{bmatrix} \sim N\left(\begin{bmatrix}
				\mu_\mu \\
				\mu_\omega \\
				\mu_{\widehat\omega} 
			\end{bmatrix}, 
			\begin{bmatrix}
				\sigma^2_{\mu - \widehat\mu} & \sigma_{\omega(\mu - \widehat\mu)} & \sigma_{\widehat\omega(\mu - \widehat\mu)} \\
				\sigma_{\omega(\mu - \widehat\mu)} & \sigma^2_\omega & \rho \sigma_{\omega}\sigma_{\widehat\omega} \\
				\sigma_{\widehat\omega(\mu - \widehat\mu)} & \rho \sigma_{\omega}\sigma_{\widehat\omega} & \sigma^2_{\widehat\omega}
			\end{bmatrix}
			\right),
		\end{align*}
        where $ \sigma_{vw} = \cov\{v(X), w(X)\}$ and $\rho = \text{cor}\{\widehat\omega(X), \omega(X)\}$. 
		We have that
		\begin{align*}
			& s_\omega(\widehat\omega(x_0), \omega(x_0); D^n) - s_\omega(\widehat\omega(x_1), \omega(x_1); D^n) \\
			& = \left\{\frac{\sigma_{\omega(\mu - \widehat\mu)} \sigma^2_{\widehat\omega} - \rho\sigma_\omega\sigma_{\widehat\omega} \sigma_{\widehat\omega(\mu - \widehat\mu)}}{\sigma^2_\omega \sigma^2_{\widehat\omega}(1 - \rho^2)} \right\}\{\omega(x_0) - \omega(x_1)\} + \left\{\frac{\sigma_{\widehat\omega(\mu - \widehat\mu)} \sigma^2_\omega - \sigma_{\omega(\mu - \widehat\mu)}\rho\sigma_\omega\sigma_{\widehat\omega}}{\sigma^2_\omega \sigma^2_{\widehat\omega}(1-\rho^2)} \right\}\{\widehat\omega(x_0) - \widehat\omega(x_1)\} \\
			& \equiv \Lambda_1 \{\omega(x_0) - \omega(x_1)\} + \Lambda_2\{\widehat\omega(x_0) - \widehat\omega(x_1)\}.
		\end{align*}
		Therefore, $s(\widehat\omega(x_0), \omega(x_0); D^n)$ would be Lipschitz if $|\Lambda_1|$ and $|\Lambda_2|$ are bounded. We can write
		\begin{align*}
			& \Lambda_1 = \frac{\sigma_{\widehat\omega(\mu - \widehat\mu)}}{\sigma_{\widehat\omega}\sigma_\omega} \cdot \frac{1}{1 + \rho} + \frac{\sigma_{\widehat\omega(\mu - \widehat\mu)}}{\sigma_{\widehat\omega}\sigma^2_\omega} \cdot \frac{\sigma_{\widehat\omega} - \sigma_\omega}{1 - \rho^2} + \frac{1}{\sigma_\omega\sigma_{\widehat\omega}} \cdot \frac{\sigma_{(\omega - \widehat\omega)(\mu - \widehat\mu)}}{1 - \rho^2}, \\
			& \Lambda_2 = \frac{\sigma_{\omega(\mu - \widehat\mu)}}{\sigma_{\widehat\omega}\sigma_\omega} \cdot \frac{1}{1 + \rho} + \frac{\sigma_{\widehat\omega(\mu - \widehat\mu)}}{\sigma^2_{\widehat\omega}\sigma_\omega} \cdot \frac{\sigma_\omega - \sigma_{\widehat\omega}}{1 - \rho^2} + \frac{1}{\sigma_{\widehat\omega}\sigma_\omega} \cdot \frac{\sigma_{(\widehat\omega - \omega)(\mu - \widehat\mu)}}{1 - \rho^2}.
		\end{align*}
		So, if $\sigma_\omega \gtrsim 1$, $\sigma_{\widehat\omega} \gtrsim 1$, $|\sigma_{\widehat\omega} - \sigma_\omega| \ \lesssim 1 - \rho^2$ and $|\sigma_{(\widehat\omega - \omega)(\mu - \widehat\mu)}| \ \lesssim 1 - \rho^2$, then $|\Lambda_1|$ and $|\Lambda_2|$ are bounded and so $s_\omega(\widehat\omega(x_0), \omega(x_0); D^n)$ is Lipschitz. These straightforward calculations provide an example of more primitive conditions, albeit under idealized conditions, under which $s_\omega(t_1, t_2; D^n)$ may be expected to possess some smoothness.
	\end{remark}
	\subsection{Estimators exploiting the smoothness of either $s_\omega$ or $s_\mu$ but not both.\label{sec:max_estimators}}
	In this section, we present two estimators, $\widehat\psi_\omega$ and $\widehat\psi_\mu$, that are tailored to models where $s_\omega$ and $s_\mu$ have some smoothness, respectively. We will focus on $\widehat\psi_\omega$, but the same reasoning essentially applies to $\widehat\psi_\mu$ as well. Recall the notation $\widehat\omega_i = \widehat\omega(X_i)$ and $\widehat\omega=\widehat\omega(X)$. Consider the estimator $\widehat\psi_\omega = \widehat\psi_\DR - T_{n\omega}$ for
	\begin{align*}
		T_{n\omega} = \frac{1}{n(n-1)} \mathop{\sum\sum}_{1 \leq i \neq j \leq n} (A_i \widehat\omega_i - 1) \frac{K_h(\widehat\omega_j - \widehat\omega_i)}{\widehat{Q}(\widehat\omega_i)} A_j(Y_j - \widehat\mu_j), \text{ where }
	\end{align*}
	$K_h(u) = h^{-1}K(u/h)$, for some vanishing bandwidth $h$, and $\widehat{Q}(\widehat\omega_i) = (n-1)^{-1}\sum_{j \neq i} A_j K_h(\widehat\omega_j - \widehat\omega_i)$. Throughout, we assume the following:
	
	\textbf{Assumption (Kernel).} The kernel $K(u)$ is a non-negative, bounded, symmetric function (around zero) that is supported on $[-1, 1]$. One example is $K(u) = 0.5 \one(|u| \ \leq 1)$.
	
	If $\E\{A_jK_h(\widehat\omega_j - \widehat\omega_i) \mid X_i, D^n)$ and $\widehat{Q}^{-1}(\widehat\omega_i)$ are bounded, then by the Cauchy-Schwarz inequality:
	\begin{align*}
		| \E(T_{n\omega} \mid D^n) | \ \lesssim \|\widehat\omega - \omega\| \|\widehat\mu - \mu\|.
	\end{align*}
	In this light, this term has expectation of the same order as $R_n$, the conditional bias of $\widehat\psi_\DR$, and thus one may hope to not degrade the performance of $\widehat\psi_\DR$, at least asymptotically, by including $T_{n\omega}$ in the final estimator. A formal calculation, however, would need to consider the variance of $T_{n\omega}$ as well. Next, we give a high level argument about why and when subtracting off $T_{n\omega}$ from $\widehat\psi_\DR$ may lead to a better estimator. Our argument is based on decomposing $T_{n\omega}$ as
	\begin{align*}
	&	\frac{1}{n(n-1)} \mathop{\sum\sum}_{1 \leq i \neq j \leq n} \{A_i \widehat\omega_i - 1\} \frac{K_h(\widehat\omega_j - \widehat\omega_i)}{\widehat{Q}(\widehat\omega_i)} A_j s_\omega(\widehat\omega_i, \omega_i; D^n) \\
        & + \frac{1}{n(n-1)} \mathop{\sum\sum}_{1 \leq i \neq j \leq n}(A_i \widehat\omega_i - 1)\frac{K_h(\widehat\omega_j - \widehat\omega_i)}{\widehat{Q}(\widehat\omega_i)} A_j \epsilon_j \\
        & + \frac{1}{n(n-1)} \mathop{\sum\sum}_{1 \leq i \neq j \leq n}(A_i \widehat\omega_i - 1)\frac{K_h(\widehat\omega_j - \widehat\omega_i)}{\widehat{Q}(\widehat\omega_i)}A_j\{s_\omega(\widehat\omega_j, \omega_j; D^n) - s_\omega(\widehat\omega_i, \omega_i; D^n)\} \\
        & = (i) + (ii) + (iii)
	\end{align*}
	Firstly, $(i)$ can be shown to have, conditional on the training sample $D^n$, mean exactly equal to $R_n$. Moreover, $(ii)$ can be shown to have mean zero by the independence of $O_i$ and $O_j$ for $i \neq j$. This is because, by definition:
	\begin{align*}
		\E(A_j \epsilon_j \mid \widehat\omega_j, \omega_j, D^n) & = \E[A_j\{Y_j - \widehat\mu_j - s_\omega(\widehat\omega_j, \omega_j; D^n)\} \mid \widehat\omega_j, \omega_j, A_j, D^n]  = 0.
	\end{align*}
	This implies that $\E(A_j\epsilon_j \mid \widehat\omega_j, A_j, D^n) = 0$ and, by independence of $O_i$ and $O_j$ for $i \neq j$, also that $\E(A_j \epsilon_j \mid \widehat\omega_j, \widehat\omega_i, A_i, A_j, D^n) = 0$. 
	Finally, if $s_\omega(t_1, t_2; D^n)$ is H\"{o}lder of order $\alpha \in [0, 1]$, then 
	\begin{align*}
		\left| \E\{(iii) \mid D^n\} \right| \lesssim \left(\|\widehat\omega - \omega\|^{1 +\alpha} \ + \ \|\widehat\omega - \omega\|h^\alpha\right) \land \|\widehat\omega - \omega\| \|\widehat\mu - \mu\|,
	\end{align*}
    provided that $\widehat{Q}(\widehat\omega_1)$ is bounded away from zero and $\E\{A_2K_h(\widehat\omega_2 - \widehat\omega_1) \mid X_1, D^n\}$ is bounded.

	For $\alpha = 1$, i.e., $s_\omega(t_1, t_2; D^n)$ is Lipschitz, choosing $h \asymp n^{-1/2}$ and combining the above results yields
	\begin{align*}
		\left| \E(\widehat\psi_\omega \mid D^n) - \psi \right| \lesssim \|\widehat\omega - \omega\|^2 \ \land \ \|\widehat\omega - \omega\| \|\widehat\mu - \mu\|.
	\end{align*}
	The bound on the right-hand-side of the display above improves upon the conditional bias of $\widehat\psi_\DR$ in the case when $\|\widehat\omega - \omega\|$ is of smaller order than $\|\widehat\mu - \mu\|$. This improvement is driven by the smoothness conditions explicitly on $s_\omega(t_1, t_2; D^n)$.
	
	Next, we look for an estimator that achieves a bound on the bias of order $\|\widehat\mu - \mu\|^2 \ \land \ \|\widehat\omega - \omega\| \|\widehat\mu - \mu\|$. A natural candidate is $\widehat\psi_\mu = \widehat\psi_\DR - T_{n\mu}$, where:
	\begin{align*}
		& T_{n\mu} = \frac{1}{n(n-1)} \mathop{\sum\sum}_{1 \leq i \neq j \leq n} (A_i \widehat\omega_i - 1) \frac{K_h(\widehat\mu_j - \widehat\mu_i)}{\widehat{Q}_{-i}(\widehat\mu_j)} A_j(Y_j - \widehat\mu_j), \\
		& \widehat{Q}_{-i}(\widehat\mu_j) = (n-1)^{-1} \sum_{s = 1, s \neq i}^n A_sK_h(\widehat\mu_s - \widehat\mu_j).
	\end{align*}
	We remove the $i^{\text{th}}$ observation in $\widehat{Q}_{-i}(\widehat\mu_j)$ because we need to ensure that
	\begin{align*}
		\E\left\{ (A_i \widehat\omega_i - 1) \frac{K_h(\widehat\mu_j - \widehat\mu_i)}{\widehat{Q}_{-i}(\widehat\mu_j)} A_j(Y_j - \widehat\mu_j) \mid D^n \right\} = \E\left\{ A_i (\widehat\omega_i - \omega_i) \frac{K_h(\widehat\mu_j - \widehat\mu_i)}{\widehat{Q}_{-i}(\widehat\mu_j)} A_j(Y_j - \widehat\mu_j) \mid D^n \right\}.
	\end{align*}
	This would not be the case if the residual $A_i\widehat\omega_i - 1$ gets multiplied by a term $\widehat{Q}(\widehat\mu_j)$ involving $A_i$. With this modification in place, calculations analogous to the ones above for $\widehat\psi_\omega$ yield that
	\begin{align*}
		\left| \E(T_{n\mu} \mid D^n) - R_n \right| \lesssim \left(\|\widehat\mu - \mu\|^{1 + \beta} \ + \ h^\beta \|\widehat\mu -\mu\| \ + \ \frac{\|\widehat\omega - \omega\|\|\widehat\mu - \mu\|}{\sqrt{nh}} \right) \ \land \ \|\widehat\omega - \omega\| \|\widehat\mu - \mu\|.
	\end{align*}
	when $s_\mu(t_1, t_2; D^n)$ is H\"{o}lder $\beta \in [0, 1]$. When $\beta = 1$ and $h \asymp n^{-1/2}$, the bound reduces to 
	\begin{align*}
		\left| \E(\widehat\psi_\mu \mid D^n) - \psi \right| \lesssim (n^{-1/2} + \|\widehat\mu - \mu\|^{2}) \ \land \ \|\widehat\omega - \omega\| \|\widehat\mu - \mu\|.
	\end{align*}
	Therefore, $\widehat\psi_\mu$ would have smaller conditional bias than $\widehat\psi_\DR$ if $\mu(x)$ is easier to estimate than $\omega(x)$ so that $\|\widehat\mu - \mu\|^2$ is the dominant term. In the following propositions, we collect more formal statements on the behavior of $\widehat\psi_\omega$ and $\widehat\psi_\mu$. To do so, let us define the following functions, which arise from the linear terms (and their limits) of Hoeffding's decomposition of the U-statistics $T_{n\omega}$ and $T_{n\mu}$:
	\begin{align*}
		& \widehat\varphi_\omega = \E(\mu - \widehat\mu \mid A = 1, \widehat\omega, D^n)(A\widehat\omega - 1), \quad \overline\varphi_\omega = \E(\mu - \overline\mu \mid A = 1, \overline\omega)(A\overline\omega - 1), \\
		& \widehat\varphi_\mu = \E(\widehat\omega - \omega \mid A = 1, \widehat\mu, D^n)A(Y - \widehat\mu), \quad \overline\varphi_\mu = \E(\overline\omega - \omega \mid A = 1, \overline\mu)A(Y - \overline\mu).
	\end{align*}
	\begin{proposition} \label{prop:max_squared_a}
		Suppose that $\E\{A_jK_h(\widehat\omega_j - \widehat\omega_i) \mid X_i, D^n\} \lesssim 1$ and that $\widehat{Q}(\widehat\omega_i)$ is bounded away from zero for $1 \leq i \neq j \leq n$. If $s_\omega(t_1,t_2; D^n)$ is H\"{o}lder of order $\alpha \in [0, 1]$, then:
		\begin{enumerate}
			\item Bias of $\widehat\psi_\omega$:
			\begin{align*}
				& \left| \E(\widehat\psi_\omega - \psi \mid D^n) \right| \lesssim \left(h^\alpha \|\widehat\omega - \omega\| \ + \ \|\widehat\omega - \omega \|^{1 + \alpha}\right) \land \|\widehat\omega - \omega\|\|\widehat\mu - \mu\|,
			\end{align*}
			\item Variance of $\widehat\psi_\omega$:
			\begin{align*}
				\var(\widehat\psi_\omega \mid D^n) \lesssim n^{-1} + (n^3h)^{-1/2} + \|\widehat\omega - \omega\|^2\|\widehat\mu - \mu\|^2(nh)^{-1},
			\end{align*}
			\item Linear expansion of $\widehat\psi_\omega - \psi$:
			\begin{align*}
				& \widehat\psi_\omega - \psi = (\Pn - \Pb) (\overline\varphi - \overline\varphi_\omega) + O_\Pb\left(n^{-1/2} \|\widehat\omega - \omega\|_\infty \ \land \ (nh)^{-1/2} \|\widehat\omega - \omega\|\right)\\
				& \hphantom{=} +  O_\Pb\left(n^{-1/2}\{ (\|\widehat\omega - \omega\|_\infty^\alpha \ + \ h^\alpha) \land \|\widehat\mu - \mu\|_\infty\} \right) + o_P(n^{-1/2}) \\
				& \hphantom{=} + O_\Pb\left(\left\{\left(h^\alpha \|\widehat\omega - \omega\| \ + \ \|\widehat\omega - \omega \|^{1 + \alpha}\right) \land \|\widehat\omega - \omega\|\|\widehat\mu - \mu\|\right\} + (n^3h)^{-\frac{1}{4}} + \|\widehat\omega - \omega\|\|\widehat\mu - \mu\|(nh)^{-\frac{1}{2}}\right).
			\end{align*}
			provided that $\|\widehat\varphi_\omega - \overline\varphi_\omega\| = o_\Pb(1)$.
		\end{enumerate}
	\end{proposition}
	\begin{proposition}\label{prop:max_squared_b}
		Suppose that $\E\{A_jK_h(\widehat\mu_j - \widehat\mu_i) \mid X_i, D^n\} \lesssim 1$ and that $\widehat{Q}_{-i}(\widehat\mu_j)$ is bounded away from zero for $1 \leq i \neq j \leq n$. If $s_\mu(t_1,t_2; D^n)$ is H\"{o}lder of order $\beta \in [0, 1]$, then:
		\begin{enumerate}
			\item Bias of $\widehat\psi_\mu$:
			\begin{align*}
				\left| \E(\widehat\psi_\mu - \psi \mid D^n) \right| \lesssim \left(h^\beta \|\widehat\mu - \mu\| \ + \ \|\widehat\mu - \mu \|^{1 + \beta} \ + \ \frac{\|\widehat\omega - \omega\|\|\widehat\mu - \mu\|}{\sqrt{nh}} \right) \land \|\widehat\omega - \omega\|\|\widehat\mu - \mu\|,
			\end{align*}
			\item Variance of $\widehat\psi_\mu$:
			\begin{align*}
				\var(\widehat\psi_\mu \mid D^n) \lesssim n^{-1} + (n^3h)^{-1/2} + \|\widehat\omega - \omega\|^2\|\widehat\mu - \mu\|^2(nh)^{-1},
			\end{align*}
			\item Linear expansion of $\widehat\psi_\mu - \psi$:
			\begin{align*}
				& \widehat\psi_\mu - \psi = (\Pn - \Pb) (\overline\varphi - \overline\varphi_\mu) + O_\Pb\left(n^{-1/2} \|\widehat\mu - \mu\|_\infty \ \land \ (nh)^{-1/2} \|\widehat\mu - \mu\|\right) \\
				& \hphantom{=} +  O_\Pb\left(n^{-1/2} \{(\|\widehat\mu - \mu\|_\infty^\beta \ + \ h^\beta) \land \|\widehat\omega - \omega\|_\infty\} \right) + o_P(n^{-1/2}) \\
				& \hphantom{=} + O_\Pb\left(\left\{\left(h^\beta \|\widehat\mu - \mu\| \ + \ \|\widehat\mu - \mu\|^{1 + \beta}\right) \land \|\widehat\omega - \omega\|\|\widehat\mu - \mu\|\right\} + (n^3h)^{-\frac{1}{4}} + \|\widehat\omega - \omega\|\|\widehat\mu - \mu\|(nh)^{-\frac{1}{2}}\right).
			\end{align*}
			provided that $\|\widehat\varphi_\mu - \overline\varphi_\mu\| = o_\Pb(1)$.
		\end{enumerate}
	\end{proposition}
	Propositions \ref{prop:max_squared_a} and \ref{prop:max_squared_b} yield that $\widehat\psi_\omega$ can improve upon the performance of $\widehat\psi_\DR$ in models where $\omega$ is easier to estimate than $\mu$, while $\widehat\psi_\mu$ can improve upon $\widehat\psi_\DR$ when $\mu$ is easier to estimate than $\omega$. The main requirement for the proposition to hold is that $s_\omega$ and $s_\mu$ possess some minimal smoothness encoded in the H\"{o}lder condition. It can be seen from a standard change of variables argument that the conditions $\E\{A_jK_h(\widehat\omega_j - \widehat\omega_i) \mid X_i, D^n\} \lesssim 1$ and $\E\{A_jK_h(\widehat\mu_j - \widehat\mu_i) \mid X_i, D^n\} \lesssim 1$ are satisfied if $\widehat\omega(X)$ and $\widehat\mu(X)$ have densities with respect to the Lebesgue measure, respectively. We expect the assumption that $\widehat{Q}(\widehat\omega_j)$ and $\widehat{Q}_{-i}(\widehat\mu_j)$ are bounded away from zero to hold in practice as long as the sample size is large enough and $\widehat\omega(X)$ and $\widehat\mu(X)$ are roughly evenly distributed on their support. 
	\begin{remark}
		It is reasonable to consider applications where $\alpha$ and $\beta$ could be greater than 1, i.e., $s_\omega$ and $s_\mu$ are potentially smoother than what considered in Propositions \ref{prop:max_squared_a} and \ref{prop:max_squared_b}. This higher-order smoothness could then be exploited using higher-order kernels or kernels of local polynomials, for example. However, our analysis (not reported in the interest of space) suggests that the bound on the bias would still contain terms of order $\|\widehat\omega - \omega\|^{1 + (\alpha \land 1)}$ and $\|\widehat\mu - \mu\|^{1 + (\beta \land 1)}$. This is consistent with our lower bound analysis (Proposition \ref{siva_prop}), which establishes that the rate of convergence for $\widehat\psi_\omega$ and $\widehat\psi_\mu$ in these models cannot be faster than $\|\widehat\omega - \omega\|^2$ and $\|\widehat\mu - \mu\|^2$, respectively. Attempting to track higher order smoothness, however, may have benefits in terms of more flexibility in choosing the bandwidth $h$, particularly in finite samples. We leave the study of higher-smoothness regimes for future work.
	\end{remark}
	\begin{remark}
		Suppose that $\alpha = 1$ and set $h \asymp n^{-1/2}$. If $\|\widehat\omega - \omega\|_\infty = o_\Pb(1)$, $\|\widehat\omega - \omega\| = o_\Pb(n^{-1/4})$, then $|\widehat\psi_\omega - \psi| = O_\Pb(n^{-1/2})$ and $\sqrt{n}(\widehat\psi_\omega - \psi) \indist N(0, \var(\overline\varphi - \overline\varphi_\omega))$\footnote{An earlier version of this manuscript incorrectly stated that the limiting variance is $\var(\overline\varphi)$ while in fact our analysis yields that the variance is $\var(\overline\varphi - \overline\varphi_\omega)$. Similar corrections have been made to Propositions \ref{prop:max_squared_b} and \ref{prop:min_squared}.}. This means that $\widehat\psi_\omega$ can be used for constructing a Wald-type confidence interval even if $\widehat\mu$ is not consistent. However, if both $\|\widehat\omega - \omega\|_\infty = o_\Pb(1)$ and $\|\widehat\mu - \mu\|_\infty = o_\Pb(1)$, then $\widehat\psi_\omega$ will also achieve the semiparametric efficiency bound $\var(\varphi)$ because, in this case, $\overline\varphi = \varphi$, since $\overline\omega = \omega$, $\overline\mu = \mu$, and $\overline\varphi_\omega = 0$. Similar considerations apply to $\widehat\psi_\mu$ with the roles of $\widehat\omega$ and $\widehat\mu$ reversed.  
	\end{remark}
	\subsection{Main estimator \label{sec:main_estimator}}
	To be useful in practice, the estimators presented in Section \ref{sec:max_estimators} require knowledge of whether $\omega$ or $\mu$ is easier to estimate. In this section, we consider the estimator $\widehat\psi = \widehat\psi_\DR - T_n$, for
	\begin{align*}
		T_n = \frac{1}{n(n-1)} \mathop{\sum\sum}_{1 \leq i \neq j \leq n} (A_i \widehat\omega_i - 1) \frac{K_h(\widehat\mu_j - \widehat\mu_i) K_h(\widehat\omega_j - \widehat\omega_i)}{\widehat{Q}(\widehat\omega_i, \widehat\mu_i)} A_j(Y_j - \widehat\mu_j)
	\end{align*}
	where $\widehat{Q}(\widehat\omega_i, \widehat\mu_i) = (n-1)^{-1} \sum_{i \neq j} A_jK_h(\widehat\omega_j - \widehat\omega_i) K_h(\widehat\mu_j - \widehat\mu_i)$. We will show that, under certain conditions, the risk of this estimator is of the same order as the lower bound on the risk of any estimator derived in Proposition \ref{siva_prop} (Claim 3), thereby establishing sufficient conditions under which this estimator is minimax optimal. 
	
	A key difference between $\widehat\psi$ and either $\widehat\psi_\omega$ or $\widehat\psi_\mu$ is that $\widehat\psi$ is tailored to models where both $f_\omega(t_1, t_2, t_3; D^n)$ and $f_\mu(t_1, t_2, t_3; D^n)$ have some smoothness. The central term multiplying the two residuals $A_i\widehat\omega_i - 1$ and $A_j(Y_j - \widehat\mu_j)$ is meant to act as a kernel of a local regression on $(\widehat\omega, \widehat\mu)$ rather than on $\widehat\omega$ or $\widehat\mu$ alone. The reason for this change is that $\widehat\psi$ is designed to correct the bias of $\widehat\psi_\DR$ by subtracting off an estimate of $R_n$ even when it is not known whether $\omega$ or $\mu$ is easier to estimate. In fact, one can see that estimating $f_\omega(\widehat\omega, \omega, \widehat\mu; D^n)$ and $f_\mu(\widehat\mu, \mu, \widehat\omega; D^n)$ can be carried out simply by modifying the outcome variable as both residuals would, in practice, only be regressed onto $(\widehat\omega, \widehat\mu)$. This then allows for the construction of estimates of $R_n$ that would be essentially the same whether one expresses $R_n$ in terms of $f_\omega$ or $f_\mu$. One caveat is that summing over $i$ in the expression for $T_n$ should return an estimate of $f_\mu(\widehat\mu_j, \mu_j, \widehat\omega_j; D^n)$, but this is not exactly the case because $\widehat{Q}(\widehat\omega_i, \widehat\mu_i)$ is localized at $(\widehat\omega_i, \widehat\mu_i)$ instead of at $(\widehat\omega_j, \widehat\mu_j)$. In Proposition \ref{prop:min_squared}, we deal with this issue by invoking a Lipschitz assumption on the density of $(\widehat\omega, \widehat\mu)$ among units with $A = 1$. Whether this condition (or similar ones) can be avoided remains an open question. The main advantage of using $\widehat\psi$ versus either $\widehat\psi_\omega$ or $\widehat\psi_\mu$ is that $\widehat\psi$ is able to correct for $R_n$ when $R_n$ is expressed as a function of $f_\omega$ or $f_\mu$ using the same term $T_n$. Relative to $\widehat\psi_\omega$ and $\widehat\psi_\mu$, the price is a moderate increase in the variance because of the product of two kernels, which is needed for estimating the regressions on a two-, as opposed to one-, dimensional domain. In the following proposition, we summarize our analysis of the error $\widehat\psi - \psi$. To describe the linear term when either $\omega$ or $\mu$ is not consistently estimated, let us define
	\begin{align*}
		& \widehat\varphi_{\omega\mu}(Z; D^n) = \E(\widehat\omega - \omega\mid A = 1, \widehat\omega, \widehat\mu, D^n)A(Y - \widehat\mu) + \E(\mu - \widehat\mu\mid A = 1, \widehat\omega, \widehat\mu, D^n)(A\widehat\omega - 1), \\
		& \overline\varphi_{\omega\mu}(Z) = \E(\overline\omega - \omega\mid A = 1, \overline\omega, \overline\mu)A(Y - \overline\mu) + \E(\mu - \overline\mu\mid A = 1, \overline\omega, \overline\mu)(A\overline\omega - 1).
	\end{align*}
	\begin{proposition} \label{prop:min_squared}
		Suppose that the distribution of $(\widehat\omega, \widehat\mu)$ among units with $A = 1$ has a density with respect to the Lebesgue measure that is Lipschitz. Suppose that $\widehat{Q}(\widehat\omega_i, \widehat\mu_i)$ is bounded away from zero for $1 \leq j \leq n$. 
		Further, suppose that $f_\omega(t_1, t_2, t_3; D^n)$ and $f_\mu(t_1, t_2, t_3; D^n)$ are H\"{o}lder smooth of orders $\alpha$ and $\beta$, respectively, with $\alpha, \beta \in[0, 1]$. Finally, assume that $nh^2 \to \infty$. Then, 
		\begin{enumerate}
			\item Bias of $\widehat\psi$:
			\begin{align*}
				&\left|\E(\widehat\psi - \psi \mid D^n) \right| \lesssim \left(\|\widehat\omega - \omega\| h^\alpha + \|\widehat\omega - \omega\|^{1 + \alpha} \right) \land \left(\|\widehat\mu - \mu\| h^\beta + \|\widehat\mu - \mu\|^{1 + \beta}  \ + \ \frac{\|\widehat\omega - \omega\| \|\widehat\mu - \mu\|}{\sqrt{nh^2}} \right) \\
				& \hphantom{\left|\E(\widehat\psi - \psi \mid D^n) \right| \lesssim } \quad \land \|\widehat\omega - \omega\| \|\widehat\mu - \mu\|, 
			\end{align*}
			\item Variance of $\widehat\psi$:
			\begin{align*}
				\var (\widehat\psi \mid D^n) \lesssim  n^{-1} +  (n^3h^2)^{-1/2} + \|\widehat\omega - \omega\|^2 \|\widehat\mu - \mu\|^2(nh^2)^{-1},
			\end{align*}
			\item Linear expansion of $\widehat\psi - \psi$:
			\begin{align*}
				\widehat\psi - \psi & = (\Pn - \Pb)(\overline\varphi - \overline\varphi_{\omega\mu}) \\
				&\hphantom{=} + O_\Pb\left((\|\widehat\mu-\mu\|^{1+\beta} \ + \ h^\beta \|\widehat\mu - \mu\|) \land (\|\widehat\omega-\omega\|^{1+\alpha} \ + \ h^\alpha \|\widehat\omega - \omega\|) \land \|\widehat\omega-\omega\|\|\widehat\mu - \mu\|  \right) \\
				& \hphantom{=} + O_\Pb\left( (n^3h^2)^{-\frac{1}{4}} + \|\widehat\omega - \omega\|\|\widehat\mu - \mu\| (nh^2)^{-\frac{1}{2}} \right) + o_P(n^{-1/2}) \\
				& \hphantom{=} + O_\Pb\left(n^{-\frac{1}{2}}\left[\{(h^\alpha + \|\widehat\omega - \omega\|_\infty) \land \|\widehat\mu - \mu\|_\infty\} + \{(h^\beta + \|\widehat\mu - \mu\|_\infty) \land \|\widehat\omega - \omega\|_\infty)\} + h\right]\right)
			\end{align*}
			provided that $\|\widehat\varphi_{\omega\mu} - \overline\varphi_{\omega\mu}\| = o_\Pb(1)$.
		\end{enumerate}
	\end{proposition}
	It can be seen from Proposition \ref{prop:min_squared} that $\widehat\psi$ has the potential to improve upon $\widehat\psi_\omega$ and $\widehat\psi_\mu$ in the sense that its bias is the minimum between their biases. This comes at the price of a smoothness assumption on both $f_\omega$ and $f_\mu$ as well as a Lipschitz condition on the density of $(\widehat\omega(X), \widehat\mu(X))$. Under these conditions, $\widehat\psi$ also improves upon $\widehat\psi_\DR$ and can deliver doubly-robust root-$n$ inference when either $\widehat\omega$ or $\widehat\mu$ converges at $n^{-1/4}$ rates, as long as $\alpha = \beta = 1$ and $h \asymp n^{-1/4}$. In practice, choosing $h$ can be nontrivial. In Section \ref{sec:sims}, we select the cross-validated bandwidth that an estimator of the regression function $f_\omega(t_1, t_2, t_3; D^n)$ would choose. This choice can be implemented using off-the-shelf routines but we do not claim any optimality for it. How to choose the bandwidth in practice remains largely an open question. We note that this difficulty in selecting the tuning parameter also arises in \cite{dukes2021doubly}.
	\begin{remark}\label{remark:main_rootn}
		Consider the case where $\alpha = \beta = 1$ and set $h \asymp n^{-1/4}$. Then, it holds that
		\begin{align*}
			\sqrt{n}(\widehat\psi - \psi) = \sqrt{n}(\Pn - \Pb)(\overline\varphi - \overline\varphi_{\omega\mu}) + O_\Pb(n^{1/2} \{\|\widehat\omega - \omega\|^2 \ \land \ \|\widehat\mu - \mu\|^2\}) + O_\Pb(n^{1/4} \|\widehat\omega - \omega\| \|\widehat\mu - \mu\|) + o_\Pb(1)
		\end{align*}
		provided that either $\|\widehat\omega - \omega\|_\infty = o_\Pb(1)$ or $\|\widehat\mu - \mu\|_\infty = o_\Pb(1)$. In this light, $\sqrt{n}(\widehat\psi - \psi) \indist N(0, \var(\overline\varphi - \overline\varphi_{\omega\mu}))$ as long as $\|\widehat\omega - \omega\| = o_\Pb(n^{-1/4})$ or $\|\widehat\mu - \mu\| = o_\Pb(n^{-1/4})$. Notice that as long as either $\overline\omega = \omega$ or $\overline\mu = \mu$, we have
		\begin{align*}
			\overline\varphi - \overline\varphi_{\omega\mu} = A\E(\omega \mid A = 1, \overline\omega, \overline\mu)\{Y - \E(\mu \mid A = 1, \overline\omega, \overline\mu)\} + \E(\mu \mid A = 1, \overline\omega, \overline\mu)\}.
		\end{align*}
		From this, one can deduce that $\widehat\varphi$ is semiparametric efficient if, in addition to the conditions above, both $\overline\omega = \omega$ and $\overline\mu = \mu$.
	\end{remark}
	
	\begin{remark}
		We believe our construction of $\widehat\psi$ sheds light on a point raised by \cite{benkeser2017doubly} about the putative superiority of TMLE over one-step corrections when it comes to performing doubly-robust inference (Section 4 in \cite{benkeser2017doubly}). The authors note the difficulty in deriving doubly-robust, asymptotic linear estimators that take the form of $\widehat\psi_\DR$ plus a correction term. This apparent difference in performance is surprising, as one-step estimators and TMLEs are grounded in the same semiparametric efficiency theory and typically, at least asymptotically and in regimes where $\psi$ admits root-$n$-consistent estimators, share largely the same properties. Proposition \ref{prop:min_squared} shows that, under certain conditions, there exists in fact an estimator that can estimate $R_n$ ``in one step" and could thus be used for doubly-robust inference.
	\end{remark}
	Finally, in Appendix \ref{appendix:partially_linear_logit}, we briefly discuss how our approach can be applied to the partially linear logistic model. This is a type of parameter which the approach from \cite{dukes2021doubly} does not readily extend to. Instead, we show how a suitable modification of $\widehat\psi$ can deliver doubly-robust inference in this setting as well. 
	
	\subsubsection{Regularity}
	From \cite{van2000asymptotic} (page 365), an estimation sequence $\widehat\psi$ is regular at $P$ for estimating $\psi(P)$ (relative to the tangent set, i.e., the set of all scores) if, for every possible score, there exists a distribution $L$ such that 1) $L$ does not depend on the score and 2) it holds $n^{1/2}(\widehat\psi - \psi_n) \stackrel{p_n}{\indist} L$, where $p_n$ denote the distribution indexed by local parameters. In this section, we investigate whether the estimator $\widehat\psi$ is irregular. To do so, we consider the case where $Y$ is binary and follow the analysis described in Section 4.2 in \cite{dukes2021doubly}. See also the recent pre-print \cite{van2024automatic} for an analysis on the regularity properties of estimators when certain nuisance functions are inconsistently estimated.
 
    Let $p_\eta(O)$ be the density of one observation, indexed by the nuisance parameters $\eta = \{\omega(x), \mu(x), f(x)\}$, for $f$ the density of $X$:
	\begin{align*}
		p_\eta(O) = \frac{f}{\omega} (\omega - 1)^{1-A} \mu^{YA} (1-\mu)^{(1-Y)A}.
	\end{align*}
	We follow \cite{robins2009quadratic} and, given directions $d_\omega$, $d_\mu$ and $d_f$, we consider parametric submodels $\omega_t = \omega + td_\omega$, $\mu_t = \mu + td_\mu$ and $f_t = f + tfd_f$. For $f_t$ to be a valid density, we require $\int d_f(x) f(x) dx = 0$ and assume that $t$ is sufficiently small so that $f_t(x) \geq 0$. Let the corresponding density be $p_{\eta_t}(o)$, with score
	\begin{align*}
		B(d_\omega, d_\mu, d_f) = \frac{d}{dt} \log p_{\eta_t}(O) \Big|_{t=0} = -\frac{A\omega - 1}{\omega(\omega-1)} d_\omega +\frac{A(Y - \mu)}{\mu(1-\mu)} d_\mu + d_f.
	\end{align*}
	Recall that the parameter is $\psi = \int \mu(x) f(x) dx$. The parameter implied by the submodels is 
	\begin{align*}
		\psi_t = \int \mu_t f_t d\nu = \psi + t\int \{d_\mu(x) + \mu(x)d_f(x)\}f(x) dx + t^2\int d_\mu(x) d_f(x) f(x) dx.
	\end{align*}
	In the following proposition, we investigate the regularity of $\widehat\psi$ (or lack thereof) by deriving its sampling distribution under sequences of distributions indexed by local parameters. For shorthand notation, let $h(\omega, \overline\mu) = \E(\mu \mid A = 1, \omega, \overline\mu)$ and $g(\overline\omega, \mu) = \E(\omega \mid A = 1, \overline\omega, \mu)$. 
 \begin{proposition}\label{prop:regularity}
     Let the setup be as described above. Let $p_{n}(o)$ denote the density of the observations $p_{\eta_t}$ setting $t = n^{-1/2}$ and $\psi_n = \psi_{n^{-1/2}}$, i.e., the value of the functional at $p_n$. Suppose that
     \begin{align*}
         n^{1/2}(\widehat\psi - \psi) = n^{1/2}(\Pn - \Pb)(\overline\varphi - \overline\varphi_{\omega\mu}) + o_P(1).
     \end{align*}
     Then, under the local distribution $p_n$:
     \begin{align*}
         n^{1/2}(\widehat\psi - \psi_n) \indist N(\theta, \var(\overline\varphi - \overline\varphi_{\omega\mu})),
     \end{align*}
     where
     \begin{align*}
         \theta = \begin{cases}
             0 & \text{if } \ \overline\omega = \omega \text{ and } \overline\mu = \mu, \\
             - \E[A\{\mu - h(\omega, \overline\mu)\}d_\omega] & \text{if } \ \overline\omega = \omega \text{ and } \overline\mu \neq \mu \\
             \E[\{Ag(\overline\omega, \mu) - 1\} d_\mu] & \text{if }  \ \overline\omega \neq \omega \text{ and } \overline\mu = \mu
         \end{cases}
     \end{align*}
 \end{proposition}
	From the third statement in Proposition \ref{prop:min_squared} and Remark \ref{remark:main_rootn}, we can derive sufficient conditions such that
	\begin{align*}
		n^{1/2}(\widehat\psi - \psi) = n^{1/2}(\Pn - \Pb)(\overline\varphi - \overline\varphi_{\omega\mu}) + o_\Pb(1),
	\end{align*}
	e.g., when $\|\widehat\omega -\omega\| \ \land \ \|\widehat\mu - \mu\| = o_\Pb(n^{-1/4})$ and $\alpha = \beta = 1$. Proposition \ref{prop:regularity} shows that, under these conditions, $\widehat\psi$ is asymptotically linear but it may not be regular unless both nuisance functions are consistently estimated. A similar finding was previously noted by \cite{dukes2021doubly} for the expected conditional covariance functional. We therefore echo their conclusion that, in most applications, one should strive to estimate both nuisances as well as possible and rely on the further correction included in $\widehat\psi$ for protection against possible misspecification. 
	\section{Connections to other literature}
	
	\subsection{Higher order influence functions \label{sec:connection_to_HOIFs}}
	In this section, we expand on the similarities and differences between $\widehat\psi$ and the estimator of $\psi$ based on the (approximate) second-order influence function. The estimator is described in detail in a series of articles by Robins and co-authors \parencite{robins2008higher, robins2009quadratic, robins2017minimax, liu2017semiparametric}. It takes the form:
	\begin{align*}
		\widehat\psi_{HOIF-2} = \widehat\psi_\DR - \frac{1}{n(n-1)} \mathop{\sum\sum}\limits_{1 \leq i \neq j \leq n} (A_i \widehat\omega_i - 1) p_i^T \widehat\Omega^{-1} p_j A_j(Y_j - \widehat\mu_j),
	\end{align*}
	where $p(x) = (p_0(x), p_1(x), \ldots, p_k(x))$ is a vector of basis functions suitable for approximating the functions $\omega(x) - \widehat\omega(x)$ and $\mu(x) - \widehat\mu(x)$, and $\Omega = \E\{Ap(X)p^T(X)\}$. When $X$ is multivariate, $p(x)$ can be taken to be a tensor product of univariate basis functions. The tuning parameter $k$ is chosen to diverge with the sample size $n$ to obtain progressively better approximations of $R_n$. A natural estimator of $\Omega$ is its empirical counterpart. However, in regimes of low-smoothness, $k$ is selected to grow faster than $n$ so that the empirical version is not invertible. In this case, an alternative estimator, which requires estimating the density of $X$, is $\int p(x) p^T(x) \widehat{g}(x) dx$, where $g(x)$ denotes the density of $X$ among the units with $A = 1$ multiplied by $\Pb(A = 1)$. See \cite{liu2017semiparametric} and \cite{liu2023new} for additional details. 
	
	One can see that $\widehat\psi$ and $\widehat\psi_{HOIF-2}$ differ only in that $\widehat\Pi_k(x_i, x_j) \equiv p_i^T\widehat\Omega^{-1} p_j$ in $\widehat\psi_{HOIF-2}$ is replaced by 
	\begin{align*}
		\widehat\Pi_h(\widehat\omega_i, \widehat\mu_i, \widehat\omega_j, \widehat\mu_j) = \widehat{Q}^{-1}(\widehat\omega_j, \widehat\mu_j)K_h(\widehat\omega_i - \widehat\omega_j)K_h(\widehat\mu_i - \widehat\mu_j)
	\end{align*}
	in $\widehat\psi$. The kernel $\widehat\Pi_k(x_i, x_j)$ is designed to approximate the kernel of an orthogonal projection in $L_2(g)$. That is,
	\begin{align*}
		& \E_{Z_j \mid D^n} \left[\Pi_k(X_i, X_j)A_j\{Y_j - \mu_j(X)\}\right] = \mu(X_i) - \widehat\mu(X_i) + \text{approximation error}.
	\end{align*}
	The larger the size of the image of the projection (i.e., the larger $k$), the smaller the approximation error is. However, the variance of $\widehat\psi_{HOIF-2}$ increases with $k$, so $k$ needs to be carefully tuned to minimize the mean-square-error. In practice, one needs to use $\widehat\Pi_k$ in place of $\Pi_k$, so that the approximation error depends also on the accuracy with which $\Omega$ can be estimated. Additional higher-order corrections are needed when $\| \widehat\Omega^{-1} - \Omega^{-1}\|_{\text{op}}$ or $\|\widehat{g} - g\|_\infty$ are not sufficiently small. One key feature of the population version $\Pi_k(x_i, x_j)$ is that it satisfies
	\begin{align*}
		\E\left\{(A_i\widehat\omega_i - 1)\Pi_k(X_i, X_j)A_j(Y_j - \widehat\mu_j)\mid D^n \right\} = \int \Pi_k (\widehat\omega - \omega) (x) \Pi_k(\widehat\mu - \mu)(x) g(x)dx,
	\end{align*}
	where $\Pi_k(f)(x) = p(x)^T \Omega^{-1} \int p(x) f(x) g(x) dx$ denote the weighted orthogonal projection (with weight $g$) of $f$ onto the space spanned by $p$. This is crucial to obtain that
	\begin{align*}
		R_n - \E\left\{(A_i\widehat\omega_i - 1)\Pi_k(X_i, X_j)A_j(Y_j - \widehat\mu_j)\mid D^n \right\} = \int (I - \Pi_k)(\widehat\omega - \omega)(x) (I - \Pi_k)(b - \widehat\mu)(x) g(x) dx,
	\end{align*}
	which is a remainder error term that is particularly small because it is a product of approximation errors. The kernel $\widehat\Pi_h(\widehat\omega_i, \widehat\omega_j, \widehat\mu_i, \widehat\mu_j)$ appearing in $\widehat\psi$ does not yield a product of approximation errors even if the true $Q(\widehat\omega_j, \widehat\mu_j)$ is used. However, it is designed to achieve the similar goal of approximating the kernel of a local regression in the sense that
	\begin{align*}
		\E_{Z_j \mid D^n} \left[\Pi_h(\widehat\omega_i, \widehat\mu_i, \widehat\omega_j, \widehat\mu_j)A_j\{Y_j - \mu_j(X)\}\right] = f_\omega(\widehat\omega_i, \omega_i, \widehat\mu_i; D^n) + \text{approximation error}.
	\end{align*}
	Crucially, it retains a symmetry property from $\Pi_k(x_i, x_j)$, which is that taking the expectation with respect to $Z_j$ yields an approximation of $f_\mu(\widehat\omega_i, \omega_i, \widehat\mu_i; D^n)$, while taking the expectation with respect to $Z_i$ yields an approximation of $f_\omega(\widehat\mu_j, \mu_j, \widehat\omega_i; D^n)$:
	\begin{align*}
		\E_{Z_i \mid D^n} \left[\widehat\Pi_h(\widehat\omega_i, \widehat\mu_i, \widehat\omega_j,  \widehat\mu_j)(A_i\widehat\omega(X_i) - 1)\right] 
		& = f_\mu(\widehat\mu_j, \mu_j, \widehat\omega_i; D^n) + \text{approximation error}.
	\end{align*}
	This then allows for the estimation of $R_n$ relying on $f_\omega$ or $f_\mu$ depending on which one is smoother. One advantage of $\widehat\psi$ over $\widehat\psi_{HOIF-2}$ (or higher order versions) is that the kernel $\Pi_h$ in $\widehat\psi$ is low dimensional no matter how large the dimension of $X$ is. The advantage of estimators based on HOIFs is that they can better exploit the regularity of $\omega$ and $\mu$ (as opposed to being agnostic with respect to their structure and instead exploiting the regularity of $f_\omega$ and $f_\mu$) when the dictionary of basis functions is chosen appropriately. Further, they are grounded in a very general theoretical framework for functional estimation.
	
	\subsection{Doubly-robust inference \label{sec:connection_to_DRAL}}
	\cite{van2014targeted} and \cite{benkeser2017doubly} derive estimators of $\psi$ that remain $\sqrt{n}$-consistent and asymptotically normal even when either $\widehat\pi = 1/ \widehat\omega$ or $\widehat\mu$ (but not both) is misspecified. More recently, \cite{dukes2021doubly} have also proposed estimators enjoying this property but focused on the expected conditional covariance functional, $\E\{\cov(Y, A \mid X)\}$. In this section, we briefly outline some similarities and differences between $\widehat\psi$ and the estimator proposed in \cite{benkeser2017doubly} and implemented in the \texttt{R} package \texttt{drtmle} \parencite{benkeser2023doubly}. 
	
	To start, the authors write $R_n = \Pb(\widehat\psi_\DR - \psi)$ as
	\begin{align*}
		R_n & = \one(\overline\pi = \pi) \E\left\{\frac{(\overline\mu - \mu)(\widehat\pi - \pi)}{\pi} \mid D^n \right\} + \one(\overline\mu = \mu) \E\left\{\frac{(\widehat\mu - \mu)(\overline\pi - \pi)}{\overline\pi} \mid D^n \right\} \\
		& \hphantom{=} + \E\left\{ \frac{(\widehat\mu - \overline\mu)(\widehat{\pi} - \overline\pi)}{\overline\pi} \mid D^n \right\} + \E\left\{\frac{(\widehat\mu - \mu)(\widehat\pi - \pi)(\widehat\pi - \overline\pi)}{\overline\pi \widehat\pi} \mid D^n \right\}
	\end{align*}
	where we recall the notation $\overline\pi$ and $\overline\mu$ to denote the limits as $n \to \infty$ of $\widehat\pi$ and $\widehat\mu$ respectively. The last two terms are asymptotically negligible as long as $\widehat\pi$ and $\widehat\mu$ converge to their corresponding limits at $n^{-1/4}$-rates and either $\overline\mu = \mu$ or $\overline\pi = \pi$. Next, they make the observation that the first term can be written as
	\begin{align*}
		\E\left\{\frac{(\overline\mu - \mu)(\widehat\pi - \pi)}{\pi} \mid D^n \right\} = - \E\left[\E\{Y - \overline\mu(X) \mid A = 1, \widehat\pi(X), \pi(X), D^n\}\left\{\frac{\widehat\pi(X) - \pi(X)}{\pi(X)} \right\} \mid D^n \right]
	\end{align*}
	Notice that $\E\{Y - \overline\mu(X) \mid A = 1, \widehat\pi(X), \pi(X), D^n\}$ is essentially $s_\omega(t_1, t_2; D^n)$ up to the parametrization in terms of $\pi(x) = 1/\omega(x)$ as opposed to $\omega(x)$, and with $\overline\mu$ replacing $\widehat\mu$. Next, they show that
	\begin{align*}
		\E\left\{\frac{(\widehat\mu - \mu)(\overline\pi - \pi)}{\overline\pi} \mid D^n \right\} & = \E\left[A\{Y - \widehat\mu(X)\} \frac{\E\left\{\frac{A - \overline{\pi}(X)}{\overline\pi(X)} \mid \widehat\mu(X), \mu(X), D^n \right\}}{\E\{A \mid \widehat\mu(X), \mu(X), D^n\}} \mid D^n \right].
	\end{align*}
	Because one does not know whether $\overline\pi = \pi$ or $\overline\mu = \mu$, following the guiding principles of TMLE, they propose fluctuating $\widehat\pi$, $\widehat\mu$, as well as the estimators of the three other nuisance functions, $\widehat\E\left\{\frac{A - \widehat{\pi}(X)}{\widehat\pi(X)} \mid \widehat\mu(X), D^n \right\}$, $\widehat\E\{A \mid \widehat\mu(X), D^n\}$, and $\widehat\E\{Y - \widehat\mu(X) \mid A = 1, \widehat\pi(X), D^n\}$, so that they simultaneously satisfy:
	\begin{align*}
		& \frac{1}{n} \sum_{i = 1}^n\frac{A_i}{\widehat\pi^*(X_i)}\{Y_i - \widehat\mu^*(X_i)\} \approx 0, \\
		& \frac{1}{n} \sum_{i = 1}^n \widehat\E^*\{Y - \widehat\mu^*(X) \mid A = 1, \widehat\pi^*(X_i), D^n\}\left\{\frac{\widehat\pi^*(X_i) - A_i}{\widehat\pi^*(X_i)} \right\} \approx 0, \\
		& \frac{1}{n} \sum_{i = 1}^n A_i\{Y_i - \widehat\mu^*(X_i)\} \frac{\widehat\E^*\left\{\frac{A - \widehat{\pi}^*(X)}{\widehat\pi^*(X)} \mid \widehat\mu^*(X), D^n \right\}}{\widehat\E^*\{A \mid \widehat\mu^*(X), D^n\}} \approx 0.
	\end{align*}
	They propose an iterative procedure to solve these three moment conditions. However, their results are in terms of high-level conditions and the convergence properties of their algorithms are not fully analyzed. For example, two conditions are that 
	\begin{align*}
		& \int \left[\widehat\E^*\{Y - \overline\mu(X) \mid A = 1, \widehat\pi^*(x), D^n\} - \E\{Y - \overline\mu(X) \mid A = 1, \widehat\pi(x), \pi(x), D^n\}\right]^2 d\Pb(x) = o_\Pb(n^{-1/2}), \\
		& \int \left[\widehat\E^*\{Y - \overline\mu(X) \mid A = 1, \widehat\pi^*(x), D^n\} - \E\{Y - \overline\mu(X) \mid A = 1, \pi(x), D^n\}\right]^2 d\Pb(x) = o_\Pb(n^{-1/2}).
	\end{align*}
	In certain applications, the last condition can be hard to justify because it pertains to the estimation of a regression function on unknown covariates for which the convergence rate can be slow; see \cite{mammen2012nonparametric} and Theorem 3 in \cite{dukes2021doubly}.
	
	Our construction relies on smoothness assumptions imposed directly on $f_\omega(\widehat\omega, \omega, \widehat\mu; D^n)$ and $f_\mu(\widehat\mu, \mu, \widehat\omega; D^n)$. This allows us to avoid devising an iterative procedure. However, we do not claim that our estimator should be preferred to that proposed in \cite{benkeser2017doubly} in applications, which in fact performs well in the small simulation study described in the next section. We view the introduction of a hybrid structure-agnostic class where doubly-robust inference is possible and the analysis of a one-step estimator in this context as our main contributions.

 We emphasize that the doubly robust inference literature typically focuses on the `inconsistency regime' where one nuisance function estimator may be inconsistent for the true function, but converges quickly to a limit. Our estimator can be justified under this regime, although it can exhibit non-regularity outside of the intersection sub-model. It can also be justified under the alternative regime where both functions are estimated consistently, but only one converges as a rate faster than $n^{-1/4}$.
 
	\section{Simulation experiments \label{sec:sims}}
	In this section, we consider simulation experiments to investigate the finite-sample properties of the estimators studied. The main goal is to verify that the estimator introduced in Section \ref{sec:main_estimator} yields a performance comparable to its Doubly-Robust TMLE counterpart implemented in the \texttt{R} package \texttt{drtmle} \parencite{benkeser2023doubly}. We consider the following data generating process:
	\begin{itemize}
		\item $X = (X_1, X_2)$, where $X_i \sim \text{Unif}(-1, 1)$ for $i \in \{1, 2\}$, $X_1 \ind X_2$,
		\item $A \sim \text{Binom}(\pi(x))$ and $Y = AY^1 + (1 - A)Y^0$, where $Y^1 \sim \text{Binom}(\mu(x))$ and $Y^0 \sim \text{Binom}(0.5)$,
		\item $\pi(x) = \text{expit}(\begin{bmatrix} 1 & x \end{bmatrix}^T \beta_\pi)$, where $\beta_\pi = \begin{bmatrix} -0.5 & 2 & 0.5 \end{bmatrix}^T$,
		\item $\mu(x) = \text{expit}(\begin{bmatrix} 1 & x \end{bmatrix}^T \beta_\mu)$, where $\beta_\mu = \begin{bmatrix} -2.5 & 5 & 2 \end{bmatrix}^T$,
		\item $\widehat\pi(x) = \text{expit}(\begin{bmatrix} 1 & x \end{bmatrix}^T \widehat\beta_\pi)$, where $\widehat\beta_\pi = \beta_\pi + N(n^{-r_\pi}, n^{-2r_\pi})$,
		\item $\widehat\mu(x) = \text{expit}(\begin{bmatrix} 1 & x \end{bmatrix}^T \widehat\beta_\mu)$, where $\widehat\beta_\mu = \beta_\mu + N(n^{-r_\mu}, n^{-2r_\mu})$.
	\end{itemize}
	In this set-up, the target parameter is $\psi = \E \{\mu(X)\} \approx 0.66$ and $\widehat\omega$ and $\widehat\mu$ converge to $\omega$ and $\mu$ at rates $n^{-r_\pi}$ and $n^{-r_\mu}$ respectively. We vary $r_\pi, r_\mu$ in  $\{0, 0.3\}$ and the sample size $n$ in $\{500, 1000, 1500, 2000\}$. We run 500 simulations. We select the bandwidth for constructing $\widehat\psi_\omega$, $\widehat\psi_\mu$ and $\widehat\psi$ as the one that we would choose to estimate $\E(Y - \widehat\mu(X) \mid A = 1, \widehat\omega(X), \widehat\mu(X), D^n\}$. In particular, we use the cross-validation-based selector from the \texttt{R} package \texttt{sm} (\texttt{h.select()} function) \parencite{sm_package}. Our theoretical results do not justify this choice, but, in this simulation set-up, it yields reasonable results. As a benchmark, we also consider the performance of an oracle estimator that has access to the true nuisance functions $\omega(x)$ and $\mu(x)$. 
	
	Figure \ref{fig:sim_boxplot} reports the distribution of $\sqrt{n}(\overline\psi - \psi)$ for different choices of $\overline\psi$ and the 4 possible combinations of $(r_\pi, r_\mu)$. The sample size is $n = 2000$. Our theoretical results suggest that $\widehat\psi_\DR$ should not perform well except in setting where $r_\pi = r_\mu = 0.3$, while all estimators should perform poorly when $r_\pi = r_\mu = 0$. However, in this small simulation setup, we find that the corrected estimators still perform reasonably well. The main take-away from Figure \ref{fig:sim_boxplot} is that $\widehat\psi_\DR$ is the only estimator suffering from the mis-specification of one nuisance function. 
	
	Figure \ref{fig:sim_coverage} reports the coverage of a Wald confidence interval as a function of the sample size $n$. The most noticeable pattern is that an interval based on $\widehat\psi_\DR$ fails to achieve nominal coverage in all scenarios considered except when both nuisance functions are correctly specified. When at least one nuisance function is correctly specified, the performance of the corrected estimators is remarkable \footnote{In an earlier version of the manuscript, we reported simulation results that showed some under-coverage for $\widehat\psi$ when the propensity score was mis-specified. We initially thought that the underestimation of the standard error was, at least partially, due to extreme propensity score estimates in this scenario. However, after correcting the variance calculation, we noticed that the estimator now achieves nominal coverage in this setting as well.} Finally, as expected, when both nuisance functions are mis-specified all estimators exhibit some under-coverage.
	\begin{figure}[H]
		\begin{subfigure}{.5\textwidth}
			\centering
			\includegraphics[scale = 0.3]{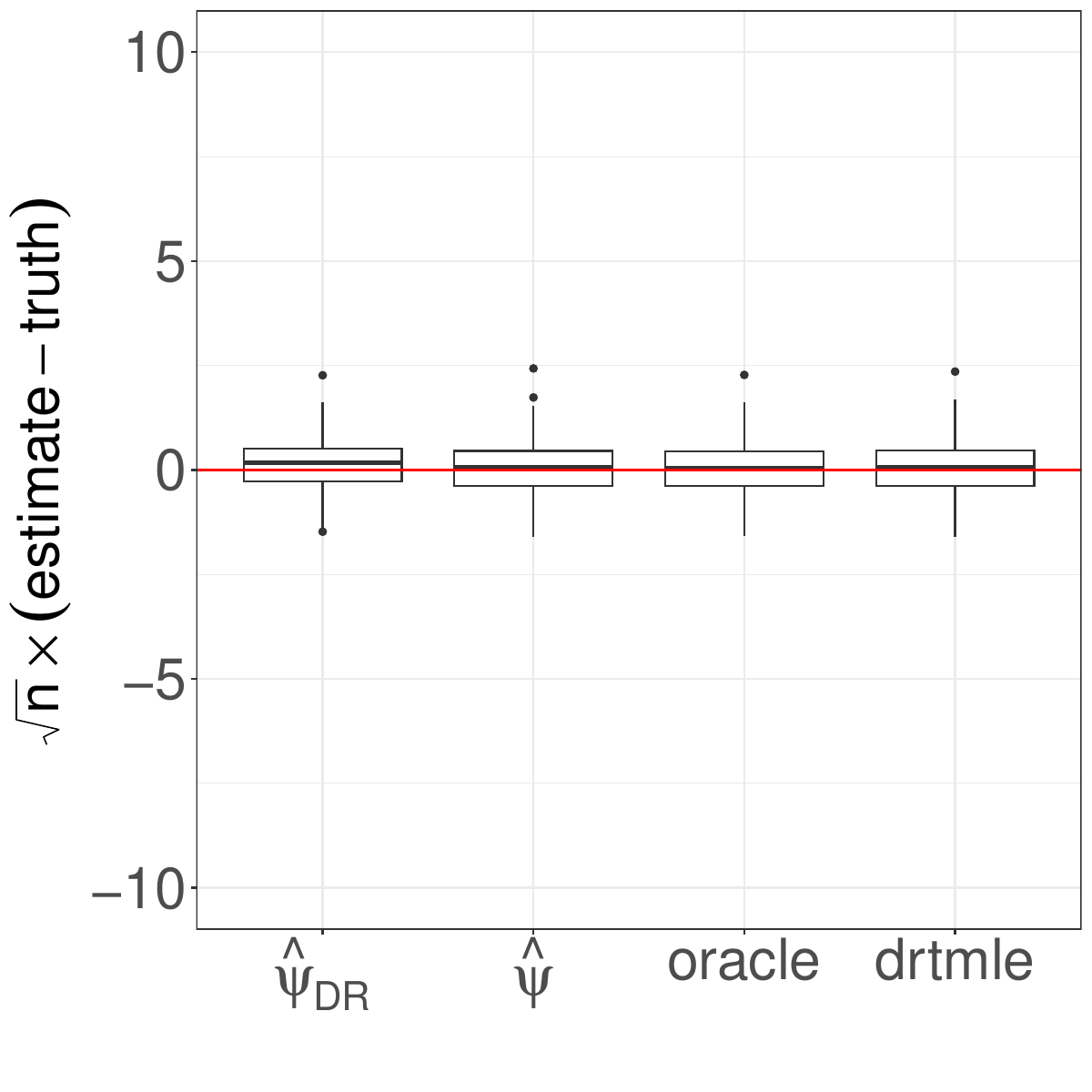}
			\caption{$r_\pi = 0.3$ and $r_\mu = 0.3$.}
		\end{subfigure}\hfill
		\begin{subfigure}{.5\textwidth}
			\centering
			\includegraphics[scale = 0.3]{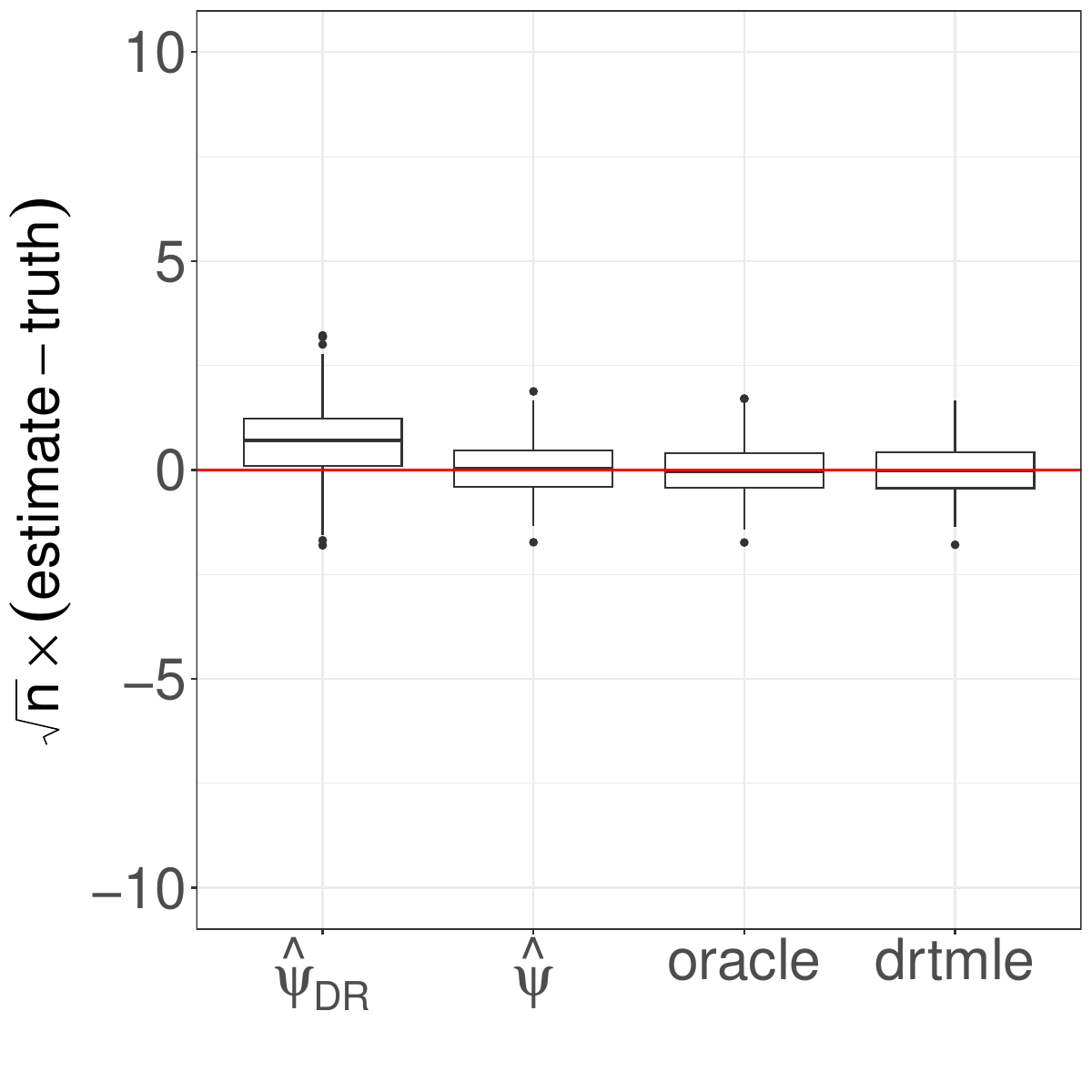}
			\caption{$r_\pi = 0.3$ and $r_\mu = 0$.}
		\end{subfigure}\hfill
		\begin{subfigure}{.5\textwidth}
			\centering
			\includegraphics[scale = 0.3]{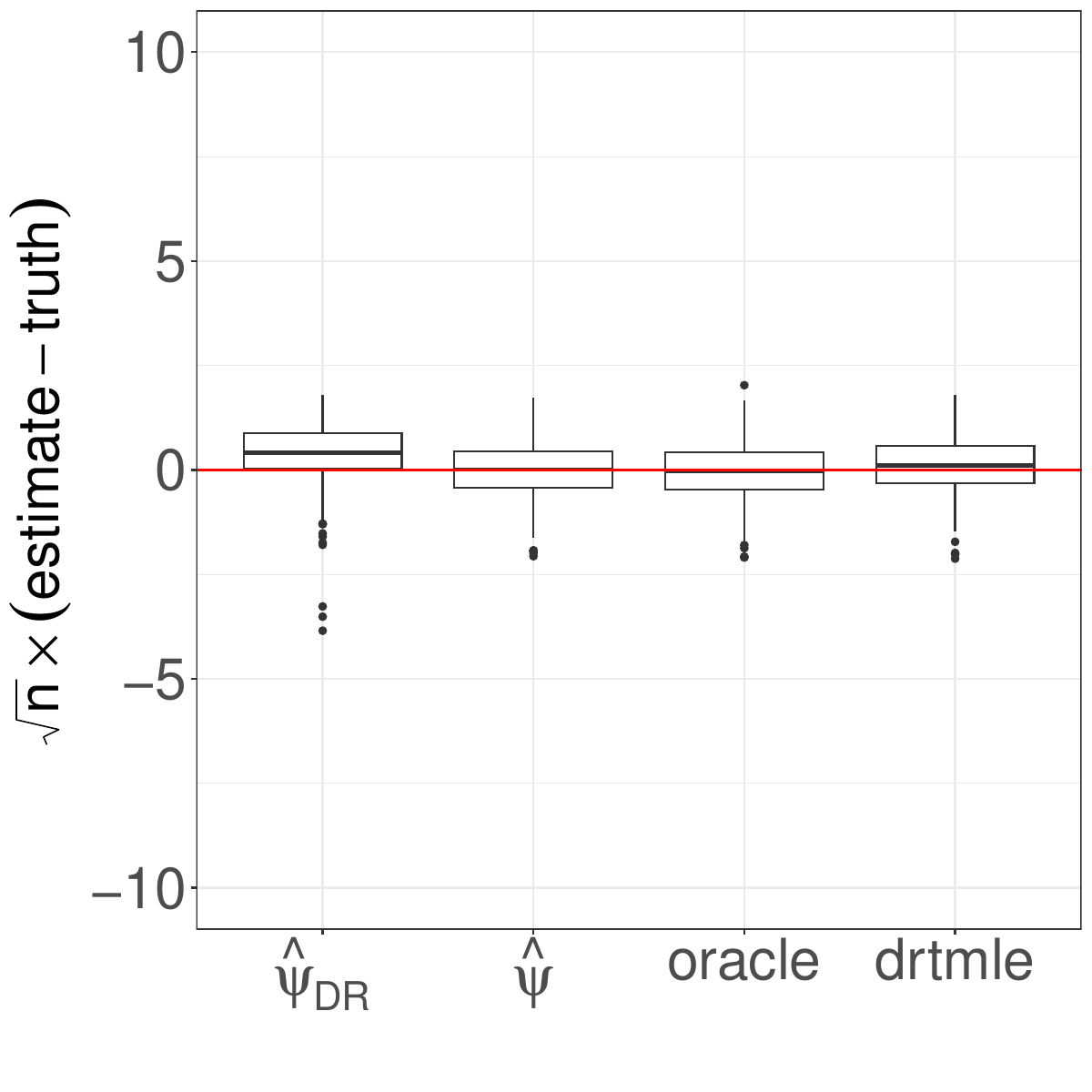}
			\caption{$r_\pi = 0$ and $r_\mu = 0.3$.}
		\end{subfigure}\hfill
		\begin{subfigure}{.5\textwidth}
			\centering
			\includegraphics[scale = 0.3]{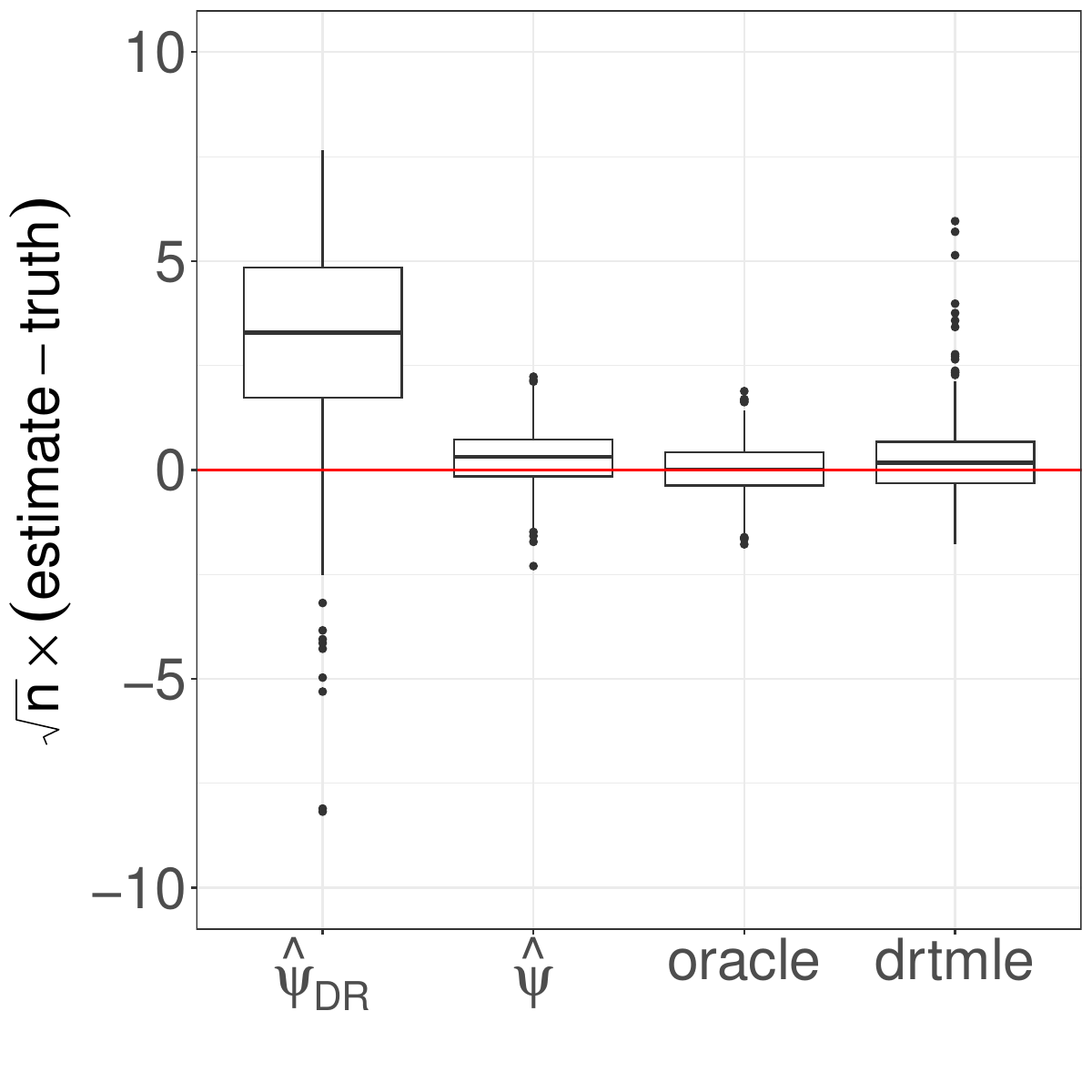}
			\caption{$r_\pi = 0$ and $r_\mu = 0$.}
		\end{subfigure}\hfill
		\caption{Distribution of the errors scaled by $\sqrt{n}$, where $n = 2000$.\label{fig:sim_boxplot}}
	\end{figure}
	
	\begin{figure}[H]
		\begin{subfigure}{.5\textwidth}
			\centering
			\includegraphics[scale = 0.3]{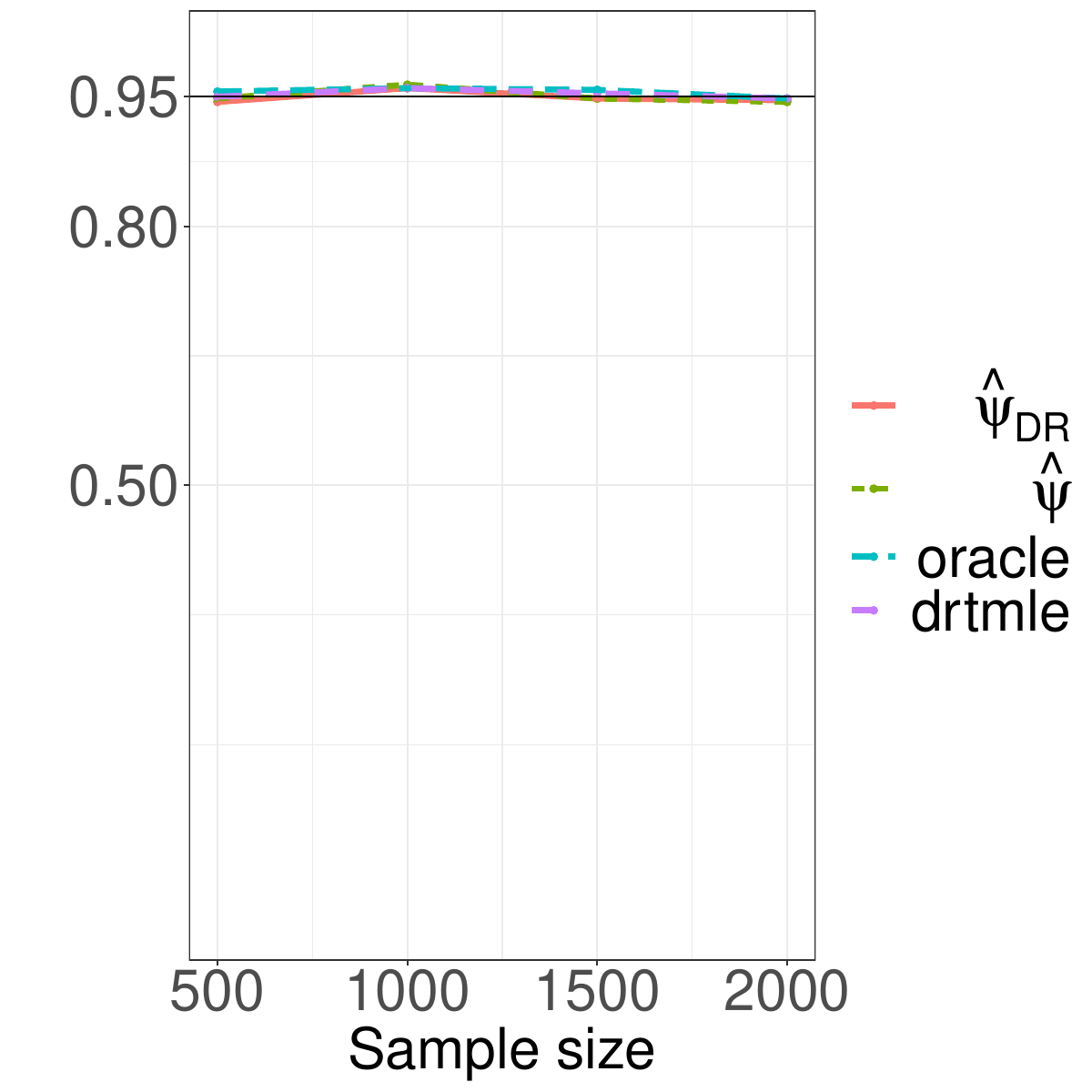}
			\caption{$r_\pi = 0.3$ and $r_\mu = 0.3$.}
		\end{subfigure}\hfill
		\begin{subfigure}{.5\textwidth}
			\centering
			\includegraphics[scale = 0.3]{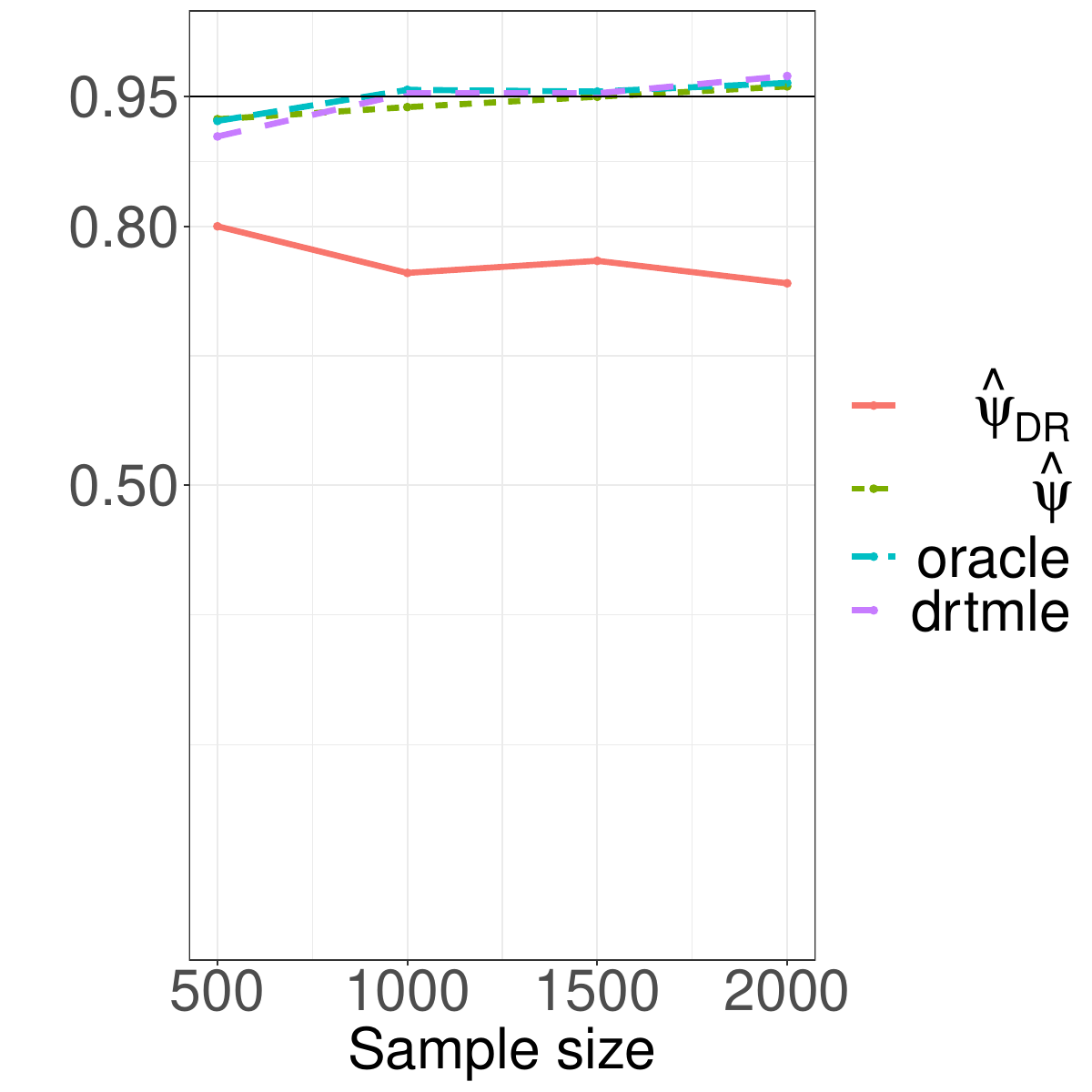}
			\caption{$r_\pi = 0.3$ and $r_\mu = 0$.}
		\end{subfigure}\hfill
		\begin{subfigure}{.5\textwidth}
			\centering
			\includegraphics[scale = 0.3]{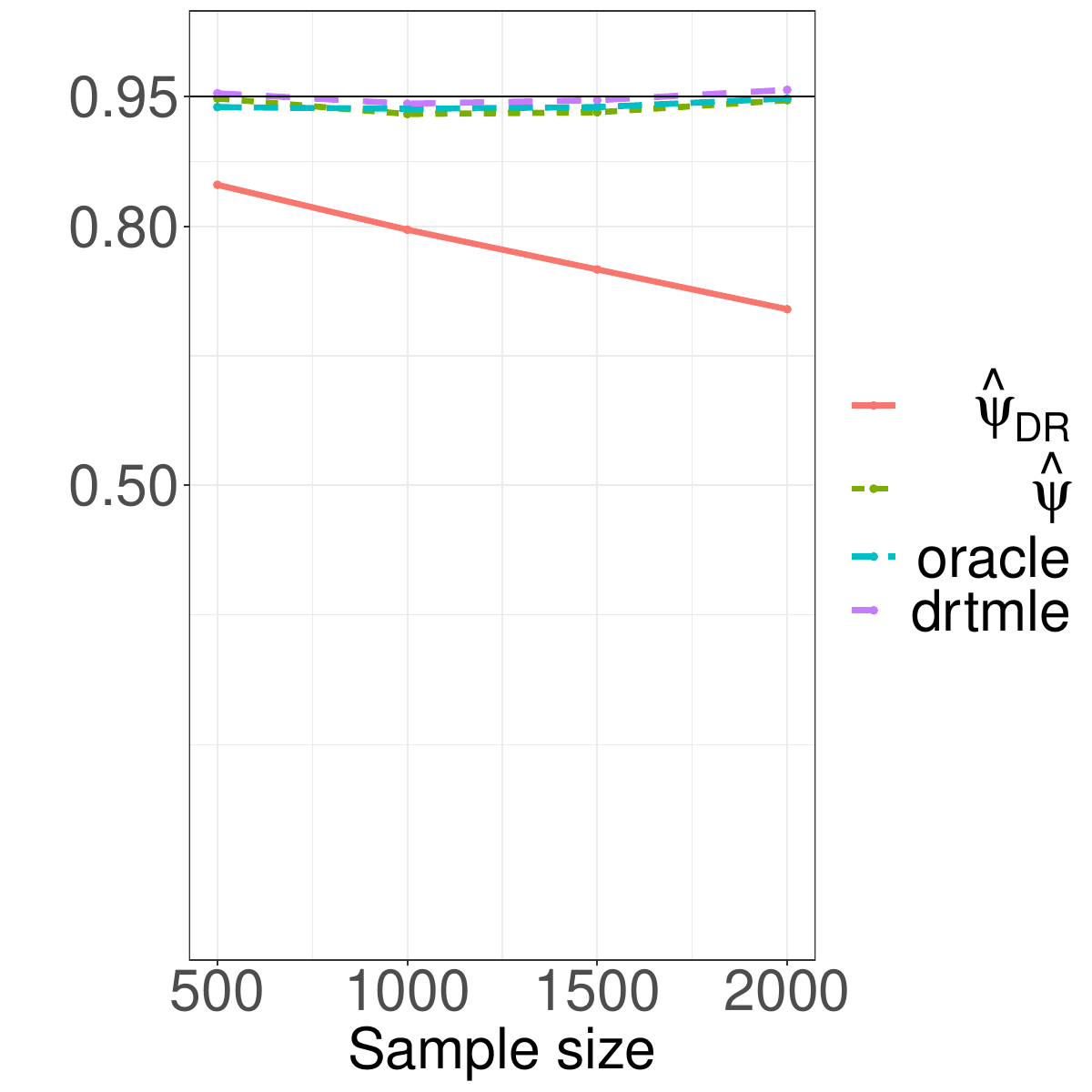}
			\caption{$r_\pi = 0$ and $r_\mu = 0.3$.}
		\end{subfigure}\hfill
		\begin{subfigure}{.5\textwidth}
			\centering
			\includegraphics[scale = 0.3]{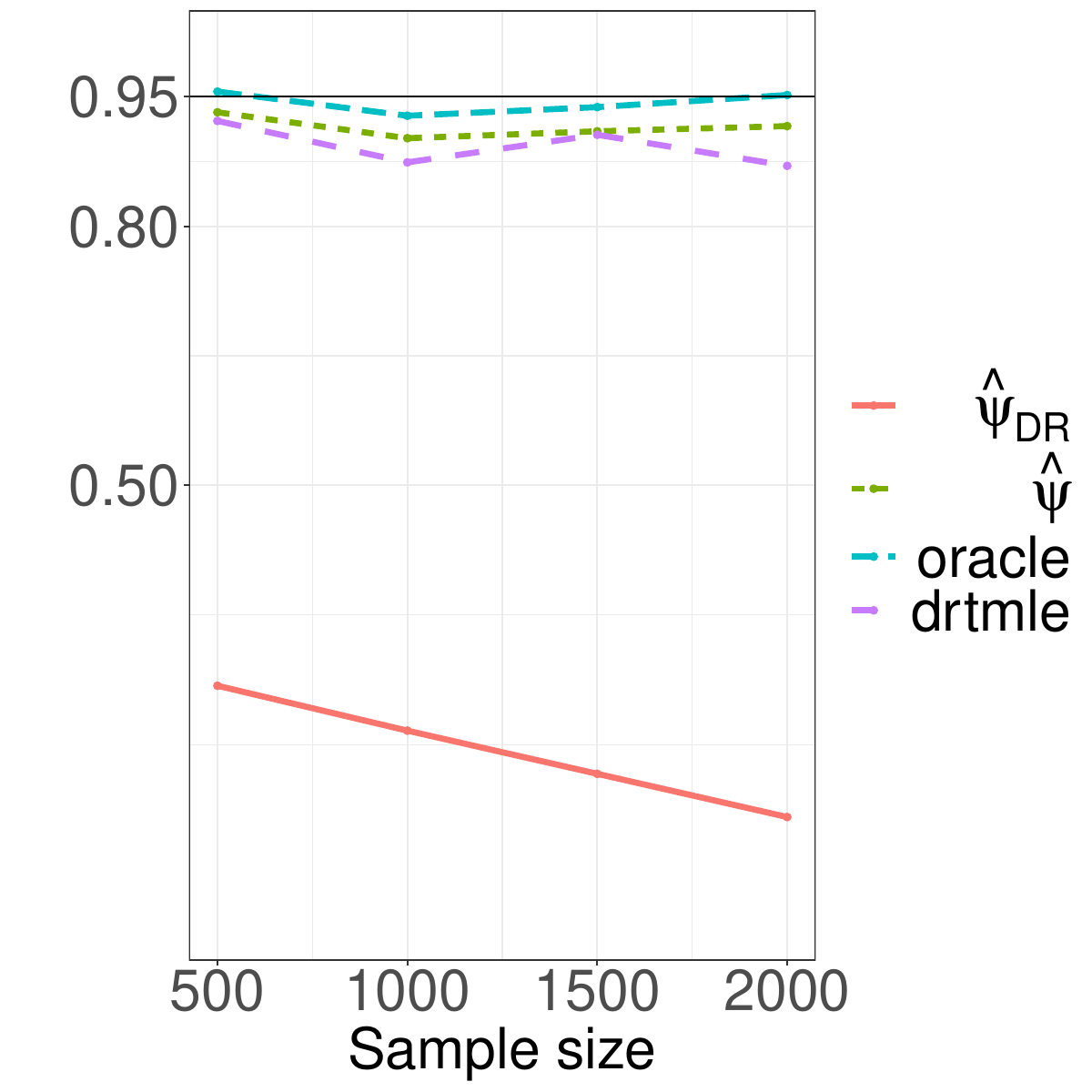}
			\caption{$r_\pi = 0$ and $r_\mu = 0$.}
		\end{subfigure}\hfill
		\caption{Coverage as a function of the sample size $n$. \label{fig:sim_coverage}.}
	\end{figure}
	\section{Conclusions}
	In this work, we have investigated the possibility of constructing estimators of the ATE functional that remain asymptotically linear even if one of the two nuisance functions is not consistently estimated. We have proposed a novel function class that is a hybrid between the purely structure-agnostic one proposed in \cite{balakrishnan2023fundamental}, and analyzed in detail in \cite{jin2024structure}, and more traditional ones based on H\"{o}lder smoothness. In completely structure-agnostic models, where all that is known are rates of convergence for the two nuisance functions, our Proposition \ref{prop:pure_sa} (covering the case when the propensity score can be estimated at an equal or faster rate than the outcome model), and, more generally the concurrent work by \cite{jin2024structure}, show that nonparametric, doubly-robust root-$n$ inference is not attainable. On the contrary, we show that this is possible in the new hybrid class proposed. Further, merging ideas from the literature on doubly-robust inference and that on HOIFs, we have constructed an estimator that, under certain conditions, exhibits doubly-robust asymptotic linearity (DRAL). The sufficient conditions needed for DRAL that we have derived are relatively straightforward to describe, although they pertain to nonstandard regression models where both the outcome and the covariates depend on a training sample used to estimate the nuisance functions. In addition, we have shown that this estimator is minimax optimal in the new hybrid model proposed. In future work, it would be interesting to expand the results presented in Proposition \ref{siva_prop} to cover cases where, for example, the additional nuisance functions involving generated regressors are H\"{o}lder smooth of orders less than 1. Furthermore, it would be also interesting to study other nonparametric models where these additional nuisance functions satisfy structure-agnostic rate conditions instead of the smoothness constrains that we have imposed.
	
	\section{Acknowledgements}
	MB gratefully acknowledges support from NSF DMS Grant 2413891. EK gratefully acknowledges support from NIH R01 Grant LM013361-01A1 and NSF CAREER Award 2047444. OD was supported by FWO grant 1222522N. MB thanks the participants at the Center for Causal Inference Seminar Series at the University of Pennsylvania, at the Center for Data Science Seminar Series at Zhejiang University, and at Seminar Series at New York University Langone, as well as Lin Liu (Shanghai Jiao Tong University), Sijia Li (Harvard University) and Xiudi Li (UC Berkeley), for very insightful conversations. 
	\printbibliography
	
	\appendix
	\section{Example: partially linear logistic model}\label{appendix:partially_linear_logit}
	We revisit Example 3 in \cite{dukes2021doubly}. Here, the goal is to derive an estimator of $\theta$ under the model
	\begin{align*}
		\logit \Pb(Y = 1 \mid A, X) = \theta_0 A + m_0(X),
	\end{align*}
	where $A \in \R$ and $m_0(X) = \E(Y \mid A = 0, X)$. \cite{tan2019doubly} showed that a doubly-robust estimator of $\theta$ can be found by solving the empirical version of the moment condition
	\begin{align*}
		\psi(\theta_0) = \E\left[ \{A - v(X)\}\{Ye^{-\theta_0 A - m_0(X)} - (1 - Y)\} \right] = 0
	\end{align*}
	where $v(X) = \E(A \mid Y = 0, X)$. By the usual decomposition for M-estimators, the conditional bias of $\widehat\theta$ solving the empirical version of the moment condition above can be derived from the conditional bias of the moment condition itself. \cite{tan2019doubly} has shown that such conditional bias is
	\begin{align*}
		R_n = \E\left[ (1 - Y) \left(e^{m_0(X) - \widehat{m}_0(X)} - 1\right)\left\{v(X) - \widehat{v}(X)\right\} \mid D^n \right],
	\end{align*}
	which is not of the variety of remainder terms studied in \cite{dukes2021doubly}. Our work shows that, if one is willing to assume smoothness of the functions
	\begin{align*}
		& f_v(t_1, t_2, t_3; D^n) = \E\left\{ e^{m_0(X) - \widehat{m}_0(X)} - 1\mid Y = 0, \widehat{v}(X) = t_1, v(X) = t_2, \widehat{m}_0(X) = t_3, D^n \right\} \\
		& f_m(t_1, t_2, t_3; D^n) = \E\left\{ v(X) - \widehat{v}(X) \mid Y = 0, \widehat{m}_0(X) = t_1, m_0(X) = t_2, \widehat{v}(X) = t_3, D^n \right\},
	\end{align*}
	then it is possible to derive an estimator of $\theta$ that admits doubly-robust asymptotic linearity. In fact, such estimator would be based on an augmented moment condition of the same variety as the estimator considered in Section \ref{sec:main_estimator}. In particular, let $\widehat\theta$ solve
	\begin{align*}
		& \widehat\psi(\widehat\theta) = \Pn \left[\{A - \widehat{v}(X)\}\{Ye^{-\widehat\theta A - \widehat{m}_0(X)} - (1 - Y)\}\right] - T_n, \text{ where } \\
		& T_n =\frac{1}{n(n-1)} \mathop{\sum\sum}\limits_{1 \leq i \neq j \leq n} \left\{e^{-\theta A_i - \widehat{m}_i} Y_i - (1-Y_i)\right\} \frac{K_h(\widehat{v}_i - \widehat{v}_j)K_h(\widehat{m}_i - \widehat{m}_j)}{\widehat{Q}(\widehat{v}_i, \widehat{m}_i)} (1 - Y_j)(A_j - \widehat{v}_j), \\
		& \widehat{Q}(\widehat{v}_i, \widehat{m}_i) = \frac{1}{n-1}\sum_{s = 1, s \neq i}^n (1 - Y_s)K_h(\widehat{v}_s - \widehat{v}_i)K_h(\widehat{m}_s - \widehat{m}_i),
	\end{align*}
	where we recall the notation $\widehat{v}_i = \widehat{v}(X_i)$, $\widehat{m}_i = \widehat{m}_0(X_i)$.
	
	A proposition analogous to Proposition \ref{prop:min_squared} can be derived for $\psi(\theta_0) - \widehat\psi(\theta_0)$. To see this, notice that
	\begin{align*}
		R_n = \E\{(1 - Y) f_v(\widehat{v}, v, \widehat{m}_0; D^n)(v - \widehat{v}) \mid D^n \} = \E\left\{(1 - Y) \left(e^{m_0 - \widehat{m}}  - 1\right)f_m(\widehat{m}_0, m_0, v; D^n) \mid D^n \right\}
	\end{align*}
	Furthermore, for $Q(\widehat{v}_i, \widehat{m}_i) = \E\{\widehat{Q}(\widehat{v_i}, \widehat{m}_i) \mid X_i, D^n\}$:
	\begin{align*}
		& \E_{Z_j \mid D^n} \left\{ \frac{K_h(\widehat{v}_i - \widehat{v}_j)K_h(\widehat{m}_i - \widehat{m}_j)}{Q(\widehat{v}_j, \widehat{m}_j)} (1 - Y_j)(A_j - \widehat{v}_j) \right\}  \\
		& = \E_{Z_j \mid D^n} \left\{ \frac{K_h(\widehat{v}_i - \widehat{v}_j)K_h(\widehat{m}_i - \widehat{m}_j)}{Q(\widehat{v}_i, \widehat{m}_i)} (1 - Y_j)(v_j - \widehat{v}_j) \right\} \\
		& = \E_{Z_j \mid D^n} \left\{ \frac{K_h(\widehat{v}_i - \widehat{v}_j)K_h(\widehat{m}_i - \widehat{m}_j)}{Q(\widehat{v}_i, \widehat{m}_i)} (1 - Y_j) f_m(\widehat{m}_j, m_j, \widehat{v}_j; D^n) \right\} \\
		& = f_m(\widehat{m}_i, m_i, \widehat{v}_j; D^n) + O(h^\alpha + \|\widehat{m} - m\|^\alpha)
	\end{align*}
	where the last equality follows if $f_m(\widehat{m}, m, \widehat{v}; D^n)$ is H\"{o}lder of order $\alpha$. Similarly, for $Q(\widehat{v}_j, \widehat{v}_j) = \E\{(1 - Y)K_h(\widehat{v} - \widehat{v}_j)K_h(\widehat{m} - \widehat{m}_j) \mid D^n\}$:
	\begin{align*}
		& \E_{Z_i \mid D^n} \left[ \left\{ e^{-\theta A_i - \widehat{m}_i} Y_i - (1-Y_i) \right\} \frac{K_h(\widehat{v}_i - \widehat{v}_j)K_h(\widehat{m}_i - \widehat{m}_j)}{Q(\widehat{v}_j, \widehat{m}_j)} \right] \\
		& = \E_{Z_i \mid D^n} \left\{\frac{e^{ m_i - \widehat{m}_i} -1}{1 + e^{\theta A_i + m_i}} \frac{K_h(\widehat{v}_i - \widehat{v}_j)K_h(\widehat{m}_i - \widehat{m}_j)}{Q(\widehat{v}_j, \widehat{m}_j)} \right\}  \\
		& = \E_{Z_1 \mid D^n} \left\{ (1 - Y_i)\left(e^{ m_i - \widehat{m}_i} -1\right) \frac{K_h(\widehat{v}_i - \widehat{v}_j)K_h(\widehat{m}_i - \widehat{m}_j)}{Q(\widehat{v}_j, \widehat{m}_j)} \right\} \\
		& = \E_{Z_i \mid D^n} \left\{ (1 - Y_i) f_v(\widehat{v}_i, v_i, \widehat{m}_i; D^n)\frac{K_h(\widehat{v}_i - \widehat{v}_j)K_h(\widehat{m}_i - \widehat{m}_j)}{Q(\widehat{v}_j, \widehat{m}_j)} \right\} \\
		& = f_v(\widehat{v}_j, v_j, \widehat{m}_j; D^n) + O(h^\beta + \|\widehat{v} - v\|^\beta)
	\end{align*}
	where the last equality follows if $f_v(\widehat{v}, v, \widehat{m}; D^n)$ is H\"{o}lder of order $\beta$. The errors $\widehat{Q}(\widehat{v}_i, \widehat{m}_i) - Q(\widehat{v}_i. \widehat{m}_i)$ and $Q(\widehat{v}_i, \widehat{m}_i) - Q(\widehat{v}_j, \widehat{m}_j)$, as well as the variance calculations, can be carried out in a way analogous to that used to prove Proposition \ref{prop:min_squared}. 
	\section{Proofs regarding the lower bounds}
	\subsection{Useful Lemmas}
	First, we recall the definition of Hellinger distance. Let $P$ and $Q$ be two probability measures with densities $p$ and $q$ relative to some $\sigma$-finite dominating measure $\nu$.
	\begin{definition} [Hellinger distance]
		The Hellinger distance between $P$ and $Q$ is
		\begin{align*}
			H(P, Q) = \left\{\int (\sqrt{p} - \sqrt{q})^2 d\nu\right\}^{1/2}.
		\end{align*}
	\end{definition}
	The following Lemma is a restatement of Theorem 2.15 in \cite{tsybakov2008introduction} and Lemma 1 in \cite{balakrishnan2023fundamental} with a few simplifications tailored to our settings.
	\begin{lemma}[Lemma 1 in \cite{balakrishnan2023fundamental} / Theorem 2.15 in \cite{tsybakov2008introduction}]\label{lemma:tsybakov_main_LB}
		Let $F(\theta)$ be a functional and $\{P_\theta, \theta \in \Theta\}$ be a statistical model. Consider two priors distributions $\mu_0$ and $\mu_1$ on $\Theta$ and define the posteriors:
		\begin{align*}
			P_j(A) = \int P_\theta^n(A) \mu_j(d\theta), \text{ for all measurable } A \text{ and } j \in \{0, 1\}.
		\end{align*}Assume that
		\begin{enumerate}
			\item There exists $c$ and $s > 0$ such that $\mu_0(\theta: F(\theta) \leq c) = 1$ and $\mu_1(\theta: F(\theta) \geq c + 2s) = 1$.
			\item The Hellinger distance satisfies $H^2(P_1, P_0) \leq \alpha < 2$.
		\end{enumerate}
		Then, it holds that
		\begin{align*}
			\inf_{T_n} \sup_{\theta \in \Theta} \E| T_n - F(\theta)| \ \geq s \cdot \frac{1 - \sqrt{\alpha(1-\alpha/4)}}{2}.
		\end{align*}
	\end{lemma}
	\begin{lemma}[Theorem 2.1 in \cite{robins2009semiparametric}]\label{lemma:hellinger_robins}
		Let $k$ in $\N$, let $\mathcal{X} = \cup_{j = 1}^k \mathcal{X}_j$ be a measurable partition of the sample space. Given a vector $\lambda = (\lambda_1, \ldots, \lambda_k)$ in some product measurable space $\Lambda = \lambda_1 \times \cdots \lambda_k$ let $P_\lambda$ and $Q_\lambda$ be probablity measures on $\mathcal{X}$ such that
		\begin{enumerate}
			\item $P_\lambda(\mathcal{X}_j) = Q_\lambda(X_j) = p_j$ for every $\lambda \in \Lambda$, for some probability vector $(p_1, \ldots, p_k)$. 
			\item The restrictions of $P_\lambda$ and $Q_\lambda$ to $\mathcal{X}_j$ depend on the $j\text{th}$ coordinate $\lambda_j$ of $\lambda = (\lambda_1, \ldots, \lambda_k)$ only.
		\end{enumerate}
		For $p_\lambda$ and $q_\lambda$ densities of the measures $P_\lambda$ and $Q_\lambda$ that are jointly measurable in the parameter $\lambda$ and the observation, and $\pi$ a probability measure on $\Lambda$, define $p = \int p_\lambda d\pi(\lambda)$  and $q = \int q_\lambda d\pi(\lambda)$, and set
		\begin{align*}
			\delta_1 = \max_j \sup_\lambda \int_{\mathcal{X}_j} \frac{(p_\lambda - p)^2}{p_\lambda} \frac{d\nu}{p_j}, \quad \delta_2 = \max_j \sup_\lambda \int_{\mathcal{X}_j} \frac{(q_\lambda - p_\lambda)^2}{p_\lambda} \frac{d\nu}{p_j},  \quad \delta_3 = \max_j \sup_\lambda \int_{\mathcal{X}_j} \frac{(q - p)^2}{p_\lambda} \frac{d\nu}{p_j}.
		\end{align*}
		If $np_j(1 + \delta_1 + \delta_2) \leq A$ for all $j$ and $\underline{B} \leq p_\lambda \leq \overline{B}$ for positive constants $A$, $\underline{B}$, $\overline{B}$, then there exists a constant $C$ that depends only on $A$, $\underline{B}$, $\overline{B}$ such that, for any product probability measure $\pi = \pi_1 \otimes \cdots \otimes p_k$,
		\begin{align*}
			H^2(P_1, P_2) \leq C n^2 (\max_j p_j)(\delta_2^2 + \delta_1\delta_2) + Cn\delta_3.
		\end{align*}
	\end{lemma}
	We proceed essentially as in the proof of Theorem 3.1 in \cite{robins2009semiparametric}, but choose certain parameters and the fluctuations in their construction differently in accordance with the assumptions defining $\mathcal{P}(\epsilon_n, \delta_n)$, $\mathcal{P}_\omega(\epsilon_n)$ and $\mathcal{P}_\mu(\delta_n)$, and $\mathcal{P}_{\omega\mu}(\epsilon_n, \delta_n)$. The construction is conceptually very similar to the one used in \cite{balakrishnan2023fundamental}. The only changes pertain to the choice of the fluctuations for the nuisance parameters, which are tailored to the ATE functional considered here.
	
	Let $B: \R^d \to \R$ be a function with support on $[0, 1/2]^d$ such that $\int  B(u) du = 0$ and $\int B^2(u) du = 1$. Let $k$ be an integer and $\mathcal{X}_1, \ldots, \mathcal{X}_k$ be traslates of the cube $k^{-1/d}[0, 1/2]^d$ that are disjoint and contained in $[0, 1]^d$. Let $m_1, \ldots, m_k$ be the bottom left corners of these cubes. Let $\lambda = (\lambda_1, \ldots, \lambda_k) \in \{-1, 1\}^k$. For shorthand notation, let $B_j(x) = B\{k^{1/d}(x- m_j)\}$. 
	\subsection{Proof of Proposition \ref{prop:pure_sa}}\label{appendix:prop_pure_sa}
	Suppose that $\epsilon_n \leq \delta_n$ and define:
	\begin{align*}
		& \omega_{p_\lambda}(x) = \omega_{q_\lambda}(x) = \widehat\omega(x) + \epsilon_n \sum_{j  = 1}^k \lambda_j B_j(x), \quad \mu_{p_\lambda}(x)  = \widehat\mu(x) - \epsilon_n \frac{\widehat\mu(x)}{\widehat\omega(x)} \sum_{j = 1}^k \lambda_j B_j(x), \\
		&  \text{and} \quad \mu_{q_\lambda}(x) = \mu_{1\lambda}(x) + \frac{\delta_n}{\widehat\omega(x)} \sum_{j  = 1}^k \lambda_j B_j(x).
	\end{align*}
	We have $\| \widehat\omega - \omega_{p_\lambda}\|_2 \ = \|\widehat\omega - \omega_{q_\lambda}\| \ =  \epsilon_n$, $\|\widehat\mu - \widehat\mu_{p_\lambda} \| \ \lesssim \epsilon_n \leq \delta_n$, and $\|\widehat\mu - \mu_{q_\lambda}\| \ \lesssim \delta_n$. Consider the densities $p_\lambda(x)$ and $q_\lambda(x)$ for $Y, A \in \{0, 1\}$ and $X \in [0, 1]^d$:
	\begin{align*}
		& p_\lambda(Y, A, X) = g(X) \{\omega_{p_\lambda}(X) - 1\}^{1-A}[\mu_{p_\lambda}(X)^Y\{1-\mu_{p_\lambda}(X)\}^{1-Y}]^A \\
		& q_\lambda(Y, A, X) = g(X) \{\omega_{q_\lambda}(X) - 1\}^{1-A}[\mu_{q_\lambda}(X)^Y\{1-\mu_{q_\lambda}(X)\}^{1-Y}]^A
	\end{align*}
	By construction, $p_\lambda$ and $q_\lambda$ are part of $\mathcal{P}(\epsilon_n, \delta_n)$. Set $g(X) = g = \{\int \widehat\omega(x) dx\}^{-1}$ so that the density of $X$, $f_{p_\lambda}(x) = f_{q_\lambda}(x) = \omega_{p_\lambda}(x) g(x) = \omega_{q_\lambda}(x) g(x)$, integrates to 1. We have
	\begin{align*}
		\psi_{p_\lambda} &=  g \int \omega_{p_\lambda}(x) \mu_{p_\lambda}(x) dx = g \int \widehat\omega(x) \widehat\mu(x) dx - g \epsilon^2_n \sum_{j = 1}^k \int B^2_j(x) \frac{\widehat\mu(x)}{\widehat\omega(x)} dx \\
		\psi_{q_\lambda}	& =  \int \omega_{q_\lambda}(x) \mu_{q_\lambda}(x) g(x) dx =  \psi_{p_\lambda} +  g \delta_n \epsilon_n \int \sum_{j = 1}^k B^2_j(x)\{\widehat\omega(x)\}^{-1} dx
	\end{align*}
	This means that  $|\psi_{q_\lambda} - \psi_{p_\lambda}| \ \gtrsim \epsilon_n \delta_n$ for every $\lambda$. Let $\omega(\lambda)$ denote a product prior on $\lambda = (\lambda_1, \ldots, \lambda_k)$ such that $\omega_j(\lambda = -1) = \omega_j(\lambda = 1) = 1/2$.
	Let $O = (X, A, Y)$ and define $\mathcal{O} = \cup_{j = 1}^k \mathcal{O}_j$, where $\mathcal{O}_j = \mathcal{X}_j \times \{0, 1\}^2$. Notice that $p_j = \int_{\mathcal{O}_j} p_\lambda d\nu = \int_{\mathcal{O}_j} q_\lambda d\nu  = k^{-1}$. Next, we have
	\begin{align*}
		& p(O):= \int p_\lambda(Y, A, X) d\omega(\lambda) = g(X) \{\widehat\omega(X) - 1\}^{1-A} [\widehat\mu(X)^Y\{1-\widehat\mu(X)\}^{1-Y}]^A, \\
		& p_\lambda(O) - p(O) = g(X) \{\omega_{p_\lambda}(X) - \widehat\omega(X)\}^{1-A}[\{\mu_{p_\lambda}(X) - \widehat\mu(X)\}^Y\{\widehat\mu(X) - \mu_{p_\lambda}(X)\}^{1-Y}]^A, \\
		& q_\lambda(O) - p_\lambda(O) = A g(X) \{\mu_{q_\lambda}(X) - \mu_{p_\lambda}(X)\}^Y\{ \mu_{p_\lambda}(X) - \mu_{q_\lambda}(X) \}^{1-Y}, \\
		& q(O) - p(O) := \int \{q_\lambda(Y, A, X) - p_\lambda(Y, A, X)\} d\omega(\lambda).
	\end{align*}
	Next, we apply Lemma \ref{lemma:hellinger_robins}. Notice that $\delta_1 \lesssim \epsilon_n^2$, $\delta_2 \lesssim \delta_n^2$ and $\delta_3 = 0$. Therefore,
	\begin{align*}
		H^2\left( \int p^n_\lambda d\omega(\lambda), \int q^n_\lambda d\omega(\lambda)\right) \lesssim n^2k^{-1}(\delta_n^4 + \delta^2_n \epsilon^2_n).
	\end{align*}
	In this light, choosing $k$ large enough yields that the Hellinger distance is bounded. The lower bound then follows by Lemma \ref{lemma:tsybakov_main_LB}. 
	\subsection{Proof of Proposition \ref{siva_prop}}
    Recall that we have assumed that $c \leq \widehat\mu(x) \leq 1- c$ for some constant $c > 0$ and every $x$.
    
	\textbf{Claim 1}. To prove this claim, we modify the definitions of $p_\lambda$ and $q_\lambda$. Let $\omega_{p_\lambda}(x) = \widehat\omega(x)$, $\omega_{q_\lambda}(x) = \widehat\omega(x) + \epsilon_n \sum_{j = 1}^k\lambda_j B_j(x)$, $\mu_{p_\lambda}(x) = 1/\widehat\omega(x)$, 
	\begin{align*}
		\mu_{q_\lambda}(x) =  \frac{1}{\widehat\omega(x)} - \frac{\omega_{q_\lambda}(x) - \widehat\omega(x)}{\widehat\omega^2(x)} = \frac{1}{\widehat\omega(x)} - \frac{\epsilon_n \sum_{j = 1}^k \lambda_j B_j(x)}{\widehat\omega^2(x)},
	\end{align*}
	and $g(x) = g = \{\int \widehat\omega(x) dx\}^{-1}$. In this light, $\E_{p_\lambda}\{\mu_{p_\lambda}(X) \mid A = 1, \omega_{p_\lambda}(X) = t_1, \widehat\omega(X) = t_2, D^n\}$ and $\E_{q_\lambda}\{\mu_{q_\lambda}(X) \mid A = 1, \omega_{q_\lambda}(X) = t_1, \widehat\omega(X) = t_2, D^n\}$ are smooth functions in $t_1$ and $t_2$. Further, $\|\widehat\omega - \omega_{p_\lambda}\| = 0$ and $\|\widehat\omega - \omega_{q_\lambda}\| = \epsilon_n$, so $p_\lambda$ and $q_\lambda$ belong to $\mathcal{P}_\omega(\epsilon_n)$. 
	
	With these modifications in place, we have $\psi_{p_\lambda} = g$ and 
	\begin{align*}
		\psi_{q_\lambda} = \psi_{p_\lambda} - g \epsilon_n^2 \sum_{j = 1}^k  \int \frac{B^2_j(x)}{\widehat\omega^2(x)} dx,
	\end{align*}
	so that $|\psi_{p_\lambda} - \psi_{q_\lambda}| \ \gtrsim \epsilon_n^2$ for every $\lambda$. Following the same of reasoning to prove Proposition \ref{prop:pure_sa}, we have $\delta_1 = 0$, $\delta_2 \lesssim \epsilon_n^2$ and $\delta_3 = 0$. Therefore, for $k$ large enough, Claim 1 follows.
	
	\textbf{Claim 2}. To prove Claim 3, we modify $p_\lambda$ and $q_\lambda$. We set $\omega_{p_\lambda}(x) = \{\widehat\mu(x)\}^{-1}$, $\mu_{p_\lambda}(x) = \widehat\mu(x)$,
	\begin{align*}
		& \omega_{q_\lambda}(x) = \frac{1}{\widehat\mu(x)}  + \frac{\mu_{q_\lambda}(x) - \widehat\mu(x)}{\widehat\mu^2(x)} =  \frac{1}{\widehat\mu(x)} + \delta_n \sum_{j = 1}^k \lambda_j B_j(x), \quad \text{ and } \\
		& \mu_{q_\lambda}(x) = \widehat\mu(x) - \widehat\mu^2(x) \delta_n  \sum_{j = 1}^k \lambda_j B_j(x)
	\end{align*}
	Under this construction, both $p_\lambda$ and $q_\lambda$ are contained in $\mathcal{P}_\mu(\delta_n)$, because $\| \widehat\mu - \mu_{p_\lambda}\| = 0$, $\|\widehat\mu - \mu_{q_\lambda}\| \lesssim \delta_n$, as well as $\E_{p_\lambda}\{\omega_{p_\lambda}(X) \mid A = 1, \widehat\mu(X), \mu_{p_\lambda}(X), D^n \}= \{\widehat\mu(X)\}^{-1}$ and $\E_{q_\lambda}\{\omega_{q_\lambda}(X) \mid A = 1, \widehat\mu(X), \mu_{q_\lambda}(X), D^n\}= \{\widehat\mu(X)\}^{-1}  + \{\mu_\lambda(X) - \widehat\mu(X)\}\{\widehat\mu^2(X)\}^{-1}$, which are both smooth functions of $\widehat\mu$ and $\mu$. 
	We have
	\begin{align*}
		& \psi_{p_\lambda} = g, \quad \text{ and } \quad \psi_{q_\lambda} = \psi_{p_\lambda} - g \delta_n^2 \sum_{j = 1}^k \int \widehat\mu^2(x)B^2_j(x) dx
	\end{align*}
	where $g = \{\int 1/ \widehat\mu(x) dx\}^{-1}$ under both $q_\lambda$ and $p_\lambda$. Therefore, for any $\lambda$, $|\psi_{p_\lambda} - \psi_{q_\lambda}| \gtrsim \delta_n^2$. Because
	\begin{align*}
		& \int \omega_{p_\lambda}(x) d\omega(\lambda) = \int \omega_{q_\lambda}(x) d\omega(\lambda) = \frac{1}{\widehat\mu(x)}\\
		& \int \mu_{p_\lambda}(x) d\omega(\lambda) = \int \mu_{q_\lambda}(x) d\omega(\lambda) = \widehat\mu(x),
	\end{align*}
	we have $\delta_1 = 0$, $\delta_2 \lesssim \delta_n^2$ and $\delta_3 = 0$. In this light, for $k$ large enough, by proceeding as in the proof of Proposition \ref{prop:pure_sa}, Claim 2 follows. 
	
	\textbf{Claim 3}. Let $\gamma_n = \epsilon_n \land \delta_n$ and define
	\begin{align*}
		& \omega_{p_\lambda}(x) = \omega_{q_\lambda}(x) = \widehat\omega(x) + \gamma_n \sum_{j = 1}^k \lambda_j B_j(x), \quad \mu_{p_\lambda}(x) = \widehat\mu(x) - \gamma_n \frac{\widehat\mu(x)}{\widehat\omega(x)} \sum_{j = 1}^k \lambda_j B_j(x) \\
		& \text{and} \quad \mu_{q_\lambda}(x) = \mu_{p_\lambda}(x) + \frac{\gamma_n}{\widehat\omega(x)} \sum_{j = 1}^k \lambda_j B_j(x).
	\end{align*}
	We have $\|\widehat\omega - \omega_{p_\lambda}\| \ = \|\widehat\omega - \omega_{q_\lambda}\| \ \leq \gamma_n \leq \epsilon_n$, $\|\mu_{p_\lambda} - \widehat\mu\| \ \lesssim \gamma_n \leq \delta_n$, and $\|\mu_{q_\lambda} - \widehat\mu\| \ \lesssim \gamma_n \leq \delta_n$. Furthermore,
	\begin{align*}
		& \omega_{p_\lambda}(x) = \omega_{q_\lambda}(x) = \widehat\omega(x) + \frac{\{\widehat\mu(x) - \mu_{p_\lambda}(x)\} \widehat\omega(x)}{\widehat\mu(x)} = \widehat\omega(x) + \frac{\{\mu_{q_\lambda}(x) - \widehat\mu(x)\}\widehat\omega(x)}{1 - \widehat\mu(x)}, \\
		& \mu_{p_\lambda}(x) = \widehat\mu(x) - \frac{\widehat\mu(x)\{\omega_{p_\lambda}(x) - \widehat\omega(x)\}}{\widehat\omega(x)}, \quad \text{and} \quad \mu_{q_\lambda}(x) = \widehat\mu(x) + \frac{\{1-\widehat\mu(x)\}\{\omega_{q_\lambda}(x) - \widehat\omega(x)\}}{\widehat\omega(x)}.
	\end{align*}
	Therefore, $\E\{\omega_{p_\lambda}(X) \mid A = 1, \widehat\mu(X) = t_1, \mu_{p_\lambda}(X) = t_2, \widehat\omega(X) = t_3, D^n\}$, $\E\{\omega_{q_\lambda}(X) \mid A = 1, \widehat\mu(X) = t_1, \mu_{q_\lambda}(X) = t_2, \widehat\omega(X) = t_3, D^n\}$, $\E\{\mu_{p_\lambda}(X) \mid A = 1, \widehat\omega(X) = t_1, \omega_{p_\lambda}(X) = t_2, \widehat\mu(X) = t_3, D^n\}$, as well as $\E\{\mu_{q_\lambda}(X) \mid A = 1, \widehat\omega(X) = t_1, \omega_{q_\lambda}(X) = t_2, \widehat\mu(X) = t_3, D^n\}$ are all smooth functions. In this respect, $p_\lambda$ and $q_\lambda$ belong to $\mathcal{P}_{\omega\mu}(\epsilon_n, \delta_n)$. Set $g(X) = g = \{\int \widehat\omega(x) dx\}^{-1}$ and define the densities:
	\begin{align*}
		& p_\lambda(Y, A, X) = g(X)\{\omega_{p_\lambda}(X) - 1\}^{1-A}[\mu_{p_\lambda}(X)^Y\{1 - \mu_{p_\lambda}(X)\}^{1 - Y}]^A, \\
		& q_\lambda(Y, A, X) = g(X)\{\omega_{q_\lambda}(X) - 1\}^{1-A}[\mu_{q_\lambda}(X)^Y\{1 - \mu_{q_\lambda}(X)\}^{1 - Y}]^A.
	\end{align*}
	We have
	\begin{align*}
		& \psi_{p_\lambda} = g\int \omega_\lambda(x) \mu_{1\lambda}(x) dx = g\int \widehat\omega(x) \widehat\mu(x) dx - g \gamma^2_n \sum_{j = 1}^k \int \frac{\widehat\mu(x)}{\widehat\omega(x)} B^2_j(x) dx, \\
		& \psi_{q_\lambda} = \psi_{p_\lambda} + g \gamma^2_n \sum_{j = 1}^k \int \frac{B^2_j(x)}{\widehat\omega(x)}dx,
	\end{align*}
	so that $|\psi_{p_\lambda} - \psi_{q_\lambda}| \ \gtrsim \gamma_n^2 = \epsilon_n^2 \land \delta_n^2$ for every $\lambda$. To conclude the proof, it is sufficient to note that $\delta_1^2 \lesssim \gamma_n^2$, $\delta_2 \lesssim \gamma_n^2$, and $\delta_3 = 0$ and then follow the arguments made to establish Proposition \ref{prop:pure_sa}.  
	\section{Proofs regarding the upper bounds}
	\subsection{Useful lemmas}
	\begin{lemma}[Lemma 2 from \cite{kennedy2018sharp}] \label{lemma:edward}
		Let $\widehat{f}(o)$ be a function estimated from a sample $O^N = (O_{n+1}, \ldots, O_N)$ and let $\Pn$ denote the empirical measure over $(O_1, \ldots, O_n)$, which is independent of $O^N$. Then,
		\begin{align*}
			(\Pn - \Pb)(\widehat{f} - f) = O_\Pb\left(\frac{\|\widehat{f} - f\|}{\sqrt{n}}\right).
		\end{align*}
	\end{lemma}
	\begin{lemma}\label{lemma:Q}
		Let the setup be as in the paper and let us use the shorthand notation $(\widehat\omega_1, \widehat\mu_1) \equiv (\widehat\omega(X_1), \widehat\mu(X_1))$. Further, for a kernel $K$, bandwidth $h$, and pair of indices $s$ and $s_0$, let
		\begin{align*}
			K_{h0}(\widehat\omega_s, \widehat\mu_s) = h^{-1}K(h^{-1}(\widehat\omega_s - \widehat\omega_{s_0})) h^{-1} K(h^{-1}(\widehat\mu_s - \widehat\mu_{s_0})).
		\end{align*}
		Similarly define $K_{h0}(\widehat\omega_s)$ and $K_{h0}(\widehat\mu_s)$. Let $N$ denote some integer. For some set of indices $\mathcal{S}$ that can be partitioned into $\mathcal{S}_1$ and $\mathcal{S}_2$, define 
		\begin{align*}
			\widehat{Q}(\widehat\omega_{s_0}, \widehat\mu_{s_0}) = \frac{1}{N} \sum_{s \in \mathcal{S}} A_s K_{h0}(\widehat\omega_s, \widehat\mu_s)
		\end{align*} 
		and $Q(\widehat\omega_{s_0}, \widehat\mu_{s_0}) = \E\{AK_{h0}(\widehat\omega, \widehat\mu) \mid D^n\}$. Similarly, define $\widehat{Q}(\widehat\omega_{s_0})$, $Q(\widehat\omega_{s_0})$, $\widehat{Q}(\widehat\mu_{s_0})$ and $Q(\widehat\mu_{s_0})$. Suppose that $Q(\widehat\omega_{s_0}, \widehat\mu_{s_0})$, $Q(\widehat\omega_{s_0})$, and $Q(\widehat\mu_{s_0})$ are bounded. Let $W_{\mathcal{S}}$ be shorthand notation for $(A_s, X_s)_{s \in \mathcal{S}}$ and denote  the size of the set of indices $\mathcal{S}$ by $|\mathcal{S}|$. Then, it holds that:
		\begin{align*}
			&  \E\left[\{\widehat{Q}(\widehat\omega_{s_0}, \widehat\mu_{s_0}) - Q(\widehat\omega_{s_0}, \widehat\mu_{s_0})\}^2 \mid X_0, W_{\mathcal{S}_2}, D^n\right] \lesssim \frac{1 + |S_2|^2}{N^2h^4} + \frac{|S_1|}{N^2h^2} + \left(\frac{|\mathcal{S}|}{N} - 1\right)^2, \\
			& \E\left[\{\widehat{Q}(\widehat\omega_{s_0}) - Q(\widehat\omega_{s_0})\}^2 \mid X_0, W_{\mathcal{S}_2}, D^n\right] \lesssim \frac{1 + |S_2|^2}{(Nh)^2} + \frac{|S_1|}{N^2h} + \left(\frac{|\mathcal{S}|}{N} - 1\right)^2, \\
			& \E\left[\{\widehat{Q}(\widehat\mu_{s_0}) - Q(\widehat\mu_{s_0})\}^2 \mid X_0, W_{\mathcal{S}_2}, D^n\right] \lesssim \frac{1 + |S_2|^2}{(Nh)^2} + \frac{|S_1|}{N^2h} + \left(\frac{|\mathcal{S}|}{N} - 1\right)^2.
		\end{align*}
	\end{lemma}
	For example, assuming $nh^2 \to \infty$, if $|\mathcal{S}| = N = n$, $|\mathcal{S}_1| = n-c$ and $|\mathcal{S}_2| = c$, for some constant $c\geq 0$, then:
	\begin{align*}
		\E\left[\{\widehat{Q}(\widehat\omega_{s_0}, \widehat\mu_{s_0}) - Q(\widehat\omega_{s_0}, \widehat\mu_{s_0})\}^2 \mid W_{\mathcal{S}_2}, D^n\right] \lesssim (nh^2)^{-1}.
	\end{align*}
	\begin{proof}
		Without essential loss of generality, we derive the results for the case when $\mathcal{S}$ does not include the index $s_0$. If it does, we can consider $\mathcal{S}_{-s_0} = \mathcal{S} \setminus s_0$, define $\widehat{Q}_{-s_0}(\widehat\omega_{s_0}, \widehat\mu_{s_0}) = (N-1)^{-1} \sum_{s \in \mathcal{S}_{-s_0}} A_s K_{h0}(\widehat\omega_s, \widehat\mu_s)$, and use the bound
  \begin{align*}
      \left\{\widehat{Q}(\widehat\omega_{s_0}, \widehat\mu_{s_0}) - Q(\widehat\omega_{s_0}, \widehat\mu_{s_0})\right\}^2 & = \left\{\widehat{Q}_{-s_0}(\widehat\omega_{s_0}, \widehat\mu_{s_0}) - Q(\widehat\omega_{s_0}, \widehat\mu_{s_0}) + \left(\frac{N-1}{N} - 1\right) \widehat{Q}_{-s_0}(\widehat\omega_{s_0}, \widehat\mu_{s_0}) \right. \\ 
     &\hphantom{=} \quad\quad \left. \vphantom{\widehat{Q}_{-s_0}(\widehat\omega_{s_0}, \widehat\mu_{s_0})} \+ + \frac{A_{s_0}K_{h0}(\widehat\omega_{s_0}, \widehat\mu_{s_0})}{N} \right\}^2\\
     & \lesssim \left\{\widehat{Q}_{-s_0}(\widehat\omega_{s_0}, \widehat\mu_{s_0}) - Q(\widehat\omega_{s_0}, \widehat\mu_{s_0})\right\}^2 + \frac{1}{N^2h^4}.
  \end{align*}
  We would then use the bounding approach below on the first term. We have
		\begin{align*}
			\left\{\widehat{Q}(\widehat\omega_{s_0}, \widehat\mu_{s_0}) - Q(\widehat\omega_{s_0}, \widehat\mu_{s_0})\right\}^2 & = \left\{\frac{1}{N} \sum_{s \in \mathcal{S}} A_sK_{h0}(\widehat\omega_{s}, \widehat\mu_{s}) - Q(\widehat\omega_{s_0}, \widehat\mu_{s_0}) \right\}^2 \\
			& = \left[\frac{1}{N} \sum_{s \in \mathcal{S}} \{A_sK_{0}(\widehat\omega_{s}, \widehat\mu_{s}) - Q(\widehat\omega_{s_0}, \widehat\mu_{s_0})\} + \left(\frac{|\mathcal{S}|}{N} - 1\right)Q(\widehat\omega_{s_0}, \widehat\mu_{s_0}) \right]^2 \\
			& \lesssim \left[\frac{1}{N} \sum_{s \in \mathcal{S}} \{A_sK_{h0}(\widehat\omega_{s}, \widehat\mu_{s}) - Q(\widehat\omega_{s_0}, \widehat\mu_{s_0})\}\right]^2 + \left(\frac{|\mathcal{S}|}{N} - 1\right)^2Q^2(\widehat\omega_{s_0}, \widehat\mu_{s_0}).
		\end{align*}
		Next, we have
		\begin{align*}
			& \left[\frac{1}{N} \sum_{s \in \mathcal{S}} \{A_sK_{h0}(\widehat\omega_s, \widehat\mu_s) - Q(\widehat\omega_{s_0}, \widehat\mu_{s_0})\}\right]^2 \\
			& = \frac{1}{N^2} \mathop{\sum\sum}_{s \neq j, \ s \text{ or } j \ \in \mathcal{S}_1} \{A_sK_{h0}(\widehat\omega_s, \widehat\mu_s) - Q(\widehat\omega_{s_0}, \widehat\mu_{s_0})\}\{A_jK_{h0}(\widehat\omega_j, \widehat\mu_j) - Q(\widehat\omega_{s_0}, \widehat\mu_{s_0})\} \\
			& \hphantom{=} + \frac{1}{N^2} \mathop{\sum\sum}_{s \neq j, \ s, j \in \mathcal{S}_2} \{A_sK_{h0}(\widehat\omega_s, \widehat\mu_s) - Q(\widehat\omega_{s_0}, \widehat\mu_{s_0})\}\{A_jK_{h0}(\widehat\omega_j, \widehat\mu_j) - Q(\widehat\omega_{s_0}, \widehat\mu_{s_0})\} \\
			& \hphantom{=} + \frac{1}{N^2} \sum_{s \in \mathcal{S}_1} \{A_sK_{h0}(\widehat\omega_s, \widehat\mu_s) - Q(\widehat\omega_{s_0}, \widehat\mu_{s_0})\}^2 \\
			& \hphantom{=} + \frac{1}{N^2} \sum_{s \in \mathcal{S}_2} \{A_sK_{h0}(\widehat\omega_s, \widehat\mu_s) - Q(\widehat\omega_{s_0}, \widehat\mu_{s_0})\}^2.
		\end{align*}
		If either $s$ or $j$ is in the set $\mathcal{S}_1$, then
		\begin{align*}
			& \E\left[\{A_sK_{h0}(\widehat\omega_s, \widehat\mu_s) - Q(\widehat\omega_{s_0}, \widehat\mu_{s_0})\}\{A_jK_{h0}(\widehat\omega_j, \widehat\mu_j) - Q(\widehat\omega_{s_0}, \widehat\mu_{s_0})\} \mid X_0, W_{\mathcal{S}_2}, D^n \right] \\
			& = \E\{A_sK_{h0}(\widehat\omega_s, \widehat\mu_s) - Q(\widehat\omega_{s_0}, \widehat\mu_{s_0}) \mid X_0, W_{\mathcal{S}_2}, D^n\}\E\{A_jK_{h0}(\widehat\omega_j, \widehat\mu_j) - Q(\widehat\omega_{s_0}, \widehat\mu_{s_0}) \mid X_0, W_{\mathcal{S}_2}, D^n \} \\
			& = 0,
		\end{align*}
  since at least one of the two terms multiplying each other is equal to zero as it does not depend on $W_{\mathcal{S}_2}$. Next, we have
		\begin{align*}
			\mathop{\sum\sum}_{s \neq j, s, j \in \mathcal{S}_2} \E\left[\{A_sK_{h0}(\widehat\omega_s, \widehat\mu_s) - Q(\widehat\omega_{s_0}, \widehat\mu_{s_0})\}\{A_jK_{h0}(\widehat\omega_j, \widehat\mu_j) - Q(\widehat\omega_{s_0}, \widehat\mu_{s_0})\} \mid X_0, W_{\mathcal{S}_2}, D^n\right] \lesssim \frac{|\mathcal{S}_2|(|S_2| - 1)}{h^4}.
		\end{align*}
		In addition, we have
		\begin{align*}
			& \sum_{s \in \mathcal{S}_1} \E\left[\{A_sK_{h0}(\widehat\omega_s, \widehat\mu_s) - Q(\widehat\omega_{s_0}, \widehat\mu_{s_0})\}^2 \mid X_0, W_{\mathcal{S}_2}, D^n\right] \\
			& = \sum_{s \in \mathcal{S}_1} \E\{A_sK^2_{h0}(\widehat\omega_s, \widehat\mu_s) \mid X_0, D^n\} - Q^2(\widehat\omega_{s_0}, \widehat\mu_{s_0}) \\
			& \lesssim h^{-2}|\mathcal{S}_1|
		\end{align*}
		Finally,
		\begin{align*}
			\sum_{s \in \mathcal{S}_2} \E\left[\{A_sK_{h0}(\widehat\omega_s, \widehat\mu_s) - Q(\widehat\omega_{s_0}, \widehat\mu_{s_0})\}^2 \mid X_0, W_{\mathcal{S}_2}, D^n \right] \lesssim h^{-4}|\mathcal{S}_2|
		\end{align*}
		The proofs for the other two statements follow analogously.
	\end{proof}
	\subsection{Proof of Proposition \ref{prop:max_squared_a}: preliminaries}
	Recall the notation $\widehat\omega = \widehat\omega(X)$ and $\widehat\omega_i = \widehat\omega(X_i)$ and that
	\begin{align*}
		s_\omega(t_1, t_2; D^n) = \E(\mu - \widehat\mu \mid A = 1, \widehat\omega = t_1, \omega = t_2, D^n).
	\end{align*}
	The estimator is $\widehat\psi_\omega = \Pn \widehat\varphi - T_{n\omega}$, where $\varphi(O) = A\omega(Y - \mu) + \mu$ and
	\begin{align*}
		& T_{n\omega} = \frac{1}{n(n-1)} \mathop{\sum\sum}\limits_{1 \leq i \neq j \leq n}(A_i\widehat\omega_i - 1) \widehat{Q}^{-1}(\widehat\omega_i) K_h(\widehat\omega_j - \widehat\omega_i) A_j(Y_j - \widehat\mu_j),
	\end{align*}
	for $\widehat{Q}(\widehat\omega_i) = (n-1)^{-1} \sum_{j = 1, j \neq i}^n A_j K_{h}(\widehat\omega_j - \widehat\omega_i)$. The decomposition \eqref{eq:main_decomposition} yields that
	\begin{align*}
		\widehat\psi_\omega - \psi = (\Pn - \Pb)(\widehat\varphi - \overline\varphi) + (\Pn - \Pb) \overline\varphi + R_n - T_{n\omega}, \text{ where } R_n = \Pb(\widehat\varphi - \varphi).
	\end{align*}
	To keep the notation compact in subsequent proofs, let us also define:
	\begin{align*}
		& \widehat\kappa_\omega(Z_1, Z_2; D^n) = (A_1\widehat\omega_1 - 1)\frac{K_h(\widehat\omega_1 - \widehat\omega_2)}{\widehat{Q}(\widehat\omega_1)} A_2(Y_2 - \widehat\mu_2), \\
		& \kappa_\omega(Z_1, Z_2; D^n) = (A_1\widehat\omega_1 - 1)\frac{K_h(\widehat\omega_1 - \widehat\omega_2)}{Q(\widehat\omega_1)} A_2(Y_2 - \widehat\mu_2),
	\end{align*}
	so that
	\begin{align*}
		& T_{n\omega} = \frac{1}{n(n-1)} \mathop{\sum\sum}\limits_{1 \leq i \neq j \leq n}\widehat\kappa_\omega(Z_i, Z_j; D^n).
	\end{align*}
	Also define
	\begin{align*}
		\widetilde{T}_{n\omega} = \frac{1}{n(n-1)} \mathop{\sum\sum}\limits_{1 \leq i \neq j \leq n}\kappa_\omega(Z_i, Z_j; D^n).
	\end{align*}
	Further, let
	\begin{align*}
		& \kappa_{\omega1}(Z; D^n) = \int \kappa_\omega(Z, z; D^n) d\Pb(z) \text{ and } \kappa_{\omega2}(Z; D^n) = \int \kappa_\omega(z, Z; D^n) d\Pb(z).
	\end{align*}
	We can decompose $T_{n\omega}$ as
	\begin{align*}
		T_{n\omega} & = \widetilde{T}_{n\omega} + T_{n\omega} - \widetilde{T}_{n\omega} \\
		& = (\Pn - \Pb) \kappa_{\omega1} + (\Pn - \Pb)\kappa_{\omega2} + S_{n\omega} + \E\{\kappa_\omega(Z_1, Z_2; D^n) \mid D^n\} + T_{n\omega} - \widetilde{T}_{n\omega},
	\end{align*}
	where
	\begin{align*}
		S_{n\omega} = \frac{1}{n(n-1)}\mathop{\sum\sum}_{1 \leq i \neq j \leq n} \left[ \kappa_\omega(Z_i, Z_j; D^n) - \kappa_{\omega1}(Z_i; D^n) - \kappa_{\omega2}(Z_j; D^n)+ \E\left\{ \kappa_\omega(Z_1, Z_2; D^n) \mid D^n \right\}\right].
	\end{align*}
	We will repeatedly use the following two lemmas:
	\begin{lemma}\label{lemma:Sna}
		It holds that
		\begin{enumerate}
			\item $\sup_z |\kappa_{\omega1}(z; D^n)| \ \lesssim \|\widehat\mu - \mu\|_\infty \ \land \ h^{-1/2} \|\widehat\mu - \mu\|$,
			\item $\sup_z|\kappa_{\omega2}(z; D^n)| \ \lesssim \|\widehat\omega - \omega\|_\infty \ \land \ h^{-1/2}\|\widehat\omega - \omega\|$,
			\item $\E(S^2_{n\omega} \mid D^n) \lesssim (n^2h)^{-1}$.
		\end{enumerate}
	\end{lemma}
	\begin{lemma}\label{lemma:Tna}
		It holds that
		\begin{align*}
			\E\{(\widetilde{T}_{n\omega} - T_{n\omega})^2 \mid D^n\} & \lesssim (n^3h)^{-1/2} + \|\widehat\omega - \omega\|^2\|\widehat\mu-\mu\|^2(nh)^{-1}.
		\end{align*}
	\end{lemma}
	\begin{lemma}\label{lemma:RTna}
		It holds that
		\begin{align*}
			\left| \E\left\{ \kappa_\omega(Z_1, Z_2; D^n) \mid D^n \right\} - R_n \right| \lesssim \|\widehat\omega-\omega\|^{1+\alpha} \ + \ h^\alpha \|\widehat\omega - \omega\|,
		\end{align*}
	\end{lemma}
	\subsubsection{Bias of $\widehat\psi_\omega$ \label{section:bias_psi_a}}
	The bias bound follows exactly as described in the main text. In particular, 
	\begin{align*}
		\left| \E(\widehat\psi_\omega - \psi \mid D^n) \right| = \left| \E(R_n - T_{n\omega} \mid D^n) \right|.
	\end{align*}
	We briefly report the calculation here for convenience. From the decomposition
	\begin{align*}
		A_j(Y_j - \widehat\mu_j) = A_j s_\omega(\widehat\omega_i, \omega_i; D^n) + A_j\{ s_\omega(\widehat\omega_j, \omega_j; D^n) - s_\omega(\widehat\omega_i, \omega_i; D^n)\} + A_j \epsilon_j,
	\end{align*}
	we obtain that
	\begin{align*}
		T_{n\omega} & = \frac{1}{n}\sum_{i = 1}^n (A_i \widehat\omega_i-1)s_\omega(\widehat\omega_i, \omega_i; D^n) \\
		& \hphantom{=} + \frac{1}{n(n-1)}\mathop{\sum\sum}\limits_{1\leq i \neq j \leq n} (A_i\widehat\omega_i - 1) \widehat{Q}^{-1}(\widehat\omega_i) K_h(\widehat\omega_j - \widehat\omega_i) A_j\{ s_\omega(\widehat\omega_j, \omega_j; D^n) - s_\omega(\widehat\omega_i, \omega_i; D^n) + \epsilon_j\}.
	\end{align*}
	Next, notice that
	\begin{align*}
		\E\left\{ \frac{1}{n}\sum_{i = 1}^n (A_i \widehat\omega_i-1)s_\omega(\widehat\omega_i, \omega_i; D^n) \mid D^n \right\} = R_n
	\end{align*}
	and that
	\begin{align*}
		\E(A_j \epsilon_j \mid A_i, A_j, \widehat\omega_j, \widehat\omega_i, D^n) & = \E(A_j \epsilon_j \mid A_j, \widehat\omega_j, D^n) \\
		& = \E\{\E(A_j \epsilon_j \mid A_j, \widehat\omega_j, \omega_j, D^n) \mid A_j, \widehat\omega_j, D^n\} \\
		& = \E(\E[A_j \{(Y_j - \widehat\mu_j) -  s_\omega(\widehat\omega_j, \omega_j; D^n)\} \mid A_j, \widehat\omega_j, \omega_j, D^n] \mid A_j, \widehat\omega_j, D^n) \\
		& = \E(\E[A_j \{(\mu_j - \widehat\mu_j) -  s_\omega(\widehat\omega_j, \omega_j; D^n)\} \mid A_j, \widehat\omega_j, \omega_j, D^n] \mid A_j, \widehat\omega_j, D^n) \\
		& = 0
	\end{align*}
	Therefore,
	\begin{align*}
		& \E(R_n - T_{n\omega} \mid D^n) \\
		& = \frac{1}{n(n-1)}\mathop{\sum\sum}\limits_{1\leq i \neq j \leq n} \E\left[ (A_i\widehat\omega_i - 1) \widehat{Q}^{-1}(\widehat\omega_i) K_h(\widehat\omega_j - \widehat\omega_i) A_j\{ s_\omega(\widehat\omega_j, \omega_j; D^n) - s_\omega(\widehat\omega_i, \omega_i; D^n)\} \mid D^n \right]
	\end{align*}
	Under the assumption that $\widehat{Q}(\widehat\omega_i)$ is bounded away from zero, we have
	\begin{align*}
		& \left| \E(R_n - T_{n\omega} \mid D^n) \right| \\
		& \lesssim \frac{1}{n(n-1)}\mathop{\sum\sum}\limits_{1\leq i \neq j \leq n} \E\left\{A_i|\widehat\omega_i - \omega_i|A_jK_h(\widehat\omega_j - \widehat\omega_i) |s_\omega(\widehat\omega_j, \omega_j; D^n) - s_\omega(\widehat\omega_i, \omega_i; D^n)| \mid D^n \right\}
	\end{align*}
	By the H\"{o}lder smoothness assumption,
	\begin{align*}
		|s_\omega(\widehat\omega_j, \omega_j; D^n) - s_\omega(\widehat\omega_i, \omega_i; D^n)| \ & \lesssim \{(\widehat\omega_j - \widehat\omega_i)^2 + (\omega_j - \widehat\omega_j)^2 + (\widehat\omega_i - \omega_i)^2\}^{\alpha / 2} \\
		& \lesssim |\widehat\omega_j - \widehat\omega_i|^\alpha \ + \ |\widehat\omega_j - \omega_j|^\alpha \ + \ |\widehat\omega_i - \omega_i|^\alpha
	\end{align*}
	We bound each of the 3 terms as follows. First, it holds that
	\begin{align*}
		\E\left\{ A_i |\widehat\omega_i - \omega_i|A_jK_h(\widehat\omega_j - \widehat\omega_i)|\widehat\omega_j - \widehat\omega_i|^\alpha \mid D^n \right\} & \leq h^\alpha \E\left[ |\widehat\omega_i - \omega_i|\E\{A_jK_h(\widehat\omega_j - \widehat\omega_i) \mid X_i, D^n\} \mid D^n \right] \\
		& \lesssim h^\alpha \|\widehat\omega - \omega\|,
	\end{align*}
	under the assumption that $\E\{A_jK_h(\widehat\omega_j - \widehat\omega_i) \mid X_i, D^n\} \lesssim 1$ and because $K(u)$ has support in $[-1, 1]$. The last equality follows because $\int |\widehat\omega(x) - \omega(x)|d\Pb(x) \leq \|\widehat\omega - \omega\|$. Next, by the Cauchy-Schwarz inequality:
	\begin{align*}
		& \left[\E\left\{ A_i|\widehat\omega_i - \omega_i|A_jK_h(\widehat\omega_j - \widehat\omega_i)|\widehat\omega_j - \omega_j|^\alpha \mid D^n\right\}\right]^2 \\
		& \leq \E\left[ (\widehat\omega_i - \omega_i)^2\E\{A_jK_h(\widehat\omega_j - \widehat\omega_i) \mid X_i, D^n\} \mid D^n \right] \E\left[ (\widehat\omega_j - \omega_j)^{2\alpha}\E\{A_iK_h(\widehat\omega_j - \widehat\omega_i) \mid X_j, D^n\} \mid D^n \right] \\
		& \lesssim \E\left\{(\widehat\omega_i - \omega_i)^2 \mid D^n \right\} \E\left\{ (\widehat\omega_j - \omega_j)^{2\alpha} \mid D^n \right\} \\
		& \lesssim \|\widehat\omega - \omega\|^{2 + 2\alpha}
	\end{align*}
	where the last inequality follows by Jensen's inequality because $|x|^\alpha$ is concave for $\alpha < 1$. Similarly,
	\begin{align*}
		\E\left\{ (\widehat\omega_i - \omega_i)^{1+\alpha}A_jK_h(\widehat\omega_j - \widehat\omega_i)\mid D^n\right\} & = \E\left[(\widehat\omega_i - \omega_i)^{1+\alpha}\E\{A_jK_h(\widehat\omega_j - \widehat\omega_i) \mid D^n, X_i\} \mid D^n\right]  \\
		& \lesssim \|\widehat\omega - \omega\|^{1+\alpha}.
	\end{align*}
	Thus, we conclude that
	\begin{align*}
		\left|\E(R_n - T_{n\omega} \mid D^n)\right| \lesssim \|\widehat\omega-\omega\|^{1+\alpha} \ + \ h^\alpha \|\widehat\omega - \omega\|.
	\end{align*}
	\subsubsection{Variance of $\widehat\psi_\omega$\label{sec:variance_psi_a}}
	First, by the boundedness of the observations and nuisance functions, we have
	\begin{align*}
		\var(\widehat\psi_\omega \mid D^n) \lesssim n^{-1} + \var(T_{n\omega} \mid D^n).
	\end{align*}
	We have
	\begin{align*}
		T_{n\omega} = (\Pn - \Pb) \kappa_{\omega1} + (\Pn - \Pb) \kappa_{\omega2} + S_{n\omega} + \E\{\kappa(Z_1, Z_2; D^n) \mid D^n\} + T_{n\omega} - \widetilde{T}_{n\omega},
	\end{align*}
	\begin{align*}
		S_{n\omega} = \frac{1}{n(n-1)}\mathop{\sum\sum}_{1 \leq i \neq j \leq n} \left[ \kappa_\omega(Z_i, Z_j; D^n) - \kappa_{\omega1}(Z_i; D^n) - \kappa_{\omega2}(Z_j; D^n)+ \E\left\{ \kappa_\omega(Z_1, Z_2; D^n) \mid D^n \right\}\right].
	\end{align*}
	In this way, we can further bound
	\begin{align*}
		\var(T_{n\omega} \mid D^n) \lesssim n^{-1} \E\{\kappa^2_{\omega1}(Z) \mid D^n\} + n^{-1} \E\{\kappa^2_{\omega2}(Z) \mid D^n\} + \E(S^2_{n\omega} \mid D^n) + \E\{(T_{n\omega} - \widetilde{T}_{n\omega})^2 \mid D^n\}. 
	\end{align*}
	The bound on the variance follows then from the Lemmas \ref{lemma:Tna} and \ref{lemma:Sna}.
	\subsubsection{Linear expansion of $\widehat\psi_\omega - \psi$\label{sec:proof_prop_max_squared_a_le}}
	Recall that we define $\widehat\varphi_\omega = \E(\mu - \widehat\mu \mid A = 1, \widehat\omega, D^n)(A\widehat\omega - 1)$ and $\overline\varphi_\omega = \E(\mu - \overline\mu \mid A = 1, \overline\omega)(A\overline\omega - 1)$. 
	
	We have the following decomposition
	\begin{align*}
		T_{n\omega} & = R_n + \E\left\{ \kappa_\omega(Z_1, Z_2; D^n) \mid D^n \right\} - R_n  + T_{n\omega} - \widetilde{T}_{n\omega} \\
		& \hphantom{=} + (\Pn - \Pb)(\kappa_{\omega1} - \widehat\varphi_\omega) + (\Pn - \Pb)(\widehat\varphi_\omega - \overline\varphi_\omega) + (\Pn - \Pb)\overline\varphi_\omega + (\Pn - \Pb) \kappa_{\omega2} + S_{n\omega}.
	\end{align*}
	The third statement of Proposition \ref{prop:max_squared_a} is implied by Lemmas \ref{lemma:Sna}, \ref{lemma:Tna}, and \ref{lemma:RTna} as well as the following statement:
	\begin{align} \label{eq:statement2a}
		& \| \kappa_{\omega1} - \widehat\varphi_\omega\| \ \lesssim (\|\widehat\omega - \omega\|^\alpha \ + \ h^\alpha) \land \|\widehat\mu - \mu\|_\infty.
	\end{align}
	To see why this is the case. Let $C$ denote the constant such that the inequality in Lemma \ref{lemma:Tna} holds. For any $\epsilon > 0$, let $M^2_\epsilon = C / \epsilon$. We have
	\begin{align*}
		& \Pb\left( \frac{|T_{n\omega} - \widetilde{T}_{n\omega}|}{(n^3h)^{-1/4} + \|\widehat\omega - \omega\| \|\widehat\mu - \mu\| (nh)^{-1/2}} > M_\epsilon\right) \\
		& = \E\left[\Pb\left( (T_{n\omega} - \widetilde{T}_{n\omega})^2 > \frac{C}{\epsilon} \cdot \left\{(n^3h)^{-1/2} + \|\widehat\omega - \omega\|^2 \|\widehat\mu - \mu\|^2 (nh)^{-1}\right\} \mid D^n\right) \right] \\
		& \leq \E\left[\frac{\epsilon}{C} \cdot \frac{\E\{(T_{n\omega} - \widetilde{T}_{n\omega})^2 \mid D^n\}}{(n^3h)^{-1/2} + \|\widehat\omega - \omega\|^2 \|\widehat\mu - \mu\|^2 (nh)^{-1}} \right] \\
		& \leq \epsilon,
	\end{align*}
	where the first inequality follows by Markov's and the second one by the definition of $C$ from Lemma \ref{lemma:Tna}. In this respect, we conclude that $T_{n\omega} - \widetilde{T}_{n\omega} = O_\Pb((n^3h)^{-1/4} + \|\widehat\omega - \omega\|\|\widehat\mu-\mu\|(nh)^{-1/2})$. A similar reasoning yields that Lemma \ref{lemma:Sna} implies that
	\begin{align*}
		S_{n\omega} = O_\Pb((n^2h)^{-1/2}) = O_\Pb((n^3h)^{-1/4}) O_\Pb((nh)^{-1/4}) = o_\Pb((n^3h)^{-1/4}), \text{ if } nh \to \infty.
	\end{align*}
	Therefore,
	\begin{align*}
		S_{n\omega} + T_{n\omega} - \widetilde{T}_{n\omega} = O_\Pb\left((n^3h)^{-1/4} + \|\widehat\omega - \omega\|\|\widehat\mu-\mu\|(nh)^{-1/2}\right).
	\end{align*}
	By Lemma \ref{lemma:edward}, Eq. \eqref{eq:statement2a} implies that
	\begin{align*}
		(\Pn - \Pb)(\kappa_{\omega1} - \widehat\varphi_\omega) = O_\Pb\left(n^{-1/2} (\|\widehat\omega - \omega\|^\alpha \ + \ h^\alpha) \right)
	\end{align*}
	By Lemmas \ref{lemma:edward} and \ref{lemma:Sna}, $(\Pn - \Pb)\kappa_{\omega2} = O_\Pb(n^{-1/2} \|\widehat\omega - \omega\|_\infty \ \land \ (nh)^{-1/2} \|\widehat\omega - \omega\|)$. Finally, by Lemma \ref{lemma:edward}, $(\Pn - \Pb)(\widehat\varphi_\omega - \overline\varphi_\omega) = o_\Pb(n^{-1/2})$ as long as $\|\widehat\varphi_\omega - \overline\varphi_\omega\| = o_\Pb(1)$. 
	\subsubsection{Proof of Lemma \ref{lemma:RTna}}
	This bound follows exactly as in the proof for the bias of $\widehat\psi_\omega$ (Section \ref{section:bias_psi_a}). 
	\subsubsection{Proof of the bound in Eq. \eqref{eq:statement2a}}
	We have
	\begin{align*}
		\int \kappa_\omega(Z_1, z_2; D^n) d\Pb(z_2) & = (A_1\widehat\omega_1 - 1)\E\left\{ A_2\frac{K_h(\widehat\omega_1 - \widehat\omega_2)}{Q(\widehat\omega_1)}(\mu_2 - \widehat\mu_2) \mid \widehat\omega_1, D^n\right\} \\
		& = (A_1\widehat\omega_1 - 1)\E\left\{ A_2\frac{K_h(\widehat\omega_1 - \widehat\omega_2)}{Q(\widehat\omega_1)} \E(\mu_2 - \widehat\mu_2 \mid A_2 = 1, \widehat\omega_2, D^n) \mid \widehat\omega_1, D^n\right\} \\
		& = \widehat\varphi_\omega(Z_1; D^n) \\
		& \hphantom{=} + (A_1\widehat\omega_1 - 1) \E\left[ A_2\frac{K_h(\widehat\omega_1 - \widehat\omega_2)}{Q(\widehat\omega_1)}\{s_\omega(\widehat\omega_2, \omega_2; D^n) - s_\omega(\widehat\omega_1, \omega_1; D^n)\} \mid \widehat\omega_1, D^n\right]
	\end{align*}
	By the smoothness assumption on $s_\omega(t_1, t_2; D^n)$, we have
	\begin{align*}
		& \left| \E\left[ A_2\frac{K_h(\widehat\omega_1 - \widehat\omega_2)}{Q(\widehat\omega_1)}\{s_\omega(\widehat\omega_2, \omega_2; D^n) - s_\omega(\widehat\omega_1, \omega_1; D^n)\} \mid X_1, D^n\right] \right| \lesssim (\|\widehat\omega - \omega\|_\infty^\alpha \ + \ h^\alpha) \land \|\widehat\mu - \mu\|_\infty
	\end{align*}
 since $|s_\omega(t_1, t_2; D^n)| \ \leq \|\widehat\mu - \mu\|_\infty$. In this light, 
	\begin{align*}
		\left| \kappa_{\omega1}(Z; D^n) - \widehat\varphi_\omega(Z; D^n)\right| \lesssim |A\widehat\omega - 1| \{(\|\widehat\omega - \omega\|_\infty^\alpha \ + \ h^\alpha) \land \|\widehat\mu - \mu\|_\infty\}
	\end{align*}
	so that $\|\kappa_{\omega1} - \widehat\varphi_\omega\| \ \lesssim (\|\widehat\omega - \omega\|_\infty^\alpha \ + \ h^\alpha) \ \land \ \|\widehat\mu - \mu\|_\infty$. 
	\subsubsection{Proof of Lemma \ref{lemma:Tna}}
	We have
	\begin{align*}
		\widetilde{T}_{n\omega} - T_{n\omega} & = \frac{1}{n(n-1)} \mathop{\sum\sum}\limits_{1\leq i \neq j \leq n} (A_i\widehat\omega_i - 1)Q^{-1}(\widehat\omega_i)\widehat{Q}^{-1}(\widehat\omega_i)\{\widehat{Q}(\widehat\omega_i) - Q(\widehat\omega_i)\}K_h(\widehat\omega_j - \widehat\omega_i)A_j(Y_j - \widehat\mu_j)  \\
		& \equiv \frac{1}{n(n-1)} \mathop{\sum\sum}\limits_{1\leq i \neq j \leq n}  T_{ij}
	\end{align*}
	We can break the square of the double sum as a sum of seven different terms:
	\begin{align*}
		\left(\mathop{\sum\sum}\limits_{1\leq i \neq j \leq n} T_{ij}\right)^2 & = \mathop{\sum\sum}\limits_{1 \leq i \neq j \leq n} (T^2_{ij} + T_{ij} T_{ji}) + \mathop{\sum\sum\sum}\limits_{1 \leq i \neq j \neq l \leq n} (T_{ij} T_{il} + T_{ij} T_{li} + T_{ij} T_{lj} + T_{ij}T_{jl}) \\
		& \hphantom{=} + \mathop{\sum\sum\sum\sum}\limits_{1 \leq i \neq j \neq l \neq k \leq n} T_{ij} T_{lk}
	\end{align*}
	and bound the expectation of each term separately. In particular,
	\begin{align*}
		& T^2_{ij} \lesssim h^{-1} A_j K_{h}(\widehat\omega_j - \widehat\omega_i) \quad \text{ and } \quad |T_{ij}T_{ji}| \ \lesssim h^{-1} A_iA_jK_{h}(\widehat\omega_j - \widehat\omega_i),
	\end{align*}
	so that, in light of $\E\{A_jK_h(\widehat\omega_j - \widehat\omega_i) \mid X_i, D^n\} \lesssim 1$:
	\begin{align*}
		\E\left( \frac{1}{n^2(n-1)^2} \mathop{\sum\sum}\limits_{1 \leq i \neq j \leq n} (T^2_{ij} + T_{ij} T_{ji}) \mid D^n \right) \lesssim (n^2h)^{-1}.
	\end{align*}
	Next, we have
	\begin{align*}
		& |T_{ij}T_{li}| \ \lesssim A_iA_jK_{h}(\widehat\omega_j- \widehat\omega_i) K_{h}(\widehat\omega_i - \widehat\omega_l)|\widehat{Q}(\widehat\omega_i) - Q(\widehat\omega_i)| |\widehat{Q}(\widehat\omega_l) - Q(\widehat\omega_l)|, \\
		& |T_{ij}T_{il}| \ \lesssim A_jA_lK_{h}(\widehat\omega_j- \widehat\omega_i) K_{h}(\widehat\omega_l - \widehat\omega_i)\{\widehat{Q}(\widehat\omega_i) - Q(\widehat\omega_i)\}^2, \\
		& |T_{ij}T_{jl}| \ \lesssim A_jA_lK_h(\widehat\omega_j - \widehat\omega_i) K_h(\widehat\omega_l - \widehat\omega_j)|\widehat{Q}(\widehat\omega_i) - Q(\widehat\omega_i)| |\widehat{Q}(\widehat\omega_j) - Q(\widehat\omega_j)|. 
	\end{align*}
	Applying Lemma \ref{lemma:Q} with $\mathcal{S} = \{1, \ldots, n\} \setminus i$ and $\mathcal{S}_2 = \{j, l\}$, we have
	\begin{align*}
		\E\left[\left\{\widehat{Q}(\widehat\omega_i) - Q(\widehat\omega_i)\right\}^2 \mid A_i, A_j, A_l, X_i, X_j, X_l, D^n\right] \lesssim (nh)^{-1}.
	\end{align*}
	Similarly, by the Cauchy-Schwarz inequality:
	\begin{align*}
		&  \E\left\{\left|\widehat{Q}(\widehat\omega_i) - Q(\widehat\omega_i)\right| \left|\widehat{Q}(\widehat\omega_l) - Q(\widehat\omega_l)\right|\mid A_i, A_j, A_l, X_i, X_j, X_l, D^n\right\} \lesssim (nh)^{-1}, \\
		& \E\left\{\left|\widehat{Q}(\widehat\omega_i) - Q(\widehat\omega_i)\right| \left|\widehat{Q}(\widehat\omega_j) - Q(\widehat\omega_j)\right|\mid A_i, A_j, A_l, X_i, X_j, X_l, D^n\right\} \lesssim (nh)^{-1}.
	\end{align*}
	In addition, we have
	\begin{align*}
		\E\{A_iA_jK_{h}(\widehat\omega_j- \widehat\omega_i) K_{h}(\widehat\omega_i - \widehat\omega_l) \mid D^n\} & = \E[A_iK_{h}(\widehat\omega_i- \widehat\omega_l) \E\{A_jK_{h}(\widehat\omega_j - \widehat\omega_i) \mid X_i, D^n\} \mid D^n] \\
		& \lesssim \E[A_iK_{h}(\widehat\omega_i- \widehat\omega_l) \mid D^n] \\
		& \lesssim 1,
	\end{align*}
	so that we conclude that
	\begin{align*}
		\E\left\{ \frac{1}{n^2(n-1)^2} \mathop{\sum\sum\sum}\limits_{1 \leq i \neq j \neq l \leq n} (T_{ij} T_{il} + T_{ij} T_{li} + T_{ij} T_{jl}) \mid D^n \right\} \lesssim (n^2h)^{-1}.
	\end{align*}
	Finally, to analyze the terms of the form $T_{ij}T_{lk}$ and $T_{ij}T_{lj}$, we need to address the fact that $\widehat{Q}(\widehat\omega_i)$ and $\widehat{Q}(\widehat\omega_l)$ contain the random variables $A_l$ and $A_i$, respectively. This adds a minor complication to the bound because $\E\{(A_i\widehat\omega_i - 1)\widehat{Q}^{-1}(\widehat\omega_l) \mid D^n\} \neq \E\{A_i(\widehat\omega_i - \omega_i) \widehat{Q}^{-1}(\widehat\omega_l) \mid D^n\}$, for example. In this light, we write
	\begin{align*}
		\widehat{Q}(\widehat\omega_i) & = \frac{1}{n-1}\sum_{s = 1, s \neq i}^n A_s K_h(\widehat\omega_s - \widehat\omega_i) \\
		& = \frac{1}{n-1}\sum_{s = 1, s \neq (i, l)}^n A_s K_h(\widehat\omega_s - \widehat\omega_i) + \frac{A_l K_h(\widehat\omega_l -\widehat\omega_i)}{n-1} \\
		& \equiv \widehat{Q}_{-l}(\widehat\omega_i) + \frac{A_l K_h(\widehat\omega_l -\widehat\omega_i)}{n-1}.
	\end{align*}
	This way, one has
	\begin{align*}
		& \left| \widehat{Q}^{-1}(\widehat\omega_i) - \widehat{Q}^{-1}_{-l}(\widehat\omega_i)\right| = \left|\frac{A_l K_h(\widehat\omega_l -\widehat\omega_i)}{(n-1)\widehat{Q}(\widehat\omega_i) \widehat{Q}_{-l}(\widehat\omega_i)}\right| \lesssim \frac{A_lK_h(\widehat\omega_l - \widehat\omega_i)}{n-1}, \\
		& \left| \widehat{Q}^{-1}(\widehat\omega_l) - \widehat{Q}^{-1}_{-i}(\widehat\omega_l)\right| = \left|\frac{A_i K_h(\widehat\omega_i -\widehat\omega_l)}{(n-1)\widehat{Q}(\widehat\omega_l) \widehat{Q}_{-i}(\widehat\omega_l)}\right| \lesssim \frac{A_iK_h(\widehat\omega_i - \widehat\omega_l)}{n-1}.
	\end{align*}
	Next, we write
	\begin{align*}
		T_{ij}T_{lk} & = \vphantom{\frac{A_iK_h(\widehat\omega_i - \widehat\omega_l)K_h(\widehat\omega_l - \widehat\omega_k)}{(n-1)\widehat{Q}(\widehat\omega_l)\widehat{Q}_{-i}(\widehat\omega_l)}} 
 (A_i\widehat\omega_i - 1)\{\widehat{Q}^{-1}_{-l}(\widehat\omega_i) - Q^{-1}(\widehat\omega_i)\}K_h(\widehat\omega_j - \widehat\omega_i)A_j(Y_j - \widehat\mu_j) \\
		& \hphantom{=} \quad\quad \times (A_l\widehat\omega_l - 1)\{\widehat{Q}^{-1}_{-i}(\widehat\omega_l) - Q^{-1}(\widehat\omega_l)\}K_h(\widehat\omega_l - \widehat\omega_k)A_k(Y_k - \widehat\mu_k) \vphantom{\frac{A_iK_h(\widehat\omega_i - \widehat\omega_l)K_h(\widehat\omega_l - \widehat\omega_k)}{(n-1)\widehat{Q}(\widehat\omega_l)\widehat{Q}_{-i}(\widehat\omega_l)}} 
 \\
		& \hphantom{=} - \vphantom{\frac{A_iK_h(\widehat\omega_i - \widehat\omega_l)K_h(\widehat\omega_l - \widehat\omega_k)}{(n-1)\widehat{Q}(\widehat\omega_l)\widehat{Q}_{-i}(\widehat\omega_l)}} 
 (A_i\widehat\omega_i - 1)\frac{A_lK_h(\widehat\omega_l - \widehat\omega_i)K_h(\widehat\omega_j - \widehat\omega_i)}{(n-1)\widehat{Q}(\widehat\omega_i)\widehat{Q}_{-l}(\widehat\omega_i)}A_j(Y_j - \widehat\mu_j)  \\
		& \hphantom{=} \quad\quad  \times (A_l\widehat\omega_l - 1)\{\widehat{Q}^{-1}(\widehat\omega_l) - Q^{-1}(\widehat\omega_l)\}K_h(\widehat\omega_l - \widehat\omega_k)A_k(Y_k - \widehat\mu_k) \\
		& \hphantom{=} - \vphantom{\frac{A_iK_h(\widehat\omega_i - \widehat\omega_l)K_h(\widehat\omega_l - \widehat\omega_k)}{(n-1)\widehat{Q}(\widehat\omega_l)\widehat{Q}_{-i}(\widehat\omega_l)}} (A_i\widehat\omega_i - 1)\{\widehat{Q}^{-1}_{-l}(\widehat\omega_i) - Q^{-1}(\widehat\omega_i)\}K_h(\widehat\omega_j - \widehat\omega_i)A_j(Y_j - \widehat\mu_j)  \\
		& \hphantom{=} \quad\quad \times (A_l\widehat\omega_l - 1)\frac{A_iK_h(\widehat\omega_i - \widehat\omega_l)K_h(\widehat\omega_l - \widehat\omega_k)}{(n-1)\widehat{Q}(\widehat\omega_l)\widehat{Q}_{-i}(\widehat\omega_l)}A_k(Y_k - \widehat\mu_k) 
	\end{align*}
	When $j = k$, the expectation of the last two terms can be upper bounded as
	\begin{align*}
		& \E \left[ \left| (A_i\widehat\omega_i - 1)(A_l\widehat\omega_l - 1)(Y_j - \widehat\mu_j)^2\frac{A_jA_lK_h(\widehat\omega_l - \widehat\omega_i)K_h(\widehat\omega_j - \widehat\omega_i)K_h(\widehat\omega_l - \widehat\omega_j)}{(n-1)\widehat{Q}(\widehat\omega_i)\widehat{Q}_{-l}(\widehat\omega_i) \widehat{Q}(\widehat\omega_l)Q(\widehat\omega_l)}\{Q(\widehat\omega_l) - \widehat{Q}(\widehat\omega_l)\}  \right. \right| \\
		& \hphantom{=} \left.\quad + \left| (A_i\widehat\omega_i - 1)(A_l\widehat\omega_l - 1)(Y_j - \widehat\mu_j)^2 \frac{A_jA_iK_h(\widehat\omega_j - \widehat\omega_i)K_h(\widehat\omega_i - \widehat\omega_l)K_h(\widehat\omega_l - \widehat\omega_j)}{(n-1)\widehat{Q}(\widehat\omega_l)\widehat{Q}_{-i}(\widehat\omega_l)\widehat{Q}_{-l}(\widehat\omega_i)Q^(\widehat\omega_i)}\{Q(\widehat\omega_i) - \widehat{Q}_{-l}(\widehat\omega_i)\}\right| \mid D^n \right] \\
		& \lesssim \E\left\{ \frac{A_jA_l K_h(\widehat\omega_l - \widehat\omega_i)K_h(\widehat\omega_j - \widehat\omega_i)}{h(n-1)} \E\left\{ \left| \widehat{Q}(\widehat\omega_l) - Q(\widehat\omega_l) \right| \mid A_j, X_j, X_l, X_i, D^n \right\} \mid D^n \right\} \\
		& \hphantom{\lesssim} + \E\left\{ \frac{A_jA_iK_h(\widehat\omega_l - \widehat\omega_i)K_h(\widehat\omega_j - \widehat\omega_i)}{h(n-1)} \E\left\{ \left| \widehat{Q}_{-l}(\widehat\omega_i) - Q(\widehat\omega_i) \right| \mid A_j, X_i, X_j, D^n \right\} \mid D^n \right\} \\
		& \lesssim (nh)^{-3/2}
	\end{align*}
	The last inequality follows because
	\begin{align*}
		\E\left\{ A_jA_l K_h(\widehat\omega_l - \widehat\omega_i)K_h(\widehat\omega_j - \widehat\omega_i)\mid D^n \right\} & \lesssim \E\left[A_jA_lK_h(\widehat\omega_l - \widehat\omega_i)K_h(\widehat\omega_j - \widehat\omega_i)\mid D^n \right] \\
		& = \E\left[A_lK_h(\widehat\omega_l - \widehat\omega_i)\E\{A_jK_h(\widehat\omega_j - \widehat\omega_i) \mid X_i, D^n\} \mid D^n \right] \\
		& \lesssim \E\left\{A_lK_h(\widehat\omega_l - \widehat\omega_i) \mid D^n\right\} \\
		& = \E\left[\E\left\{A_lK_h(\widehat\omega_l - \widehat\omega_i) \mid X_i, D^n\right\}\mid D^n \right] \\
		& \lesssim 1.
	\end{align*}
	and because, by Lemma \ref{lemma:Q},
	\begin{align*}
		& \E\left[ \left\{ \widehat{Q}(\widehat\omega_l) - Q(\widehat\omega_l) \right\}^2 \mid A_j, X_j, X_l, X_i, D^n \right] \lesssim (nh)^{-1},\\
		& \E\left[ \left\{ \widehat{Q}_{-l}(\widehat\omega_i) - Q(\widehat\omega_i) \right\}^2 \mid A_j, X_i, X_j, D^n \right] \lesssim (nh)^{-1}.
	\end{align*}
	On the other hand, when $j \neq k$:
	\begin{align*}
		& \E \left[ \left| (A_i\widehat\omega_i - 1)\frac{A_lK_h(\widehat\omega_l - \widehat\omega_i)K_h(\widehat\omega_j - \widehat\omega_i)}{(n-1)\widehat{Q}(\widehat\omega_i)\widehat{Q}_{-l}(\widehat\omega_i)}A_j(Y_j - \widehat\mu_j) \right. \right. \\
		& \hphantom{=} \quad\quad \left. \times (A_l\widehat\omega_l - 1)\{\widehat{Q}^{-1}(\widehat\omega_l) - Q^{-1}(\widehat\omega_l)\}K_h(\widehat\omega_l - \widehat\omega_k)A_k(Y_k - \widehat\mu_k) \vphantom{\frac{A_iK_h(\widehat\omega_i - \widehat\omega_l)K_h(\widehat\omega_l - \widehat\omega_k)}{(n-1)\widehat{Q}(\widehat\omega_l)\widehat{Q}_{-i}(\widehat\omega_l)}} \right| \\
		& \hphantom{=} \quad + \left| \vphantom{\frac{A_iK_h(\widehat\omega_i - \widehat\omega_l)K_h(\widehat\omega_l - \widehat\omega_k)}{(n-1)\widehat{Q}(\widehat\omega_l)\widehat{Q}_{-i}(\widehat\omega_l)}} 
 (A_i\widehat\omega_i - 1)\{\widehat{Q}^{-1}_{-l}(\widehat\omega_i) - Q^{-1}(\widehat\omega_i)\}K_h(\widehat\omega_j - \widehat\omega_i)A_j(Y_j - \widehat\mu_j) \right. \\
		& \hphantom{=} \left. \left. \quad\quad\quad \times (A_l\widehat\omega_l - 1)\frac{A_iK_h(\widehat\omega_i - \widehat\omega_l)K_h(\widehat\omega_l - \widehat\omega_k)}{(n-1)\widehat{Q}(\widehat\omega_l)\widehat{Q}_{-i}(\widehat\omega_l)}A_k(Y_k - \widehat\mu_k)\right| \mid D^n \right] \\
		& \lesssim \E\left\{ \frac{A_jA_lA_k K_h(\widehat\omega_l - \widehat\omega_i)K_h(\widehat\omega_j - \widehat\omega_i)K_h(\widehat\omega_l - \widehat\omega_k)}{(n-1)} \E\left\{ \left| \widehat{Q}(\widehat\omega_l) - Q(\widehat\omega_l) \right| \mid A_j, A_k, X_j, X_l, X_i, X_k, D^n \right\} \mid D^n \right\} \\
		& \hphantom{\lesssim} + \E\left\{ \frac{A_jA_iA_kK_h(\widehat\omega_l - \widehat\omega_i)K_h(\widehat\omega_j - \widehat\omega_i)K_h(\widehat\omega_l - \widehat\omega_k)}{(n-1)} \E\left\{ \left| \widehat{Q}_{-l}(\widehat\omega_i) - Q(\widehat\omega_i) \right| \mid A_j, A_k, X_i, X_j, X_k, D^n \right\} \mid D^n \right\} \\
		& \lesssim (n^3h)^{-1/2}
	\end{align*}
	The last inequality follows because
	\begin{align*}
		& \E\left\{ A_jA_kA_l K_h(\widehat\omega_l - \widehat\omega_i)K_h(\widehat\omega_j - \widehat\omega_i)K_h(\widehat\omega_l - \widehat\omega_k) \mid D^n \right\} \\
		& = \E\left[A_jA_l K_h(\widehat\omega_l - \widehat\omega_i)K_h(\widehat\omega_j - \widehat\omega_i)\E\{A_kK_h(\widehat\omega_l - \widehat\omega_k) \mid X_l, D^n\} \mid D^n \right] \\
		& \lesssim \E\left\{A_j A_lK_h(\widehat\omega_l - \widehat\omega_i)K_h(\widehat\omega_j - \widehat\omega_i) \mid D^n \right\} \\
		& = \E\left[A_lK_h(\widehat\omega_l - \widehat\omega_i)\E\{A_jK_h(\widehat\omega_j - \widehat\omega_i) \mid X_i, D^n\} \mid D^n \right] \\
		& \lesssim \E\left\{A_lK_h(\widehat\omega_l - \widehat\omega_i) \mid D^n\right\} \\
		& = \E\left[\E\left\{A_lK_h(\widehat\omega_l - \widehat\omega_i) \mid X_i, D^n\right\}\mid D^n \right] \\
		& \lesssim 1.
	\end{align*}
	and by Lemma \ref{lemma:Q}. 
	
	Next, we proceed to bound the first term appearing in $T_{ij}T_{lk}$, which is the main term. 
	
	When $k = j$, we have 
	\begin{align*}
		& \left| \E\left[(A_i\widehat\omega_i - 1)\{\widehat{Q}^{-1}_{-l}(\widehat\omega_i) - Q^{-1}(\widehat\omega_i)\}K_h(\widehat\omega_j - \widehat\omega_i)A_j(Y_j - \widehat\mu_j)^2 \right. \right. \\
		& \left. \left. \quad\quad \times (A_l\widehat\omega_l - 1)\{\widehat{Q}^{-1}_{-i}(\widehat\omega_l) - Q^{-1}(\widehat\omega_l)\}K_h(\widehat\omega_l - \widehat\omega_j)\mid D^n \right] \right| \\
		& = \left| \E\left[A_i(\widehat\omega_i - \omega_i)\{\widehat{Q}^{-1}_{-l}(\widehat\omega_i) - Q^{-1}(\widehat\omega_i)\}K_h(\widehat\omega_j - \widehat\omega_i)A_j(Y_j - \widehat\mu_j)^2 \right. \right. \\
		& \left. \left. \quad\quad \times A_l(\widehat\omega_l - \omega_l)\{\widehat{Q}^{-1}_{-i}(\widehat\omega_l) - Q^{-1}(\widehat\omega_l)\}K_h(\widehat\omega_l - \widehat\omega_j)\mid D^n \right] \right| \\
		& \lesssim \left| \E\left[A_iA_lK_h(\widehat\omega_j - \widehat\omega_i)K_h(\widehat\omega_l - \widehat\omega_j)\E\left\{|\widehat{Q}_{-l}(\widehat\omega_i) - Q(\widehat\omega_i)| | \widehat{Q}_{-i}(\widehat\omega_l) - Q(\widehat\omega_l)| \mid X_i, X_j, X_l, D^n \right\} \right] \mid D^n \right| \\
		& \lesssim (nh)^{-1}.
	\end{align*}
	In this light, 
	\begin{align*}
		\E\left\{\frac{1}{n^2(n-1)^2} \mathop{\sum\sum\sum}\limits_{1 \leq i \neq j \neq l \leq n} T_{ij}T_{lj} \mid D^n \right\} \lesssim (n^2h)^{-1}.
	\end{align*}
	Next, for $k \neq j$, we have
	\begin{align*}
		& \left| \E\left[(A_i\widehat\omega_i - 1)\{\widehat{Q}^{-1}_{-l}(\widehat\omega_i) - Q^{-1}(\widehat\omega_i)\}K_h(\widehat\omega_j - \widehat\omega_i)A_j(Y_j - \widehat\mu_j) \right. \right. \\
		& \left. \left. \quad\quad \times (A_l\widehat\omega_l - 1)\{\widehat{Q}^{-1}_{-i}(\widehat\omega_l) - Q^{-1}(\widehat\omega_l)\}K_h(\widehat\omega_l - \widehat\omega_k)A_k(Y_k - \widehat\mu_k) \mid D^n \right] \right| \\
		& = \left| \E\left[A_i(\widehat\omega_i - \omega_i)\{\widehat{Q}^{-1}_{-l}(\widehat\omega_i) - Q^{-1}(\widehat\omega_i)\}K_h(\widehat\omega_j - \widehat\omega_i)A_j(\mu_j - \widehat\mu_j) \right. \right. \\
		& \left. \left. \quad\quad \times A_l(\widehat\omega_l - \omega_l)\{\widehat{Q}^{-1}_{-i}(\widehat\omega_l) - Q^{-1}(\widehat\omega_l)\}K_h(\widehat\omega_l - \widehat\omega_k)A_k(\mu_k - \widehat\mu_k) \mid D^n \right] \right| \\
		& \lesssim \left| \E\left[|\widehat\omega_i - \omega_i||\mu_j - \widehat\mu_j||\widehat\omega_l - \omega_l||\mu_k - \widehat\mu_k|A_iA_jA_kA_lK_h(\widehat\omega_j - \widehat\omega_i)K_h(\widehat\omega_l - \widehat\omega_k) \right. \right. \\
		& \left. \left. \quad\quad\quad \times  \E\left\{|\widehat{Q}_{-l}(\widehat\omega_i) - Q(\widehat\omega_i)| | \widehat{Q}_{-i}(\widehat\omega_l) - Q(\widehat\omega_l)| \mid A_j, A_k, X_i, X_j, X_l, X_k, D^n \right\} \right] \mid D^n \right| \\
		& \lesssim (nh)^{-1} \E\left[|\widehat\omega_i - \omega_i| |\widehat\omega_l - \omega_l| 
		A_jA_iA_kA_lK_h(\widehat\omega_j - \widehat\omega_i)K_h(\widehat\omega_k - \widehat\omega_l)|\widehat\mu_j - \mu_j||\widehat\mu_k - \mu_k| \mid D^n \right] \\
		& \lesssim (nh)^{-1} \|\widehat\omega - \omega\|^2\|\widehat\mu - \mu\|^2.
	\end{align*}
	where the second-to-last inequality follows by Lemma \ref{lemma:Q} and the Cauchy-Schwarz inequality, while the last one follows by the Cauchy-Schwarz inequality and the independence of the observations because:
	\begin{align*}
		& \left\{\E\left(|\widehat\omega_i - \omega_i|
		A_iA_jK_h(\widehat\omega_j - \widehat\omega_i)|\widehat\mu_j - \mu_j| \mid D^n\right)\right\}^2 \\
		& \leq \E\left[(\widehat\omega_i - \omega_i)^2
		\E\{A_jK_h(\widehat\omega_j - \widehat\omega_i) \mid X_i, D^n\} \mid D^n \right] \E\left[(\widehat\mu_j - \mu_j)^2  \E\{A_iK_h(\widehat\omega_j - \widehat\omega_i) \mid X_j, D^n\} \mid D^n \right] \\
		& \lesssim \|\widehat\omega - \omega\|^2\|\widehat\mu - \mu\|^2.
	\end{align*}
	We have thus reached:
	\begin{align*}
		& \left| \E(T_{ij}T_{lk} \mid D^n) \right| \lesssim \|\widehat\omega - \omega\|^2\|\widehat\mu - \mu\|^2(nh)^{-1} + (n^3h)^{-1/2},
	\end{align*}
	implying that
	\begin{align*}
		\E\left(\frac{1}{n^2(n-1)^2} \mathop{\sum\sum\sum\sum}\limits_{1 \leq i \neq j \neq l \neq k \leq n} T_{ij} T_{lk} \mid D^n\right) \lesssim \|\widehat\omega - \omega\|^2 \|\widehat\mu- \mu\|^2(nh)^{-1} + (n^3h)^{-1/2}
	\end{align*}
	Putting everything together, because $nh \to \infty$, we have that
	\begin{align*}
		\E\{(\widetilde{T}_n - T_n)^2 \mid D^n\} \lesssim (n^3h)^{-1/2} + \|\widehat\omega - \omega\|^2\|\widehat\mu - \mu\|^2(nh)^{-1}.
	\end{align*}
	This concludes our proof of Lemma \ref{lemma:Tna}.
	\subsubsection{Proof of Lemma \ref{lemma:Sna}}
	We have that
	\begin{align*}
		\left|  \kappa_{\omega1}(Z_1; D^n) \right| & = \left| \int \kappa_\omega(Z_1, z_2; D^n) d\Pb(z_2)\right| \\
		& = \left| (A_1\widehat\omega_1 - 1) \E\left\{\frac{K_h(\widehat\omega_1 - \widehat\omega_2)}{Q(\widehat\omega_1)}A_2(\mu_2 - \widehat\mu_2) \mid X_1, D^n)\right\}\right| \\
		& \lesssim \|\widehat\mu - \mu\|_\infty \E\left\{A_2\frac{K_h(\widehat\omega_1 - \widehat\omega_2)}{Q(\widehat\omega_1)} \mid X_1, D^n \right\} \\
		& =\|\widehat\mu - \mu\|_\infty 
	\end{align*}
	and, by Cauchy-Schwarz, that
	\begin{align*}
		\left|\kappa_{\omega1}(Z_1; D^n) \right| & = \left| (A_1\widehat\omega_1 - 1)\E\left\{\frac{K_h(\widehat\omega_1 - \widehat\omega_2)}{Q(\widehat\omega_1)}A_2(\mu_2 - \widehat\mu_2) \mid X_1, D^n\right\}\right| \\
		&  \leq \left| (A_1\widehat\omega_1 - 1) \right| \left[\E\left\{\frac{K^2_h(\widehat\omega_1 - \widehat\omega_2)}{Q^2(\widehat\omega_1)} \mid X_1, D^n \right\} \right]^{1/2} \left[\E\left\{(\mu_2 - \widehat\mu_2)^2 \mid X_1, D^n\right\}\right]^{1/2} \\
		& \lesssim h^{-1/2} \|\widehat\mu - \mu\|
	\end{align*}
	Further, we have that
	\begin{align*}
		\left|\kappa_{\omega2}(Z_2; D^n) \right| & = \left| \int \kappa_\omega(z_1, Z_2; D^n) d\Pb(z_1)\right| \\
		& = \left| \E\{A_1(\widehat\omega_1 - \omega_1) \frac{K_h(\widehat\omega_1 - \widehat\omega_2)}{Q(\widehat\omega_1)} \mid X_2, D^n)A_2(Y_2 - \widehat\mu_2)\right| \\
		& \lesssim \|\widehat\omega - \omega\|_\infty
	\end{align*}
	and, by Cauchy-Schwarz, that
	\begin{align*}
		\left|\kappa_{\omega2}(Z_2; D^n) \right| & = \left| \E\left\{A_1(\widehat\omega_1 - \omega_1) \frac{K_h(\widehat\omega_1 - \widehat\omega_2)}{Q(\widehat\omega_1)} \mid X_2, D^n\right\}A_2(Y_2 - \widehat\mu_2)\right| \\
		& \leq \left[\E\{(\widehat\omega_1 - \omega_1)^2 \mid X_2, D^n\}\right]^{1/2} \left[\E\left\{\frac{A_1K^2_h(\widehat\omega_1 - \widehat\omega_2)}{Q^2(\widehat\omega_1)} \mid X_2, D^n\right\}\right]^{1/2} \left|A_2(Y_2 - \widehat\mu_2)\right| \\
		& \lesssim h^{-1/2} \|\widehat\omega - \omega\|
	\end{align*}
	Finally, recall that
	\begin{align*}
		S_{n\omega} = \frac{1}{n(n-1)}\mathop{\sum\sum}_{1 \leq i \neq j \leq n} \left[ \kappa_\omega(Z_i, Z_j; D^n) - \kappa_{\omega1}(Z_i; D^n) - \kappa_{\omega2}(Z_j; D^n)+ \E\left\{ \kappa_\omega(Z_1, Z_2; D^n) \mid D^n \right\}\right],
	\end{align*}
	and notice that
	\begin{align*}
		& \E\left[\kappa_\omega(Z_i, Z_j; D^n) - \kappa_{\omega1}(Z_i; D^n) - \kappa_{\omega2}(Z_j; D^n)+ \E\left\{ \kappa_\omega(Z_1, Z_2; D^n) \mid D^n \right\} \mid Z_i, D^n\right] \\
		& = \E\left[\kappa_\omega(Z_i, Z_j; D^n) - \kappa_{\omega1}(Z_i; D^n) - \kappa_{\omega2}(Z_j; D^n)+ \E\left\{ \kappa_\omega(Z_1, Z_2; D^n) \mid D^n \right\} \mid Z_j, D^n\right] \\
		& = 0.
	\end{align*}
	Therefore,
	\begin{align*}
		& \E(S^2_{n\omega} \mid D^n)  \lesssim \frac{1}{n^2} \E\left\{\kappa^2_\omega(Z_1, Z_2; D^n) \mid D^n \right\}  \lesssim (n^2h)^{-1}.
	\end{align*}
	\subsection{Proof of Proposition \ref{prop:max_squared_b}: preliminaries \label{sec:proof_prop_max_squared_b}}
	The proof essentially follows that of Proposition \ref{prop:max_squared_a}, so we omit certain details. Recall
	\begin{align*}
		s_\mu(t_1, t_2; D^n) = \E(\widehat\omega - \omega \mid A = 1, \widehat\mu = t_1, \mu = t_2, D^n).
	\end{align*}
	The estimator is $\widehat\psi_\mu = \Pn \widehat\varphi - T_{n\mu}$, where
	\begin{align*}
		& T_{n\mu} = \frac{1}{n(n-1)} \mathop{\sum\sum}\limits_{1 \leq i \neq j \leq n}(A_i\widehat\omega_i - 1) \widehat{Q}^{-1}_{-i}(\widehat\mu_j) K_h(\widehat\mu_j - \widehat\mu_i) A_j(Y_j - \widehat\mu_j), \text{for} \\
		& \widehat{Q}_{-i}(\widehat\mu_j) = \frac{1}{n-1} \sum_{s = 1, s \neq i}^n A_sK_h(\widehat\mu_s - \widehat\mu_j).
	\end{align*}
	The decomposition \eqref{eq:main_decomposition} yields that
	\begin{align*}
		\widehat\psi_\mu - \psi = (\Pn - \Pb)(\widehat\varphi - \overline\varphi) + (\Pn - \Pb) \overline\varphi + R_n - T_{n\mu}, \text{ where } R_n = \Pb(\widehat\varphi - \varphi).
	\end{align*}
	To keep the notation compact, let us define:
	\begin{align*}
		& \widehat\kappa_\mu(Z_1, Z_2; D^n) = (A_1\widehat\omega_1 - 1)\frac{K_h(\widehat\mu_1 - \widehat\mu_2)}{\widehat{Q}_{-1}(\widehat\mu_2)} A_2(Y_2 - \widehat\mu_2), \\
		& \kappa_\mu(Z_1, Z_2; D^n) = (A_1\widehat\omega_1 - 1)\frac{K_h(\widehat\mu_1 - \widehat\mu_2)}{Q(\widehat\mu_2)} A_2(Y_2 - \widehat\mu_2),
	\end{align*}
	so that
	\begin{align*}
		& T_{n\mu} = \frac{1}{n(n-1)} \mathop{\sum\sum}\limits_{1 \leq i \neq j \leq n}\widehat\kappa_\mu(Z_i, Z_j; D^n).
	\end{align*}
	Also, define
	\begin{align*}
		\widetilde{T}_{n\mu} = \frac{1}{n(n-1)} \mathop{\sum\sum}\limits_{1 \leq i \neq j \leq n}\kappa_\mu(Z_i, Z_j; D^n).
	\end{align*}
	Further, let
	\begin{align*}
		& \kappa_{\mu1}(Z; D^n) = \int \kappa_\mu(Z, z; D^n) d\Pb(z) \quad \text{ and } \quad \kappa_{\mu2}(Z; D^n) = \int \kappa_\mu(z, Z; D^n) d\Pb(z). 
	\end{align*}
	We can decompose $T_{n\mu}$ as
	\begin{align*}
		T_{n\mu} & = \widetilde{T}_{n\mu} + T_{n\mu} - \widetilde{T}_{n\mu} \\
		& = (\Pn - \Pb)\kappa_{\mu1}+ (\Pn - \Pb) \kappa_{\mu2} + S_{n\mu} + \E\{\kappa_\mu(Z_1, Z_2; D^n) \mid D^n\} + T_{n\mu} - \widetilde{T}_{n\mu},
	\end{align*}
	where
	\begin{align*}
		S_{n\mu} = \frac{1}{n(n-1)}\mathop{\sum\sum}_{1 \leq i \neq j \leq n} \left[ \kappa_\mu(Z_i, Z_j; D^n) - \kappa_{\mu1}(Z_i; D^n) - \kappa_{\mu2}(Z_j; D^n)+ \E\left\{ \kappa_\mu(Z_1, Z_2; D^n) \mid D^n \right\}\right].
	\end{align*}
	We will repeatedly use the following lemmas:
	\begin{lemma}\label{lemma:Snb}
		It holds that
		\begin{enumerate}
			\item $\sup_z|\kappa_{\mu1}(z; D^n)| \ \lesssim \|\widehat\mu - \mu\|_\infty \ \land \ h^{-1/2} \|\widehat\mu - \mu\|$,
			\item $\sup_z|\kappa_{\mu2}(z; D^n)| \ \lesssim \|\widehat\omega - \omega\|_\infty \ \land \ h^{-1/2} \|\widehat\omega - \omega\|$,
			\item $\E(S^2_{n\mu} \mid D^n) \lesssim (n^2h)^{-1}$.
		\end{enumerate}
	\end{lemma}
	\begin{lemma}\label{lemma:Tnb}
		It holds that
		\begin{align*}
			\E\{(\widetilde{T}_{n\mu} - T_{n\mu})^2 \mid D^n\} \lesssim (n^3h)^{-1/2} + \|\widehat\omega - \omega\|^2\|\widehat\mu - \mu\|^2(nh)^{-1}.
		\end{align*}
	\end{lemma}
	\begin{lemma}\label{lemma:RTnb}
		It holds that
		\begin{align*}
			\left|R_n - \E(\widetilde{T}_n \mid D^n) \right| = \left| R_n - \E\{\kappa_\mu(Z_1, Z_2; D^n) \mid D^n\} \right| \lesssim (h^\beta \|\widehat\mu - \mu\| \ + \ \|\widehat\mu - \mu\|^{1 + \beta}) \land \|\widehat\omega - \omega\|\|\widehat\mu - \mu\|.
		\end{align*}
	\end{lemma}
	\subsubsection{Bias of $\widehat\psi_\mu$}
	In light of Lemma \ref{lemma:RTnb} and the decomposition in Eq. \eqref{eq:main_decomposition}, the bound on the bias of $\widehat\psi_\mu$ follows after bounding $\left|\E(T_{n\mu} - \widetilde{T}_{n\mu} \mid D^n)\right|$. We have
	\begin{align*}
		& \left| \E(T_{n\mu}  - \widetilde{T}_{n\mu}\mid D^n) \right| \\
		& = \frac{1}{n(n-1)} \left| \mathop{\sum\sum}_{1 \leq i \neq j \leq n}\E\left[(A_i\widehat\omega_i-1)\{\widehat{Q}_{-i}^{-1}(\widehat\mu_j) - Q^{-1}(\widehat\mu_j)\}K_h(\widehat\mu_i - \widehat\mu_j)A_j(Y_j - \widehat\mu_j) \mid D^n \right]  \right|.
	\end{align*}
	Because $\widehat{Q}_{-i}(\widehat\mu_j)$ does not depend on $A_i$, we have
	\begin{align*}
		& \left| \E\left[(A_i\widehat\omega_i-1)\{\widehat{Q}_{-i}^{-1}(\widehat\mu_j) - Q^{-1}(\widehat\mu_j)\}K_h(\widehat\mu_i - \widehat\mu_j)A_j(Y_j - \widehat\mu_j) \mid D^n \right] \right| \\
		& = \left| \E\left[A_i(\widehat\omega_i-\omega_i)\{\widehat{Q}_{-i}^{-1}(\widehat\mu_j) - Q^{-1}(\widehat\mu_j)\}K_h(\widehat\mu_i - \widehat\mu_j)A_j(\mu_j - \widehat\mu_j) \mid D^n \right] \right| \\
		& \lesssim \left| \E\left[|\widehat\omega_i-\omega_i|\E\{|\widehat{Q}_{-i}(\widehat\mu_j) - Q(\widehat\mu_j)| \mid A_j, X_j, D^n\}A_iA_jK_h(\widehat\mu_i - \widehat\mu_j)|\mu_j - \widehat\mu_j| \mid D^n \right] \right|
	\end{align*}
	By Lemma \ref{lemma:Q}, we have
	\begin{align*}
		\E\{|\widehat{Q}_{-i}(\widehat\mu_j) - Q(\widehat\mu_j)| \mid A_j, X_j, D^n\} \lesssim (nh)^{-1/2}.
	\end{align*}
	In this light, we have
	\begin{align*}
		\left| \E\left[(A_i\widehat\omega_i-1)\{\widehat{Q}_{-i}^{-1}(\widehat\mu_j) - Q^{-1}(\widehat\mu_j)\}K_h(\widehat\mu_i - \widehat\mu_j)A_j(Y_j - \widehat\mu_j) \mid D^n \right] \right| \lesssim \|\widehat\omega - \omega\| \|\widehat\mu - \mu\|(nh)^{-1/2}.
	\end{align*}
	Putting everything together, we have reached that
	\begin{align*}
		\left|\E(\widehat\psi_\mu - \psi \mid D^n)\right| & \lesssim \left|\E(R_n - \widetilde{T}_{n\mu} \mid D^n)\right|  + \left|\E(\widetilde{T}_{n\mu} - T_{n\mu} \mid D^n)\right| \\
		& \lesssim h^\beta \|\widehat\mu - \mu\| \ + \ \|\widehat\mu - \mu\|^{1 + \beta} \ + \ \|\widehat\mu - \mu\| \|\widehat\omega - \omega\|(nh)^{-1/2}.
	\end{align*}
	Finally, we also have the bound
	\begin{align*}
		\left|\E(\widehat\psi_\mu - \psi \mid D^n)\right| \lesssim |R_n| \ + \ |\E(T_{n\mu}\mid D^n)| \ \lesssim \|\widehat\omega - \omega\| \|\widehat\mu - \mu\| \ + \ |\E(T_n\mid D^n)| 
	\end{align*}
	and
	\begin{align*}
		\left| \E(T_{n\mu}\mid D^n) \right| & \leq \frac{1}{n(n-1)} \left| \mathop{\sum\sum}\limits_{1 \leq i \neq j \leq n} \E\left\{ (A_i\widehat\omega_i - 1)\widehat{Q}^{-1}_{-i}(\widehat\mu_j)K_h(\widehat\mu_i - \widehat\mu_j)A_j(Y_j - \widehat\mu_j) \mid D^n\right\} \right| \\
		& \lesssim \|\widehat\omega - \omega\|\|\widehat\mu - \mu\|.
	\end{align*}
	This concludes our derivation of the bound on the bias of $\widehat\psi_\mu$ given the training sample $D^n$. 
	\subsubsection{Variance of $\widehat\psi_\mu$}
	The bound on the variance follows from the same reasoning used to bound the variance of $\widehat\psi_\omega$ in Section \ref{sec:variance_psi_a}. In particular, it follows from Lemmas \ref{lemma:Snb} and \ref{lemma:Tnb}.
	\subsubsection{Linear expansion of $\widehat\psi_\mu - \psi_\mu$}
	We have defined $\widehat\varphi_\mu = \E(\widehat\omega - \omega \mid A = 1, \widehat\mu, D^n)A(Y - \widehat\mu)$ and $\overline\varphi_\mu = \E(\overline\omega - \omega \mid A = 1, \overline\mu)A(Y - \overline\mu)$. Following the reasoning in Section \ref{sec:proof_prop_max_squared_a_le}, the third statement of Proposition \ref{prop:max_squared_b} is implied by Lemmas \ref{lemma:Tnb}, \ref{lemma:Snb} and \ref{lemma:RTnb} as well as the following statement:
	\begin{align}\label{eq:statement3b}
		\|\widehat\kappa_{\mu2} - \widehat\varphi_\mu\| \lesssim (\|\widehat\mu - \mu\|_\infty^\beta \ + \ h^\beta) \land \|\widehat\omega - \omega\|_\infty.
	\end{align}
	\subsubsection{Proof Lemma \ref{lemma:RTnb}}
	We have
	\begin{align*}
		& \E\{\kappa_\mu(Z_1, Z_2; D^n) \mid D^n\} = \E\left\{ (A_1\widehat\omega_1 - 1)\{Q(\widehat\mu_2)\}^{-1} K_h(\widehat\mu_1 - \widehat\mu_2)A_2(Y_2 - \widehat\mu_2) \mid D^n \right\} \\
		& = \E\left\{ A_1s_\mu(\widehat\mu_2, \mu_2; D^n)\{Q(\widehat\mu_2)\}^{-1} K_h(\widehat\mu_1 - \widehat\mu_2)A_2(Y_2 - \widehat\mu_2) \mid D^n \right\} \\
		& \hphantom{=} + \E\left[ A_1\{s_\mu(\widehat\mu_1, \mu_1; D^n) - s_\mu(\widehat\mu_2, \mu_2; D^n)\}K_h(\widehat\mu_1 - \widehat\mu_2)\{Q(\widehat\mu_2)\}^{-1} A_2(Y_2 - \widehat\mu_2) \mid D^n \right] \\
		& = R_n + \E\left[ A_1\{s_\mu(\widehat\mu_1, \mu_1; D^n) - s_\mu(\widehat\mu_2, \mu_2; D^n)\}K_h(\widehat\mu_1 - \widehat\mu_2)\{Q(\widehat\mu_2)\}^{-1}A_2(\mu_2 - \widehat\mu_2) \mid D^n \right].
	\end{align*}
	The second equality follows because
	\begin{align*}
		& \E\{A_1(\widehat\omega_1 - \omega_1) - A_1s_\mu(\widehat\mu_1, \mu_1; D^n) \mid A_1, A_2, Y_2, \widehat\mu_1, \widehat\mu_2, D^n\} \\
		& = \E\{A_1(\widehat\omega_1 - \omega_1) - A_1s_\mu(\widehat\mu_1, \mu_1; D^n) \mid A_1, \widehat\mu_1, D^n\} \\
		& = \E\left[\E\{A_1(\widehat\omega_1 - \omega_1) \mid A_1, \widehat\mu_1, \mu_1, D^n\} - A_1s_\mu(\widehat\mu_1, \mu_1; D^n) \mid A_1, \widehat\mu_1, D^n\right] \\
		& = \E\{A_1s_\mu(\widehat\mu_1, \mu_1; D^n) - A_1s_\mu(\widehat\mu_1, \mu_1; D^n) \mid A_1, \widehat\mu_1, D^n\}\\
		& = 0.
	\end{align*}
	The last equality follows because
	\begin{align*}
		Q(\widehat\mu_2) = \E\{A_1 K_h(\widehat\mu_1 - \widehat\mu_2) \mid D^n, X_2\} \quad \text{ and } \quad R_n = \E\left\{s_\mu(\widehat\mu_2, \mu_2; D^n)A_2(Y_2 - \widehat\mu_2) \mid D^n \right\}.
	\end{align*}
	Under the smoothness condition of $s_\mu(t_1, t_2; D^n)$, one also has
	\begin{align*}
		& \left| \E\left[ A_1\{s_\mu(\widehat\mu_1, \mu_1; D^n) - s_\mu(\widehat\mu_2, \mu_2; D^n)\}\frac{K_h(\widehat\mu_1 - \widehat\mu_2)}{Q(\widehat\mu_2)} A_2(\mu_2 - \widehat\mu_2) \mid D^n \right] \right| \\
		& \lesssim \E\left[\{|\widehat\mu_2 - \widehat\mu_1|^\beta \ + \ |\widehat\mu_2 - \mu_2|^\beta \ + \ |\widehat\mu_1 - \mu_1|^\beta\}A_1A_2K_h(\widehat\mu_1 - \widehat\mu_2)(\mu_2 - \widehat\mu_2) \mid D^n \right] \\
		& \lesssim h^\beta \|\widehat\mu - \mu\| \ + \ \|\widehat\mu - \mu\|^{1 + \beta}.
	\end{align*}
	Finally, we also have the bound
	\begin{align*}
		\left| \E\{\kappa_\mu(Z_1, Z_2; D^n) \mid D^n\} \right| & \leq \left|\E\left\{ (A_1\widehat\omega_1 - 1)Q^{-1}(\widehat\mu_2)K_h(\widehat\mu_1 - \widehat\mu_2)A_2(Y_2 - \widehat\mu_2) \mid D^n\right\} \right| \\
		& \lesssim \|\widehat\omega - \omega\|\|\widehat\mu - \mu\|.   
	\end{align*}
	\subsubsection{Proof of Lemma \ref{lemma:Tnb}}
	Let
	\begin{align*}
		T_{ij} & = (A_i\widehat\omega_i - 1)\{\widehat{Q}_{-i}^{-1}(\widehat\mu_j) - Q^{-1}(\widehat\mu_j)\}K_{h}(\widehat\mu_i - \widehat\mu_j)A_j(Y_j - \widehat\mu_j) \\
		& = (A_i\widehat\omega_i - 1)\widehat{Q}^{-1}_{-i}(\widehat\mu_j)Q^{-1}(\widehat\mu_j)\{Q(\widehat\mu_j) - \widehat{Q}_{-i}(\widehat\mu_j)\}K_{h}(\widehat\mu_i - \widehat\mu_j)A_j(Y_j - \widehat\mu_j)
	\end{align*}
	and notice that
	\begin{align*}
		&(T_{n\mu} - \widetilde{T}_{n\mu})^2 =
		\left(\mathop{\sum\sum}\limits_{1\leq i \neq j \leq n} T_{ij}\right)^2 \\
		& = \mathop{\sum\sum}\limits_{1 \leq i \neq j \leq n} (T^2_{ij} + T_{ij} T_{ji}) + \mathop{\sum\sum\sum}\limits_{1 \leq i \neq j \neq l \leq n} (T_{ij} T_{il} + T_{ij} T_{li} + T_{ij} T_{jl} + T_{ij} T_{lj}) + \mathop{\sum\sum\sum\sum}\limits_{1 \leq i \neq j \neq l \neq k \leq n} T_{ij} T_{lk}.
	\end{align*}
	Just like in the proof of the bound in Eq. \eqref{eq:statement2a}, we have
	\begin{align*}
		& T_{ij}^2 \lesssim A_jh^{-1}K_{h}(\widehat\mu_i - \widehat\mu_j), \quad |T_{ij}T_{ji}| \ \lesssim h^{-1}A_iA_jK_{h}(\widehat\mu_i - \widehat\mu_j) \\
		& |T_{ij}T_{il}| \ \lesssim A_jA_lK_{h}(\widehat\mu_i - \widehat\mu_j)K_{h}(\widehat\mu_i - \widehat\mu_l)\{Q(\widehat\mu_j) - \widehat{Q}_{-i}(\widehat\mu_j)\}\{Q(\widehat\mu_l) - \widehat{Q}_{-i}(\widehat\mu_l)\}, \\
		& |T_{ij}T_{li}| \ \lesssim A_iA_jK_{h}(\widehat\mu_i - \widehat\mu_j)K_{h}(\widehat\mu_l - \widehat\mu_i)\{Q(\widehat\mu_j) - \widehat{Q}_{-i}(\widehat\mu_j)\}\{Q(\widehat\mu_i) - \widehat{Q}_{-l}(\widehat\mu_i)\} \\
		& |T_{ij}T_{jl}| \ \lesssim A_jA_lK_{h}(\widehat\mu_i - \widehat\mu_j)K_{h}(\widehat\mu_j - \widehat\mu_l)\{Q(\widehat\mu_j) - \widehat{Q}_{-i}(\widehat\mu_j)\}\{Q(\widehat\mu_l) - \widehat{Q}_{-j}(\widehat\mu_l)\}, \\
		& |T_{ij}T_{lj}| \ \lesssim A_jK_{h}(\widehat\mu_i - \widehat\mu_j)K_{h}(\widehat\mu_l - \widehat\mu_j)\{Q(\widehat\mu_j) - \widehat{Q}_{-i}(\widehat\mu_j)\}\{Q(\widehat\mu_j) - \widehat{Q}_{-l}(\widehat\mu_j)\}.
	\end{align*}
	By Lemma \ref{lemma:Q}, we have
	\begin{align*}
		\E\{|\widehat{Q}_{-i}(\widehat\mu_j) - Q(\widehat\mu_j)| \mid A_l, X_j, X_l, D^n\} \lesssim (nh)^{-1/2}.
	\end{align*}
	This yields that
	\begin{align*}
		& \left| \frac{1}{n^2(n-1)^2}\E\left\{\mathop{\sum\sum}\limits_{1 \leq i \neq j \leq n} (T^2_{ij} + T_{ij} T_{ji}) \mid D^n \right\} \right| \lesssim (n^2h)^{-1}, \\
		& \left| \frac{1}{n^2(n-1)^2}\E\left\{\mathop{\sum\sum\sum}\limits_{1 \leq i \neq j \neq l \leq n} (T_{ij} T_{il} + T_{ij} T_{li} + T_{ij} T_{jl}) \mid D^n \right\} \right| \lesssim (n^2h)^{-1}.
	\end{align*}
	Next, define 
	\begin{align*}
		\widehat{Q}_{-il}(\widehat\mu_j) = \frac{1}{n-1}\sum_{s = 1, s \neq (i, l)}^n A_sK_h(\widehat\mu_s - \mu_j),
	\end{align*}
	so that
	\begin{align*}
		\widehat{Q}_{-il}(\widehat\mu_j) - \widehat{Q}_{-i}(\widehat\mu_j) = - \frac{A_lK_h(\widehat\mu_l - \widehat\mu_j)}{n-1}.
	\end{align*}
	In this light, we have
	\begin{align*}
		T_{ij}T_{lk} & = (A_i\widehat\omega_i - 1)\{\widehat{Q}_{-il}^{-1}(\widehat\mu_j) - Q^{-1}(\widehat\mu_j)\}K_{h}(\widehat\mu_i - \widehat\mu_j)A_j(Y_j - \widehat\mu_j) \\
		& \hphantom{=} \quad \times (A_l\widehat\omega_l - 1)\{\widehat{Q}_{-il}^{-1}(\widehat\mu_k) - Q^{-1}(\widehat\mu_k)\}K_{h}(\widehat\mu_l - \widehat\mu_k)A_k(Y_k - \widehat\mu_k) \\
		& \hphantom{=} - (A_i\widehat\omega_i - 1)\frac{A_lK_h(\widehat\mu_l - \widehat\mu_j)K_h(\widehat\mu_i - \widehat\mu_j)}{(n-1)\widehat{Q}_{-il}(\widehat\mu_j)\widehat{Q}_{-i}(\widehat\mu_j)}A_j(Y_j - \widehat\mu_j) \\
		& \hphantom{= - } \quad \times (A_l\widehat\omega_l - 1)\{\widehat{Q}^{-1}_{-l}(\widehat\mu_k) - Q^{-1}(\widehat\mu_k)\}K_{h}(\widehat\mu_l - \widehat\mu_k)A_k(Y_k - \widehat\mu_k) \\
		& \hphantom{=} - (A_i\widehat\omega_i - 1)\{\widehat{Q}_{-il}^{-1}(\widehat\mu_j) - Q^{-1}(\widehat\mu_j)\}K_{h}(\widehat\mu_i - \widehat\mu_j)A_j(Y_j - \widehat\mu_j) \\
		& \hphantom{= - } \quad \times (A_l\widehat\omega_l - 1)\frac{ A_iK_h(\widehat\mu_i - \widehat\mu_k)K_{h}(\widehat\mu_l - \widehat\mu_k)}{(n-1)\widehat{Q}_{-il}(\widehat\mu_k)\widehat{Q}(\widehat\mu_k)}A_k(Y_k - \widehat\mu_k).
	\end{align*}
	When $j = k$, the expectation of the last two terms can be upper bounded as
	\begin{align*}
		& \E \left[ \left| (A_i\widehat\omega_i - 1)(A_l\widehat\omega_l - 1)(Y_j - \widehat\mu_j)^2\frac{A_jA_lK^2_h(\widehat\mu_l - \widehat\mu_i)K_h(\widehat\mu_i - \widehat\mu_j)}{(n-1)\widehat{Q}(\widehat\mu_j)\widehat{Q}_{-il}(\widehat\mu_j) \widehat{Q}_{-l}(\widehat\mu_j)Q(\widehat\mu_j)}\{Q(\widehat\mu_j) - \widehat{Q}_{-l}(\widehat\mu_j)\}  \right| \right. \\
		& \hphantom{=} \left. \quad + \left| (A_i\widehat\omega_i - 1)(A_l\widehat\omega_l - 1)(Y_j - \widehat\mu_j)^2 \frac{A_jA_iK^2_h(\widehat\mu_i - \widehat\mu_j)K_h(\widehat\mu_l - \widehat\mu_j)}{(n-1)\widehat{Q}_{-il}(\widehat\mu_j)\widehat{Q}(\widehat\mu_j)\widehat{Q}_{-il}(\widehat\mu_j)Q^(\widehat\mu_j)}\{Q(\widehat\mu_j) - \widehat{Q}_{-il}(\widehat\mu_j)\}\right| \mid D^n \right] \\
		& \lesssim \E\left\{ \frac{A_jA_l K_h(\widehat\mu_l - \widehat\mu_i)K_h(\widehat\mu_i - \widehat\mu_j)}{h(n-1)} \E\left\{ \left| Q(\widehat\mu_j) - \widehat{Q}_{-l}(\widehat\mu_j) \right| \mid A_j, X_j, X_l, X_i, D^n \right\} \mid D^n \right\} \\
		& \hphantom{\lesssim} + \E\left\{ \frac{A_jA_iK_h(\widehat\mu_i - \widehat\mu_j)K_h(\widehat\mu_l - \widehat\mu_j)}{h(n-1)} \E\left\{ \left| Q(\widehat\mu_j) - \widehat{Q}_{-il}(\widehat\mu_j) \right| \mid A_j, X_j, D^n \right\} \mid D^n \right\} \\
		& \lesssim (nh)^{-3/2}
	\end{align*}
	The last inequality follows because
	\begin{align*}
		\E\left\{ A_jA_l K_h(\widehat\mu_l - \widehat\mu_i)K_h(\widehat\mu_i - \widehat\mu_j) \mid D^n \right\} & = \E\left[A_lK_h(\widehat\mu_l - \widehat\mu_i)\E\{A_jK_h(\widehat\mu_i - \widehat\mu_j) \mid X_i, D^n\} \mid D^n \right] \\
		& \lesssim \E\left\{A_lK_h(\widehat\mu_l - \widehat\mu_i) \mid D^n\right\} \\
		& = \E\left[\E\left\{A_lK_h(\widehat\mu_l - \widehat\mu_i) \mid X_i, D^n\right\}\mid D^n \right] \\
		& \lesssim 1.
	\end{align*}
	and because, by Lemma \ref{lemma:Q},
	\begin{align*}
		& \E\left[ \left| \widehat{Q}_{-l}(\widehat\mu_l) - Q(\widehat\mu_j) \right| \mid A_j, X_j, X_l, X_i, D^n \right] \lesssim (nh)^{-1/2},\\
		& \E\left[ \left| \widehat{Q}_{-il}(\widehat\mu_j) - Q(\widehat\mu_j) \right| \mid A_j, X_j, D^n \right] \lesssim (nh)^{-1/2}.
	\end{align*}
	On the other hand, when $j \neq k$:
	\begin{align*}
		& \E \left[ \left| (A_i\widehat\omega_i - 1)\frac{A_lK_h(\widehat\mu_l - \widehat\mu_j)K_h(\widehat\mu_i - \widehat\mu_j)}{(n-1)\widehat{Q}_{il}(\widehat\mu_j)\widehat{Q}_{-i}(\widehat\mu_j)}A_j(Y_j - \widehat\mu_j) \right. \right. \\
		& \hphantom{=} \quad\quad \left. \vphantom{\frac{A_lK_h(\widehat\mu_l - \widehat\mu_j)K_h(\widehat\mu_i - \widehat\mu_j)}{(n-1)\widehat{Q}_{il}(\widehat\mu_j)\widehat{Q}_{-i}(\widehat\mu_j)}} \times (A_l\widehat\omega_l - 1)\{\widehat{Q}^{-1}_{-l}(\widehat\mu_k) - Q^{-1}(\widehat\mu_k)\}K_h(\widehat\mu_l - \widehat\mu_k)A_k(Y_k - \widehat\mu_k)\right| \\
		& \hphantom{=} \quad + \left| \vphantom{\frac{A_lK_h(\widehat\mu_l - \widehat\mu_j)K_h(\widehat\mu_i - \widehat\mu_j)}{(n-1)\widehat{Q}_{il}(\widehat\mu_j)\widehat{Q}_{-i}(\widehat\mu_j)}} (A_i\widehat\omega_i - 1)\{\widehat{Q}^{-1}_{-il}(\widehat\mu_j) - Q^{-1}(\widehat\mu_j)\}K_h(\widehat\mu_i - \widehat\mu_j)A_j(Y_j - \widehat\mu_j) \right. \\
		& \hphantom{=} \left. \left. \quad\quad\quad \times (A_l\widehat\omega_l - 1)\frac{A_iK_h(\widehat\mu_i - \widehat\mu_k)K_h(\widehat\mu_l - \widehat\mu_k)}{(n-1)\widehat{Q}_{-il}(\widehat\mu_k)\widehat{Q}(\widehat\mu_k)}A_k(Y_k - \widehat\mu_k)\right| \mid D^n \right] \\
		& \lesssim \E\left\{ \frac{A_jA_lA_k K_h(\widehat\mu_l - \widehat\mu_j)K_h(\widehat\mu_i - \widehat\mu_j)K_h(\widehat\mu_l - \widehat\mu_k)}{(n-1)} \E\left\{ \left| \widehat{Q}_{-l}(\widehat\mu_k) - Q(\widehat\mu_k) \right| \mid A_j, A_k, X_j, X_i, X_k, D^n \right\} \mid D^n \right\} \\
		& \hphantom{\lesssim} + \E\left\{ \frac{A_jA_iA_kK_h(\widehat\mu_i - \widehat\mu_j)K_h(\widehat\mu_i - \widehat\mu_k)K_h(\widehat\mu_l - \widehat\mu_k)}{(n-1)} \E\left\{ \left| \widehat{Q}_{-il}(\widehat\mu_j) - Q(\widehat\mu_j) \right| \mid A_j, A_k, X_j, X_k, D^n \right\} \mid D^n \right\} \\
		& \lesssim (n^3h)^{-1/2}
	\end{align*}
	The last inequality follows because
	\begin{align*}
		& \E\left\{ A_jA_lA_k K_h(\widehat\mu_l - \widehat\mu_j)K_h(\widehat\mu_i - \widehat\mu_j)K_h(\widehat\mu_l - \widehat\mu_k) \mid D^n \right\} \\
		& = \E\left[A_jA_l K_h(\widehat\mu_l - \widehat\mu_j)K_h(\widehat\mu_i - \widehat\mu_j)\E\{A_kK_h(\widehat\mu_l - \widehat\mu_k) \mid X_l, D^n\} \mid D^n \right] \\
		& \lesssim \E\left\{A_j A_lK_h(\widehat\mu_l - \widehat\mu_j)K_h(\widehat\mu_i - \widehat\mu_j) \mid D^n \right\} \\
		& = \E\left[A_lK_h(\widehat\mu_l - \widehat\mu_j)\E\{A_jK_h(\widehat\mu_i - \widehat\mu_j) \mid X_i, D^n\} \mid D^n \right] \\
		& \lesssim \E\left\{A_lK_h(\widehat\mu_l - \widehat\mu_j) \mid D^n\right\} \\
		& = \E\left[\E\left\{A_lK_h(\widehat\mu_l - \widehat\mu_j) \mid X_i, D^n\right\}\mid D^n \right] \\
		& \lesssim 1.
	\end{align*}
	and by Lemma \ref{lemma:Q}. 
	
	Next, we proceed to bound the first term appearing in $T_{ij}T_{lk}$, which is the main term. 
	
	When $k = j$, we have 
	\begin{align*}
		& \left| \E\left[(A_i\widehat\omega_i - 1)\{\widehat{Q}^{-1}_{-il}(\widehat\mu_j) - Q^{-1}(\widehat\mu_j)\}^2K_h(\widehat\mu_i - \widehat\mu_j)K_h(\widehat\mu_l - \widehat\mu_j)A_j(Y_j - \widehat\mu_j)^2\right] \right| \\
		& = \left| \E\left[A_i(\widehat\omega_i - \omega_i)\{\widehat{Q}^{-1}_{-il}(\widehat\mu_j) - Q^{-1}(\widehat\mu_j)\}^2K_h(\widehat\mu_i - \widehat\mu_j)K_h(\widehat\mu_l - \widehat\mu_j)A_j(Y_j - \widehat\mu_j)^2A_l(\widehat\omega_l - \omega_l) \right] \right| \\
		& \lesssim \E\left(A_iA_jK_h(\widehat\mu_i - \widehat\mu_j)K_h(\widehat\mu_l - \widehat\mu_j)\E\left[\{Q(\widehat\mu_j) - \widehat{Q}_{-il}(\widehat\mu_j)\}^2 \mid X_i, X_j, X_l, D^n \right] \mid D^n \right) \\
		& \lesssim (nh)^{-1}.
	\end{align*}
	In this light, 
	\begin{align*}
		\E\left\{\frac{1}{n^2(n-1)^2} \mathop{\sum\sum\sum}\limits_{1 \leq i \neq j \neq l \leq n} T_{ij}T_{lj} \mid D^n \right\} \lesssim (n^2h)^{-1}.
	\end{align*}
	Next, for $k \neq j$, we have
	\begin{align*}
		& \left| \E\left[(A_i\widehat\omega_i - 1)\{\widehat{Q}^{-1}_{-il}(\widehat\mu_j) - Q^{-1}(\widehat\mu_j)\}K_h(\widehat\mu_i - \widehat\mu_j)A_j(Y_j - \widehat\mu_j) \right. \right. \\
		& \left. \left. \quad\quad \times (A_l\widehat\omega_l - 1)\{\widehat{Q}^{-1}_{-il}(\widehat\mu_k) - Q^{-1}(\widehat\mu_k)\}K_h(\widehat\mu_l - \widehat\mu_k)A_k(Y_k - \widehat\mu_k) \mid D^n \right] \right| \\
		& = \left| \E\left[A_i(\widehat\omega_i - \omega_i)\{\widehat{Q}^{-1}_{-il}(\widehat\mu_j) - Q^{-1}(\widehat\mu_j)\}K_h(\widehat\mu_i - \widehat\mu_j)A_j(\mu_j - \widehat\mu_j) \right. \right. \\
		& \left. \left. \quad\quad \times A_l(\widehat\omega_l - \omega_l)\{\widehat{Q}^{-1}_{-il}(\widehat\mu_k) - Q^{-1}(\widehat\mu_k)\}K_h(\widehat\mu_l - \widehat\mu_k)A_k(\mu_k - \widehat\mu_k) \mid D^n \right] \right| \\
		& \lesssim \left| \E\left[|\widehat\omega_i - \omega_i||\mu_j - \widehat\mu_j||\widehat\omega_l - \omega_l||\mu_k - \widehat\mu_k|A_iA_jA_kA_lK_h(\widehat\mu_i - \widehat\mu_j)K_h(\widehat\mu_l - \widehat\mu_k) \right. \right. \\
		& \left. \left. \quad\quad\quad \times  \E\left\{|\widehat{Q}_{-il}(\widehat\mu_j) - Q(\widehat\mu_j)| | \widehat{Q}_{-il}(\widehat\mu_k) - Q(\widehat\mu_k)| \mid A_j, A_k, X_j, X_k, D^n \right\} \right] \mid D^n \right| \\
		& \lesssim (nh)^{-1} \E\left[|\widehat\omega_i - \omega_i| |\widehat\omega_l - \omega_l| 
		A_jA_iA_kA_lK_h(\widehat\mu_i - \widehat\mu_j)K_h(\widehat\mu_l - \widehat\mu_k)|\widehat\mu_j - \mu_j||\widehat\mu_k - \mu_k| \mid D^n \right] \\
		& \lesssim (nh)^{-1} \|\widehat\omega - \omega\|^2\|\widehat\mu - \mu\|^2.
	\end{align*}
	where the second-to-last inequality follows by Lemma \ref{lemma:Q} and the Cauchy-Schwarz inequality, while the last one follows by the Cauchy-Schwarz inequality and the independence of the observations because:
	\begin{align*}
		& \left\{\E\left(|\widehat\omega_i - \omega_i|
		A_iA_jK_h(\widehat\mu_i - \widehat\mu_j)|\widehat\mu_j - \mu_j| \mid D^n\right)\right\}^2 \\
		& \leq \E\left[(\widehat\omega_i - \omega_i)^2
		\E\{A_jK_h(\widehat\mu_i - \widehat\mu_j) \mid X_i, D^n\} \mid D^n \right] \E\left[(\widehat\mu_j - \mu_j)^2  \E\{A_iK_h(\widehat\mu_i - \widehat\mu_j) \mid X_j, D^n\} \mid D^n \right] \\
		& \lesssim \|\widehat\omega - \omega\|^2\|\widehat\mu - \mu\|^2.
	\end{align*}
	We have thus reached, for $j \neq k$:
	\begin{align*}
		& \left| \E(T_{ij}T_{lk} \mid D^n) \right| \lesssim \|\widehat\omega - \omega\|^2\|\widehat\mu - \mu\|^2(nh)^{-1} + (n^3h)^{-1/2},
	\end{align*}
	implying that
	\begin{align*}
		\E\left(\frac{1}{n^2(n-1)^2} \mathop{\sum\sum\sum\sum}\limits_{1 \leq i \neq j \neq l \neq k \leq n} T_{ij} T_{lk} \mid D^n\right) \lesssim \|\widehat\omega - \omega\|^2 \|\widehat\mu- \mu\|^2(nh)^{-1} + (n^3h)^{-1/2}
	\end{align*}
	Putting everything together, because $nh \to \infty$, we have that
	\begin{align*}
		\E\{(\widetilde{T}_{n\mu} - T_{n\mu})^2 \mid D^n\} \lesssim (n^3h)^{-1/2} + \|\widehat\omega - \omega\|^2\|\widehat\mu - \mu\|^2(nh)^{-1}.
	\end{align*}
	This concludes our proof of Lemma \ref{lemma:Tnb}.
	\subsubsection{Proof of the bounds in Eq. \eqref{eq:statement3b}}
	We have
	\begin{align*}
		\int \kappa_\mu(z_1, Z_2; D^n) d\Pb(z_1) & = \E\left\{(A_1\widehat\omega_1 - 1)\frac{K_h(\widehat\mu_1 - \widehat\mu_2)}{Q(\widehat\mu_2)} \mid X_2, D^n\right\}A_2(Y_2 - \widehat\mu_2) \\
		& = \E\left\{A_1\E(\widehat\omega_1 - \omega_1 \mid A_1 =1, \widehat\mu_1, D^n)\frac{K_h(\widehat\mu_1 - \widehat\mu_2)}{Q(\widehat\mu_2)} \mid X_2, D^n\right\}A_2(Y_2 - \widehat\mu_2) \\
		& = \widehat\varphi_\mu(Z_2; D^n) \\
		& \hphantom{=} + \E\left[A_1\frac{K_h(\widehat\mu_1 - \widehat\mu_2)}{Q(\widehat\mu_2)}\{s_\mu(\widehat\mu_1, \mu_1; D^n) - s_\mu(\widehat\mu_2, \mu_2; D^n)\} \mid X_2, D^n\right]A_2(Y_2 - \widehat\mu_2)
	\end{align*}
	By the smoothness assumption on $s_\mu(t_1, t_2; D^n)$, we have
	\begin{align*}
		& \left| \E\left[A_1\frac{K_h(\widehat\mu_1 - \widehat\mu_2)}{Q(\widehat\mu_2)}\{s_\mu(\widehat\mu_1, \mu_1; D^n) - s_\mu(\widehat\mu_2, \mu_2; D^n)\} \mid X_2, D^n\right] \right| \lesssim \|\widehat\mu - \mu\|^\beta \ + \ h^\beta
	\end{align*}
	In this light, 
	\begin{align*}
		\left| \kappa_{\mu2}(Z; D^n) - \widehat\varphi_\mu(Z; D^n)\right| \lesssim |A(Y - \widehat\mu)| (\|\widehat\mu - \mu\|^\beta \ + \ h^\beta)
	\end{align*}
	so that $\|\kappa_{\mu2} - \widehat\varphi_\mu\| \ \lesssim \|\widehat\mu - \mu\|^\beta \ + \ h^\beta$. 
	
	\subsubsection{Proof of Lemma \ref{lemma:Snb}}
	The proof of Lemma \ref{lemma:Snb} follows by arguments identical to that used to prove Lemma \ref{lemma:Sna}.
	\subsection{Proof of Proposition \ref{prop:min_squared}: preliminaries}
	Recall our shorthand notation $\widehat\omega(X_i) = \widehat\omega_i$, $\widehat\omega = \widehat\omega(X)$ and so on, as well as $K_{hi}(\widehat\omega_j) = h^{-1}K\left(\frac{\widehat\omega_i - \widehat\omega_j}{h}\right)$ with defined $K_{hi}(\widehat\mu_j)$ similarly. Further, we define
	\begin{align*}
		& \widehat{Q}(\widehat\omega_i, \widehat\mu_i) =  \frac{1}{n-1}\sum_{1 \leq j \leq n, j \neq i} A_j K_{hi}(\widehat\omega_j) K_{hi}(\widehat\mu_j), \\
		& Q(\widehat\omega_i, \widehat\mu_i) = \E\left\{\widehat{Q}(\widehat\omega_i, \widehat\mu_i) \mid X_i, D^n\right\} \\
		& \hphantom{Q(\widehat\omega_i, \widehat\mu_i)} =  \int \{\omega(x)\}^{-1}K_h(\widehat\omega(x) - \widehat\omega(X_i))K_h(\widehat\mu(x) - \widehat\mu(X_i)) d\Pb(x) \\
		& \hphantom{Q(\widehat\omega_i, \widehat\mu_i)}  = \Pb(A = 1) \int K(u) K(v)  d\Pb_{\widehat\omega, \widehat\mu \mid A = 1, D^n}(hu + \widehat\omega(X_i), hv + \widehat\mu(X_i)), \\
		& Q(\widehat\omega_j, \widehat\mu_j) =  \int \{\omega(x)\}^{-1}K_h(\widehat\omega(x) - \widehat\omega(X_j))K_h(\widehat\mu(x) - \widehat\mu(X_j)) d\Pb(x) \\
		& \hphantom{Q(\widehat\omega_j, \widehat\mu_j)} = \Pb(A = 1) \int K(u) K(v)  d\Pb_{\widehat\omega, \widehat\mu \mid A = 1, D^n}(hu + \widehat\omega(X_j), hv + \widehat\mu(X_j)),
	\end{align*}
	where $d\Pb_{\widehat\omega, \widehat\mu \mid A = 1, D^n}(u, v)$ is the density of $(\widehat\omega(X), \widehat\mu(X))$ among units with $A = 1$ keeping $\widehat\omega(\cdot)$ and $\widehat\mu(\cdot)$ as fixed functions given the training sample $D^n$. 
	The estimator is $\widehat\psi = \Pn \widehat\varphi - T_n$, where $\widehat\varphi(O) = A\widehat\omega(Y - \widehat\mu) + \widehat\mu$ and
	\begin{align*}
		& T_n = \frac{1}{n(n-1)} \mathop{\sum\sum}\limits_{1\leq i \neq j \leq n} (A\widehat\omega_i - 1) \{\widehat{Q}(\widehat\omega_i, \widehat\mu_i)\}^{-1} K_{hi}(\widehat\omega_j) K_{hi}(\widehat\mu_j) A_j(Y_j - \widehat\mu_j) \equiv \frac{1}{n(n-1)} \mathop{\sum\sum}\limits_{1\leq i \neq j \leq n} T_{ij}.
	\end{align*}
	From the main decomposition \eqref{eq:main_decomposition}, we have
	\begin{align*}
		\widehat\psi - \psi = (\Pn - \Pb)(\widehat\varphi - \overline\varphi) + (\Pn - \Pb)\overline\varphi + R_n - T_n, \text{ where } R_n = \Pb(\widehat\varphi - \varphi). 
	\end{align*}
	For compact notation, let us define:
	\begin{align*}
		& \widehat\kappa(Z_1, Z_2; D^n) = (A_1\widehat\omega_1 - 1)\frac{K_h(\widehat\omega_1 - \widehat\omega_2)K_h(\widehat\mu_1 - \widehat\mu_2)}{\widehat{Q}(\widehat\omega_1, \widehat\mu_1)} A_2(Y_2 - \widehat\mu_2), \\
		& \kappa(Z_1, Z_2; D^n) = (A_1\widehat\omega_1 - 1)\frac{K_h(\widehat\omega_1 - \widehat\omega_2)K_h(\widehat\mu_1 - \widehat\mu_2)}{Q(\widehat\omega_1, \widehat\mu_1)} A_2(Y_2 - \widehat\mu_2),
	\end{align*}
	so that
	\begin{align*}
		& T_{n} = \frac{1}{n(n-1)} \mathop{\sum\sum}\limits_{1 \leq i \neq j \leq n}\widehat\kappa(Z_i, Z_j; D^n).
	\end{align*}
	Also define
	\begin{align*}
		\widetilde{T}_{n} = \frac{1}{n(n-1)} \mathop{\sum\sum}\limits_{1 \leq i \neq j \leq n}\kappa(Z_i, Z_j; D^n),
	\end{align*}
	and let
	\begin{align*}
		& \widehat\varphi_{\omega\mu} = \E(\omega - \widehat\omega \mid A = 1, \widehat\omega, \widehat\mu, D^n)A(Y - \widehat\mu) + \E(\mu - \widehat\mu \mid A = 1, \widehat\omega, \widehat\mu, D^n)(A\widehat\omega - 1), \\
		& \overline\varphi_{\omega\mu} = \E(\omega - \overline\omega \mid A = 1, \overline\omega, \overline\mu)A(Y - \overline\mu) + \E(\mu - \overline\mu \mid A = 1, \overline\omega, \overline\mu)(A\overline\omega - 1).
	\end{align*}
	Finally, let
	\begin{align*}
		& \kappa_1(Z; D^n) = \int \kappa(Z, z; D^n) d\Pb(z) \quad \text{ and } \quad \kappa_2(Z; D^n) = \int \kappa(z, Z; D^n) d\Pb(z). 
	\end{align*}
	We have the following decomposition
	\begin{align*}
		T_{n} & = R_n + \E\left\{ \kappa(Z_1, Z_2; D^n) \mid D^n \right\} - R_n  + T_{n} - \widetilde{T}_{n} \\
		& \hphantom{=} + (\Pn - \Pb)(\kappa_1 + \kappa_2 - \widehat\varphi_{\omega\mu}) + (\Pn - \Pb)(\widehat\varphi_{\omega\mu} - \overline\varphi_{\omega\mu}) + (\Pn - \Pb)\overline\varphi_{\omega\mu} + S_{n},
	\end{align*}
	where
	\begin{align*}
		S_n = \frac{1}{n(n-1)}\mathop{\sum\sum}_{1 \leq i \neq j \leq n} \left[ \kappa(Z_i, Z_j; D^n) - \kappa_1(Z_i; D^n) - \kappa_2(Z_j; D^n)+ \E\left\{ \kappa(Z_1, Z_2; D^n) \mid D^n \right\}\right].
	\end{align*}
	We will repeatedly use the following lemmas:
	\begin{lemma}\label{lemma:Sn}
		It holds that
		\begin{enumerate}
			\item $\sup_z |\kappa_{1}(z; D^n)| \ \lesssim \|\widehat\mu - \mu\|_\infty \ \land \ h^{-1} \|\widehat\mu - \mu\|$,
			\item $\sup_z|\kappa_{2}(z; D^n)| \ \lesssim \|\widehat\omega - \omega\|_\infty \ \land \ h^{-1}\|\widehat\omega - \omega\|$,
			\item $\E(S^2_{n} \mid D^n) \lesssim (nh)^{-2}$.
		\end{enumerate}
	\end{lemma}
	\begin{lemma}\label{lemma:Tn}
		It holds that
		\begin{align*}
			\E\{(T_n - \widetilde{T}_n)^2 \mid D^n\} \lesssim (n^3h^2)^{-1/2} + \|\widehat\omega - \omega\|^2\|\widehat\mu - \mu\|^2 (nh^2)^{-1}.
		\end{align*}
	\end{lemma}
	\begin{lemma}\label{lemma:RTn}
		It holds that
		\begin{align*}
			\left| \E(\widetilde{T}_n \mid D^n) - R_n \right| & = \left| \E\{\kappa(Z_1, Z_2; D^n) \mid D^n\} - R_n \right| \\
			& \lesssim (\|\widehat\mu-\mu\|^{1+\beta} \ + \ h^\beta \|\widehat\mu - \mu\|) \land (\|\widehat\omega-\omega\|^{1+\alpha} \ + \ h^\alpha \|\widehat\omega - \omega\|)  \nonumber \\  
			& \hphantom{\lesssim} \quad\quad \land \|\widehat\omega - \omega\|\|\widehat\mu-\mu\|
		\end{align*}
	\end{lemma}
	\subsubsection{Bias of $\widehat\psi$ \label{sec:bias_psi}}
	\textbf{Bound 1.} We start with the following decomposition:
	\begin{align*}
		T_{ij} = (A_i\widehat\omega_i - 1) \{Q(\widehat\omega_j, \widehat\mu_j)\}^{-1} K_{hi}(\widehat\omega_j) K_{hi}(\widehat\mu_j) A_j(Y_j - \widehat\mu_j) + T_{2ij} + T_{3ij},
	\end{align*}
	where
	\begin{align*}
		& T_{2ij} = (A_i\widehat\omega_i - 1) \left[\{\widehat{Q}(\widehat\omega_i, \widehat\mu_i)\}^{-1} - \{Q(\widehat\omega_i, \widehat\mu_i)\}^{-1} \right] K_{hi}(\widehat\omega_j) K_{hi}(\widehat\mu_j) A_j(Y_j - \widehat\mu_j), \\
		& T_{3ij} = (A_i\widehat\omega_i - 1) \left[\{Q(\widehat\omega_i, \widehat\mu_i)\}^{-1} - \{Q(\widehat\omega_j, \widehat\mu_j)\}^{-1} \right] K_{hi}(\widehat\omega_j) K_{hi}(\widehat\mu_j) A_j(Y_j - \widehat\mu_j).
	\end{align*}
	We have
	\begin{align*}
		& \E\left[\frac{1}{n(n-1)}\mathop{\sum\sum}\limits_{1 \leq i \neq j \leq n} (A_i\widehat\omega_i - 1) \{Q(\widehat\omega_j, \widehat\mu_j)\}^{-1} K_{hi}(\widehat\omega_j) K_{hi}(\widehat\mu_j) A_j(Y_j - \widehat\mu_j) \mid D^n \right] \\
		& = \E\left[\frac{1}{n(n-1)}\mathop{\sum\sum}\limits_{1 \leq i \neq j \leq n} A_i(\widehat\omega_i - \omega_i) \{Q(\widehat\omega_j, \widehat\mu_j)\}^{-1} K_{hi}(\widehat\omega_j) K_{hi}(\widehat\mu_j) A_j(Y_j - \widehat\mu_j) \mid D^n \right] \\
		& = \E\left[\frac{1}{n(n-1)}\mathop{\sum\sum}\limits_{1 \leq i \neq j \leq n} A_if_\mu(\widehat\mu_i, \mu_i, \widehat\omega_i; D^n) \{Q(\widehat\omega_j, \widehat\mu_j)\}^{-1} K_{hi}(\widehat\omega_j) K_{hi}(\widehat\mu_j) A_j(Y_j - \widehat\mu_j) \mid D^n \right] \\
		& = \E\left[\frac{1}{n(n-1)}\mathop{\sum\sum}\limits_{1 \leq i \neq j \leq n} A_if_\mu(\widehat\mu_j, \mu_j, \widehat\omega_j; D^n) \{Q(\widehat\omega_j, \widehat\mu_j)\}^{-1} K_{hi}(\widehat\omega_j) K_{hi}(\widehat\mu_j) A_j(Y_j - \widehat\mu_j) \mid D^n \right] \\
		& \hphantom{=} + \E\left[\frac{1}{n(n-1)}\mathop{\sum\sum}\limits_{1 \leq i \neq j \leq n} A_i \{f_\mu(\widehat\mu_i, \mu_i, \widehat\omega_i; D^n) - f_\mu(\widehat\mu_j, \mu_j, \widehat\omega_j; D^n)\} \frac{A_jK_{hi}(\widehat\omega_j) K_{hi}(\widehat\mu_j)(\mu_j - \widehat\mu_j)}{Q(\widehat\omega_j, \widehat\mu_j)}\mid D^n \right] \\
		& = R_n + \E\left(\frac{1}{n(n-1)}\mathop{\sum\sum}\limits_{1 \leq i \neq j \leq n} T_{1ij} \mid D^n\right)
	\end{align*}
	The second equality follows because
	\begin{align*}
		& \E\{A_i(\widehat\omega_i - \omega_i) - A_if_\mu(\widehat\mu_i, \mu_i, \widehat\omega_i; D^n) \mid A_i, A_j, Y_j, \widehat\omega_i, \widehat\omega_j, \widehat\mu_i, \widehat\mu_j, D^n\} \\
		& = \E\{A_i(\widehat\omega_i - \omega_i) - A_if_\mu(\widehat\mu_i, \mu_i, \widehat\omega_i; D^n) \mid A_i, \widehat\omega_i, \widehat\mu_i, D^n\} \\
		& = \E\left[\E\{A_i(\widehat\omega_i - \omega_i) \mid A_i, \widehat\omega_i, \widehat\mu_i, \mu_i, D^n\} - A_if_\mu(\widehat\mu_i, \mu_i, \widehat\omega_i; D^n) \mid A_i, \widehat\omega_i, \widehat\mu_i, D^n \right] \\
		& = \E\left\{A_i\E(\widehat\omega_i - \omega_i \mid A_i = 1, \widehat\omega_i, \widehat\mu_i, \mu_i, D^n) - A_if_\mu(\widehat\mu_i, \mu_i, \widehat\omega_i; D^n) \mid A_i, \widehat\omega_i, \widehat\mu_i, D^n \right\} \\
		& = 0
	\end{align*}
	The last equality follows because
	\begin{align*}
		& \E\left[\frac{1}{n(n-1)}\mathop{\sum\sum}\limits_{1 \leq i \neq j \leq n} A_if_\mu(\widehat\mu_j, \mu_j, \widehat\omega_j; D^n) \{Q(\widehat\omega_j, \widehat\mu_j)\}^{-1} K_{hi}(\widehat\omega_j) K_{hi}(\widehat\mu_j) A_j(Y_j - \widehat\mu_j) \mid D^n \right] \\
		& = \E\left[\frac{1}{n(n-1)}\mathop{\sum\sum}\limits_{1 \leq i \neq j \leq n} \frac{\E\{A_iK_{hi}(\widehat\omega_j) K_{hi}(\widehat\mu_j)\mid X_j, D^n\}}{Q(\widehat\omega_j, \widehat\mu_j)} f_\mu(\widehat\mu_j, \mu_j, \widehat\omega_j; D^n) A_j(Y_j - \widehat\mu_j) \mid D^n \right] \\
		& = \E\left[\frac{1}{n(n-1)}\mathop{\sum\sum}\limits_{1 \leq i \neq j \leq n} f_\mu(\widehat\mu_j, \mu_j, \widehat\omega_j; D^n) A_j(Y_j - \widehat\mu_j) \mid D^n \right] \\
		& = \E\left[\frac{1}{n}\sum_{j = 1}^n f_\mu(\widehat\mu_j, \mu_j, \widehat\omega_j; D^n) A_j(Y_j - \widehat\mu_j) \mid D^n \right] \\
		& = R_n.
	\end{align*}
	Thus, we have reached
	\begin{align*}
		\E\left\{ \frac{1}{n(n-1)}\mathop{\sum\sum}\limits_{1 \leq i \neq j \leq n} T_{ij} \mid D^n \right\} = R_n + \E\left(\frac{1}{n(n-1)}\mathop{\sum\sum}\limits_{1 \leq i \neq j \leq n} T_{1ij} + T_{2ij} + T_{3ij} \mid D^n \right)
	\end{align*}
	so that
	\begin{align*}
		\E(\widehat\psi - \psi \mid D^n) = \E(R_n - T_n \mid D^n) = \E\left(\frac{1}{n(n-1)}\mathop{\sum\sum}\limits_{1 \leq i \neq j \leq n} T_{1ij} + T_{2ij} + T_{3ij} \mid D^n \right)
	\end{align*}
	We bound each term separately.
	
	\textit{Term $T_1$}. By the smoothness assumption on $f_\mu$, the last term can be upper bounded as
	\begin{align*}
		& \left|\E\left(\frac{1}{n(n-1)}\mathop{\sum\sum}\limits_{1 \leq i \neq j \leq n} T_{1ij} \mid D^n\right)\right| \\
		& = \left|\E\left[\frac{1}{n(n-1)}\mathop{\sum\sum}\limits_{1 \leq i \neq j \leq n} A_i \{f_\mu(\widehat\mu_i, \mu_i, \widehat\omega_i; D^n) - f_\mu(\widehat\mu_j, \mu_j, \widehat\omega_j; D^n)\}\frac{A_jK_{hi}(\widehat\omega_j) K_{hi}(\widehat\mu_j)(\mu_j - \widehat\mu_j)}{Q(\widehat\omega_j, \widehat\mu_j)} \mid D^n \right] \right| \\
		& \lesssim \E\left\{ \left(A_i| \widehat\mu_i - \widehat\mu_j|^\beta \ + \ |\mu_i - \mu_j|^\beta \ + \ |\widehat\omega_i - \widehat\omega_j|^\beta\right)A_jK_{hi}(\widehat\omega_j) K_{hi}(\widehat\mu_j) |\mu_j - \widehat\mu_j| \mid D^n \right\} \\
		& \leq \E\left\{ A_i \left(h^\beta \ + \ |\mu_i - \widehat\mu_i|^\beta \ + \ |\widehat\mu_j - \mu_j|^\beta\right)A_jK_{hi}(\widehat\omega_j) K_{hi}(\widehat\mu_j) |\mu_j - \widehat\mu_j| \mid D^n \right\} \\
		& \lesssim h^\beta \E\left[\E\left\{A_iK_{hi}(\widehat\omega_j) K_{hi}(\widehat\mu_j) \mid X_j, D^n\right\} |\mu_j - \widehat\mu_j| \mid D^n \right] \\
		& \hphantom{\lesssim} + \left(\E\left[\E\left\{A_jK_{hi}(\widehat\omega_j) K_{hi}(\widehat\mu_j) \mid X_i, D^n\right\} (\mu_i - \widehat\mu_i)^{2\beta} \mid D^n \right] \right)^{1/2} \\
		& \hphantom{\lesssim} \quad\quad \times \left(\E\left[\E\left\{A_iK_{hi}(\widehat\omega_j) K_{hi}(\widehat\mu_j) \mid X_j, D^n\right\} (\mu_j - \widehat\mu_j)^{2} \mid D^n \right] \right)^{1/2} \\
		& \hphantom{\lesssim} + \E\left[\E\left\{A_iK_{hi}(\widehat\omega_j) K_{hi}(\widehat\mu_j) \mid X_j, D^n\right\} |\mu_j - \widehat\mu_j|^{1 + \beta} \mid D^n \right]  \\
		& \lesssim h^\beta \| \widehat\mu - \mu\| \ + \ \|\widehat\mu - \mu\|^{1 + \beta}
	\end{align*}
	where the last inequality follows because $\E\left\{A_jK_{hi}(\widehat\omega_j) K_{hi}(\widehat\mu_j) \mid X_i, D^n\right\} \lesssim 1$. 
	
	\textit{Term $T_2$.} Let $A^n = (A_1, \ldots, A_n)$ and $A^n_{-i} = A^n \setminus A_i$. Because $\widehat{Q}(\widehat\omega_i, \widehat\mu_i)$ does not depend on $A_i$, we have $\E\{\widehat{Q}(\widehat\omega_i, \widehat\mu_i) \mid X^n, A^n_{-i}, D^n\} = \widehat{Q}(\widehat\omega_i, \widehat\mu_i)$. Therefore:
	\begin{align*}
		\E(T_{2ij} \mid X^n, A^n_{-i}, D^n) = \frac{(\widehat\omega_i - \omega_i)A_j(\mu_j - \widehat\mu_j)}{\omega_i\widehat{Q}(\widehat\omega_i, \widehat\mu_i)Q(\widehat\omega_i, \widehat\mu_i)}K_{hi}(\widehat\omega_j)K_{hi}(\widehat\mu_j) \{Q(\widehat\omega_i, \widehat\mu_i) - \widehat{Q}(\widehat\omega_i, \widehat\mu_i)\} 
	\end{align*}
	Therefore, we have
	\begin{align*}
		|\E(T_{2ij} \mid D^n)| & \leq \E\left[\left| \frac{A_i(\widehat\omega_i - \omega_i)A_j(\mu_j - \widehat\mu_j)}{\widehat{Q}(\widehat\omega_i, \widehat\mu_i)Q(\widehat\omega_i, \widehat\mu_i)}K_{hi}(\widehat\omega_j)K_{hi}(\widehat\mu_j) \{Q(\widehat\omega_i, \widehat\mu_i) - \widehat{Q}(\widehat\omega_i, \widehat\mu_i)\}  \right| \mid D^n \right] \\
		& \lesssim \E\left\{ |\widehat\omega_i - \omega_i||\mu_j - \widehat\mu_j| A_iA_jK_{hi}(\widehat\omega_j)K_{hi}(\widehat\mu_j) \left|Q(\widehat\omega_i, \widehat\mu_i) - \widehat{Q}(\widehat\omega_i, \widehat\mu_i)\right| \mid D^n \right\} \\
		& = \E\left[|\widehat\omega_i - \omega_i||\mu_j - \widehat\mu_j| A_iA_jK_{hi}(\widehat\omega_j)K_{hi}(\widehat\mu_j) \E\left\{\left|Q(\widehat\omega_i, \widehat\mu_i) - \widehat{Q}(\widehat\omega_i, \widehat\mu_i)\right| \mid A_j, X_i, X_j, D^n\right\} \mid D^n \right]
	\end{align*}
	By Lemma \ref{lemma:Q}, we have
	\begin{align*}
		\E\left\{\left|Q(\widehat\omega_i, \widehat\mu_i) - \widehat{Q}(\widehat\omega_i, \widehat\mu_i)\right| \mid A_j, X_i, X_j, D^n\right\} \lesssim (nh^2)^{-1/2}.
	\end{align*}
	Therefore, we have
	\begin{align*}
		|\E(T_{2ij} \mid D^n)| \ & \lesssim \E\left\{ |\widehat\omega_i - \omega_i||\mu_j - \widehat\mu_j|A_i A_j K_{hi}(\widehat\omega_j)K_{hi}(\widehat\mu_j) \mid D^n \right\}(nh^2)^{-1/2} \\
		& \lesssim \|\widehat\omega - \omega\| \|\widehat\mu - \mu\|(nh^2)^{-1/2},
	\end{align*}
	so that
	\begin{align*}
		\left| \E\left(\frac{1}{n(n-1)}\mathop{\sum\sum}\limits_{1 \leq i \neq j \leq n} T_{2ij} \mid D^n \right)\right| \lesssim \|\widehat\omega - \omega\| \|\widehat\mu - \mu\|(nh^2)^{-1/2}.
	\end{align*}
	\textit{Term $T_3$}. We have
	\begin{align*}
		\E(T_{3ij} \mid D^n) & = \E\left[\frac{(\widehat\omega_i - \omega_i)(\mu_j - \widehat\mu_j)}{{Q}(\widehat\omega_i, \widehat\mu_i)Q(\widehat\omega_j, \widehat\mu_j)}A_iA_jK_{hi}(\widehat\omega_j)K_{hi}(\widehat\mu_j) \{Q(\widehat\omega_j, \widehat\mu_j) - Q(\widehat\omega_i, \widehat\mu_i)\} \mid D^n \right]
	\end{align*}
	Under the assumption that the density of $(\widehat\omega, \widehat\mu)$ given $A = 1$ is Lipschitz:
	\begin{align*}
		& \{\Pb(A = 1)\}^{-1}\left| Q(\widehat\omega_i, \widehat\mu_i) - Q(\widehat\omega_j, \widehat\mu_j)\right|  \\
		& = \left| \int K(u) K(v) \left\{d\Pb_{\widehat\omega, \widehat\mu \mid A = 1, D^n}(uh + \widehat\omega_i, vh + \widehat\mu_i) - d\Pb_{\widehat\omega, \widehat\mu \mid A = 1, D^n}(uh + \widehat\omega_j, vh + \widehat\mu_j) \right\} \right| \\
		& \lesssim |\widehat\omega_i - \widehat\omega_j|  \ + \ |\widehat\mu_i - \widehat\mu_j|.
	\end{align*}
	Therefore,
	\begin{align*}
		| \E(T_{3ij} \mid D^n) | \ & \lesssim \E\left[ |\widehat\omega_i - \omega_i|| \mu_j - \widehat\mu_j| K_{hi}(\widehat\omega_j)K_{hi}(\widehat\mu_j) \{|\widehat\omega_i - \widehat\omega_j| \ + \ |\widehat\mu_i - \widehat\mu_j|\} \mid D^n \right] \\
		& \lesssim h\E\left\{ |\widehat\omega_i - \omega_i| |\mu_j - \widehat\mu_j| K_{hi}(\widehat\omega_j)K_{hi}(\widehat\mu_j) \mid D^n \right\} \\
		& \lesssim h\|\widehat\omega - \omega\| \| \widehat\mu - \mu\|\\
		& \lesssim h\|\widehat\mu -\mu\|.
	\end{align*}
	where the last inequality follows by the Cauchy-Schwarz inequality. 
	
	Putting everything together, we have reached that
	\begin{align*}
		\left|\E(\widehat\psi - \psi \mid D^n) \right| & = \E\left(\frac{1}{n(n-1)}\mathop{\sum\sum}\limits_{1\leq i \neq j \leq n} T_{1ij} + T_{2ij} + T_{3ij} \mid D^n \right) \\
		& \lesssim h^\beta\|\widehat\mu - \mu\| \ + \ \|\widehat\mu - \mu\|^{1+\beta} \ + \ \|\widehat\omega - \omega\| \|\widehat\mu - \mu\|(nh^2)^{-1/2}.
	\end{align*}
	Next, we proceed with a different bound on the bias.
	
	\textbf{Bound 2.} Recall that
	\begin{align*}
		& T_n = \frac{1}{n(n-1)} \mathop{\sum\sum}\limits_{1\leq i \neq j \leq n} (A\widehat\omega_i - 1) \{\widehat{Q}(\widehat\omega_i, \widehat\mu_i)\}^{-1} K_{hi}(\widehat\omega_j) K_{hi}(\widehat\mu_j) A_j(Y_j - \widehat\mu_j) \equiv \frac{1}{n(n-1)} \mathop{\sum\sum}\limits_{1\leq i \neq j \leq n} T_{ij}.
	\end{align*}
	so that
	\begin{align*}
		& \E(T_n \mid D^n) = \E\left[\frac{1}{n(n-1)} \mathop{\sum\sum}\limits_{1 \leq i \neq j \leq n} (A_i\widehat\omega_i - 1) \{\widehat{Q}(\widehat\omega_i, \widehat\mu_i)\}^{-1} K_{hi}(\widehat\omega_j) K_{hi}(\widehat\mu_j) A_jf_\omega(\widehat\omega_j, \omega_j, \widehat\mu_j;D^n)\mid D^n \right] \\
		& = \E\left[\frac{1}{n(n-1)} \mathop{\sum\sum}\limits_{1 \leq i \neq j \leq n} (A_i\widehat\omega_i - 1) f_\omega(\widehat\omega_i, \omega_i, \widehat\mu_i;D^n)\mid D^n \right] \\
		& \hphantom{=} + \E\left[\frac{1}{n(n-1)} \mathop{\sum\sum}\limits_{1 \leq i \neq j \leq n} (A_i\widehat\omega_i - 1) \frac{K_{hi}(\widehat\omega_j) K_{hi}(\widehat\mu_j)}{\widehat{Q}(\widehat\omega_i, \widehat\mu_i)} A_j\{f_\omega(\widehat\omega_j, \omega_j, \widehat\mu_j; D^n) - f_\omega(\widehat\omega_i, \omega_i, \widehat\mu_i;D^n)\}\mid D^n \right] \\
		& = R_n + \E\left[\frac{1}{n(n-1)} \mathop{\sum\sum}\limits_{1 \leq i \neq j \leq n} A_i(\widehat\omega_i - \omega_i) \frac{K_{hi}(\widehat\omega_j) K_{hi}(\widehat\mu_j)}{\widehat{Q}(\widehat\omega_i, \widehat\mu_i)} A_j\{f_\omega(\widehat\omega_j, \omega_j, \widehat\mu_j;D^n) - f_\omega(\widehat\omega_i, \omega_i, \widehat\mu_i;D^n)\}\mid D^n \right]
	\end{align*}
	The first equality follows because
	\begin{align*}
		& \E\{A_j(Y_j - \widehat\mu_j) - A_jf_\omega(\widehat\omega_j, \omega_j, \widehat\mu_j; D^n) \mid A_j, \widehat\omega_j, \widehat\mu_j, A_i, \widehat\omega_i, \widehat\mu_i, D^n\} \\
		& = \E\{A_j(Y_j - \widehat\mu_j) - A_jf_\omega(\widehat\omega_j, \omega_j, \widehat\mu_j) \mid A_j, \widehat\omega_j, \widehat\mu_j, D^n\} \\
		& = \E\left[ \E\{A_j(Y_j - \widehat\mu_j) \mid A_j, \widehat\omega_j, \omega_j, \widehat\mu_j, D^n\}- A_jf_\omega(\widehat\omega_j, \omega_j, \widehat\mu_j; D^n) \mid A_j, \widehat\omega_j, \widehat\mu_j, D^n\right] \\
		& = \E\left( \E[\E\{A_j(Y_j - \widehat\mu_j) \mid A_j, \widehat\omega_j, \omega_j, \widehat\mu_j, X_j, D^n\} \mid A_j, \widehat\omega_j, \omega_j, \widehat\mu_j, D^n]- A_jf_\omega(\widehat\omega_j, \omega_j, \widehat\mu_j;D^n) \mid A_j, \widehat\omega_j, \widehat\mu_j, D^n\right) \\
		& = \E\left[ \E\{A_j(\mu_j - \widehat\mu_j) \mid A_j, \widehat\omega_j, \omega_j, \widehat\mu_j, D^n\}- A_jf_\omega(\widehat\omega_j, \omega_j, \widehat\mu_j; D^n) \mid A_j, \widehat\omega_j, \widehat\mu_j, D^n\right] \\
		& = \E\left[ A_j f_\omega(\widehat\omega_j, \omega_j, \widehat\mu_j; D^n) - A_jf_\omega(\widehat\omega_j, \omega_j, \widehat\mu_j; D^n) \mid A_j, \widehat\omega_j, \widehat\mu_j, D^n\right] \\
		& = 0.
	\end{align*}
	The last equality follows because $\widehat{Q}(\widehat\omega_i, \widehat\mu_i)$ does not depend on $A_i$ so that
	\begin{align*}
		\E\{(A_i\widehat\omega_i - 1)\widehat{Q}^{-1}(\widehat\omega_i, \widehat\mu_i) \mid X^n, A^n_{-i}, D^n\} & = \left(\frac{\widehat\omega_i}{\omega_i} - 1\right) \widehat{Q}^{-1}(\widehat\omega_i, \widehat\mu_i) \\
		& = \E\{(A_i(\widehat\omega_i - \omega_i)\widehat{Q}^{-1}(\widehat\omega_i, \widehat\mu_i) \mid X^n, A^n_{-i}, D^n\},
	\end{align*}
	where $X^n = (X_1, \ldots, X_n)$ and $A^n_{-i} = A^n \setminus A_i$. 
	
	Next, by the smoothness assumption on $f_\omega$ and the boundedness assumption on $\widehat{Q}^{-1}(\widehat\omega_i, \widehat\mu_i)$, we have
	\begin{align*}
		\left| f_\omega(\widehat\omega_j, \omega_j, \widehat\mu_j; D^n) - f_\omega(\widehat\omega_i, \omega_i, \widehat\mu_i; D^n) \right| & \lesssim  |\widehat\omega_j - \widehat\omega_i|^\alpha \ + \ |\omega_j - \omega_i|^\alpha \ + \ |\widehat\mu_j - \widehat\mu_i|^\alpha \\
		& \lesssim |\widehat\omega_j - \widehat\omega_i|^\alpha \ + \ |\widehat\omega_j - \omega_j|^\alpha \ + \ |\widehat\omega_i - \omega_i|^\alpha \ + \ |\widehat\mu_j - \widehat\mu_i|^\alpha
	\end{align*}
	so that
	\begin{align*}
		& \left| \E\left[\frac{1}{n(n-1)} \mathop{\sum\sum}\limits_{1 \a i \neq j \leq n} A_i(\widehat\omega_i - \omega_i) \frac{K_{hi}(\widehat\omega_j) K_{hi}(\widehat\mu_j)}{\widehat{Q}(\widehat\omega_i, \widehat\mu_i)} A_j\{f_\omega(\widehat\omega_j, \omega_j, \widehat\mu_j; D^n) - f_\omega(\widehat\omega_i, \omega_i, \widehat\mu_i; D^n)\}\mid D^n \right] \right| \\
		& \lesssim  \E\left[|\widehat\omega_i - \omega_i| A_iA_jK_{hi}(\widehat\omega_j) K_{hi}(\widehat\mu_j) \{|\widehat\omega_j - \widehat\omega_i|^\alpha \ + \ |\widehat\omega_j - \omega_j|^\alpha \ + \ |\widehat\omega_i - \omega_i|^\alpha \ + \ |\widehat\mu_j - \widehat\mu_i|^\alpha\} \mid D^n \right] \\
		& \lesssim h^\alpha \E\left[|\widehat\omega_i - \omega_i|\E\{A_jK_{hi}(\widehat\omega_j) K_{hi}(\widehat\mu_j) \mid X_i, D^n\} \mid D^n \right] \\
		& \hphantom{\lesssim} + \E\left\{|\widehat\omega_i - \omega_i| A_iA_jK_{hi}(\widehat\omega_j) K_{hi}(\widehat\mu_j)|\widehat\omega_j - \omega_j|^\alpha \mid D^n \right\} \\
		& \hphantom{\lesssim} + \E\left[(\widehat\omega_i - \omega_i)^{1+\alpha}\E\{A_jK_{hi}(\widehat\omega_j) K_{hi}(\widehat\mu_j) \mid X_i, D^n\} \right] \\
		& \lesssim h^\alpha \| \widehat\omega - \omega\| \ + \ \|\widehat\omega - \omega\|^{1+\alpha}
	\end{align*}
	under the condition that $\E\{A_jK_{hi}(\widehat\omega_j) K_{hi}(\widehat\mu_j) \mid X_i, D^n\}$. Thus, we also have that
	\begin{align*}
		|\E(\widehat\psi - \psi \mid D^n)| \ = |\E(T_n \mid D^n) - R_n| \ \lesssim h^\alpha \|\widehat\omega - \omega\| \ + \ \|\widehat\omega - \omega\|^{1+\alpha}.
	\end{align*}
	Combining \textbf{Bound 1} and \textbf{Bound 2}, we conclude that
	\begin{align*}
		& |\E(\widehat\psi - \psi \mid D^n)| \ \lesssim \left(h^\beta\|\widehat\mu - \mu\| \ + \ \|\widehat\mu - \mu\|^{1+\beta} \ + \ \|\widehat\omega - \omega\| \|\widehat\mu - \mu\|(nh^2)^{-1/2}\right) \land \left(h^\alpha \|\widehat\omega - \omega\| \ + \ \|\widehat\omega - \omega\|^{1+\alpha}\right).
	\end{align*}
	This concludes our proof of the bound on the bias of $\widehat\psi$.
	\subsubsection{Variance of $\widehat\psi$}
	The bound on the variance follows from the same reasoning used to bound the variance of $\widehat\psi_\omega$ in Section \ref{sec:variance_psi_a}. In particular, it follows from Lemmas \ref{lemma:Tn} and \ref{lemma:Sn} as well as the decomposition of $T_n$ as
	\begin{align*}
		T_n = (\Pn - \Pb)(\kappa_1 + \kappa_2) + S_n + \E\{\kappa(Z_1, Z_2; D^n) \mid D^n\} + T_n  - \widetilde{T}_n.
	\end{align*}
	\subsubsection{Linear expansion of $\widehat\psi - \psi$}
	We have defined 
	\begin{align*}
	&	\widehat\varphi_{\omega\mu} = \E\{\widehat\omega - \omega\mid A = 1, \widehat\omega, \widehat\mu, D^n\}A(Y - \widehat\mu) + \E\{\mu - \widehat\mu \mid A = 1, \widehat\omega, \widehat\mu, D^n\}(A\widehat\omega - 1), \\
	&	\overline\varphi_{\omega\mu} = \E\{\overline\omega - \omega\mid A = 1, \overline\omega, \overline\mu\}A(Y - \overline\mu) + \E\{\mu - \overline\mu \mid A = 1, \overline\omega, \overline\mu\}(A\overline\omega - 1).
	\end{align*}
	Following the reasoning in Section \ref{sec:proof_prop_max_squared_a_le}, the third statement of Proposition \ref{prop:min_squared} is implied by Lemmas \ref{lemma:Tn}, \ref{lemma:Sn} and \ref{lemma:RTn} as well as the following statement:
	\begin{align} \label{eq:statement3}
		& \| \kappa_1 + \kappa_2 - \widehat\varphi_{\omega\mu}\| \ \lesssim \{(h^\alpha + \|\widehat\omega - \omega\|_\infty) \land \|\widehat\mu - \mu\|_\infty\} + \{(h^\beta + \|\widehat\mu - \mu\|_\infty) \land \|\widehat\omega - \omega\|_\infty)\} + h.
	\end{align}
	\subsubsection{Proof of Lemma \ref{lemma:RTn}}
	The proof follows by similar arguments as those used to derive the bound on the bias of $\widehat\psi$ in Section \ref{sec:bias_psi}.
	\subsubsection{Proof of Lemma \ref{lemma:Tn}}
	Recall the definitions
	\begin{align*}
		& \widehat{Q}(\widehat{a_i}, \widehat\mu_i) = \frac{1}{n-1}\sum_{s = 1, \ s \neq i}^n A_sK_{hi}(\widehat\omega_s)K_{hi}(\widehat\mu_s) \\
		& \widehat{Q}_{-l}(\widehat\omega_i, \widehat\mu_i) = \frac{1}{n-1}\sum_{s = 1, \ s \neq (i, l)}^n A_sK_{hi}(\widehat\omega_s)K_{hi}(\widehat\mu_s) \\
		& Q(\widehat\omega_i, \widehat\mu_i) = \E\{\widehat{Q}(\widehat\omega_i, \widehat\mu_i) \mid X_i, D^n\} \\
		& \hphantom{Q(\widehat\omega_i, \widehat\mu_i)} = \int \{\omega(x)\}^{-1} K_h(\widehat\omega(x) - \omega(X_i))K_h(\widehat\mu(x) - \mu(X_i)) d\Pb(x) \\
		& \hphantom{Q(\widehat\omega_i, \widehat\mu_i)} = \Pb(A = 1) \int K(u)K(v) d\Pb_{\widehat\omega, \widehat\mu \mid A = 1, D^n}(uh + \widehat\omega(X_i), vh + \widehat\mu(X_i))
	\end{align*}
	We have
	\begin{align*}
		\widetilde{T}_n - T_n & = \frac{1}{n(n-1)}\mathop{\sum\sum}\limits_{1 \leq i \neq j \leq n} (A_i\widehat\omega_i - 1) \{Q^{-1}(\widehat\omega_i, \widehat\mu_i) - \widehat{Q}^{-1}(\widehat\omega_i, \widehat\mu_i)\} K_{hi}(\widehat\omega_j) K_{hi}(\widehat\mu_j) A_j(Y_j - \widehat\mu_j) \\
		& = \frac{1}{n(n-1)}\mathop{\sum\sum}\limits_{1 \leq i \neq j \leq n} T_{ij} 
	\end{align*}
	Just like in the proof of Lemmas \ref{lemma:Tna} and \ref{lemma:Tnb}, we break the square of the double sum above in seven terms:
	\begin{align*}
		\left(\mathop{\sum\sum}\limits_{1\leq i \neq j \leq n} T_{ij}\right)^2 & = \mathop{\sum\sum}\limits_{1 \leq i \neq j \leq n} (T^2_{ij} + T_{ij} T_{ji}) + \mathop{\sum\sum\sum}\limits_{1 \leq i \neq j \neq l \leq n} (T_{ij} T_{il} + T_{ij} T_{li} + T_{ij} T_{jl} + T_{ij} T_{lj}) \\
		& \hphantom{=} + \mathop{\sum\sum\sum\sum}\limits_{1 \leq i \neq j \neq l \neq k \leq n} T_{ij} T_{lk}.
	\end{align*}
	We have
	\begin{align*}
		T^2_{ij} \lesssim A_jh^{-2}K_{hi}(\widehat\omega_i)K_{hi}(\widehat\mu_i) \quad \text{and} \quad |T_{ij}T_{ji}| \ \lesssim A_iA_jh^{-2}K_{hi}(\widehat\omega_i)K_{hi}(\widehat\mu_i).
	\end{align*}
	In this light,
	\begin{align*}
		\E\left(\mathop{\sum\sum}\limits_{1 \leq i \neq j \leq n} T^2_{ij} + \mathop{\sum\sum}\limits_{1 \leq i \neq j \leq n} T_{ij} T_{ji} \mid D^n \right) \lesssim n(n-1)h^{-2},
	\end{align*}
	under the condition that $\E\{A_jK_{hi}(\widehat\omega_i)K_{hi}(\widehat\mu_i)\mid X_i, D^n\} \lesssim 1$. Next, notice that
	\begin{align*}
		& |T_{ij}T_{il}| \ \lesssim A_jA_lK_{hi}(\widehat\omega_j)K_{hi}(\widehat\mu_j)K_{hi}(\widehat\omega_l)K_{hi}(\widehat\mu_l)\{Q(\widehat\omega_i, \widehat\mu_i) - \widehat{Q}(\widehat\omega_i, \widehat\mu_i)\}^2 \\
		& |T_{ij}T_{li}| \ \lesssim A_iA_jK_{hi}(\widehat\omega_j)K_{hi}(\widehat\mu_j)K_{hi}(\widehat\omega_l)K_{hi}(\widehat\mu_l)|Q(\widehat\omega_i, \widehat\mu_i) - \widehat{Q}(\widehat\omega_i, \widehat\mu_i)||Q(\widehat\omega_l, \widehat\mu_l) - \widehat{Q}(\widehat\omega_l, \widehat\mu_l)| \\
		& |T_{ij}T_{jl}| \ \lesssim A_jA_lK_{hi}(\widehat\omega_j)K_{hi}(\widehat\mu_j)K_{hj}(\widehat\omega_l)K_{hj}(\widehat\mu_l)|Q(\widehat\omega_i, \widehat\mu_i) - \widehat{Q}(\widehat\omega_i, \widehat\mu_i)||Q(\widehat\omega_j, \widehat\mu_j) - \widehat{Q}(\widehat\omega_j, \widehat\mu_j)| \\
		& |T_{ij}T_{lj}| \ \lesssim A_jK_{hi}(\widehat\omega_j)K_{hi}(\widehat\mu_j)K_{hj}(\widehat\omega_l)K_{hj}(\widehat\mu_l)|Q(\widehat\omega_i, \widehat\mu_i) - \widehat{Q}(\widehat\omega_i, \widehat\mu_i)||Q(\widehat\omega_l, \widehat\mu_l) - \widehat{Q}(\widehat\omega_l, \widehat\mu_l)| .
	\end{align*}
	By Lemma \ref{lemma:Q}, we have
	\begin{align*}
		\E\left\{\left|Q(\widehat\omega_i, \widehat\mu_i) - \widehat{Q}(\widehat\omega_i, \widehat\mu_i)\right| \mid A_j, A_l, X_i, X_j, X_l, D^n\right\} \lesssim (nh^2)^{-1/2}.
	\end{align*}
	In this light, we have
	\begin{align*}
		& \E(|T_{ij}T_{il}| \mid D^n) \lesssim  \E\left[A_jA_lK_{hi}(\widehat\omega_j)K_{hi} (\widehat\mu_j)K_{hi}(\widehat\omega_l)K_{hi}(\widehat\mu_l)\{Q(\widehat\omega_i, \widehat\mu_i) - \widehat{Q}(\widehat\omega_i, \widehat\mu_i)\}^2 \mid D^n\right] \\
		& = \E\left(A_jA_lK_{hi}(\widehat\omega_j)K_{hi}(\widehat\mu_j)K_{hi}(\widehat\omega_l)K_{hi}(\widehat\mu_l)\E\left[\{Q(\widehat\omega_i, \widehat\mu_i) - \widehat{Q}(\widehat\omega_i, \widehat\mu_i)\}^2 \mid A_j, A_l, X_i, X_j, X_l, D^n \right]\mid D^n\right) \\
		& = (nh^2)^{-1}\E\left[A_jK_{hi}(\widehat\omega_j)K_{hi}(\widehat\mu_j)\E\{A_lK_{hi}(\widehat\omega_l)K_{hi}(\widehat\mu_l) \mid X_i, D^n\} \mid D^n\right] \\
		& \lesssim (nh^2)^{-1}\E\{A_jK_{hi}(\widehat\omega_j)K_{hi}(\widehat\mu_j) \mid D^n\} \\
		& = (nh^2)^{-1}\E\left[\E\{A_jK_{hi}(\widehat\omega_j)K_{hi}(\widehat\mu_j) \mid X_i, D^n\} \mid D^n \right] \\
		& \lesssim (nh^2)^{-1}.
	\end{align*}
	The expectation of the terms $|T_{ij}T_{li}|$ and $|T_{ij}T_{jl}|$ can be similarly upper bounded by a $(nh^2)^{-1}$. Therefore, we have
	\begin{align*}
		\E\left\{\frac{1}{n^2(n-1)^2} \mathop{\sum\sum\sum}\limits_{1 \leq i \neq j \neq l \leq n} (T_{ij} T_{il} + T_{ij} T_{li} + T_{ij} T_{jl}) \mid D^n \right\} \lesssim (nh)^{-2}.
	\end{align*}
	The remaining term to control is that involving the terms $T_{ij}T_{lk}$, possibly with $j = k$. We have
	\begin{align*}
		\widehat{Q}(\widehat\omega_i, \widehat\mu_i) & = \frac{1}{n-1}\sum_{s = 1, s \neq i}^n A_s K_{hi}(\widehat\omega_s)K_{hi}(\widehat\mu_j) \\
		& = \frac{1}{n-1}\sum_{s = 1, s \neq (i, l)}^n A_s K_{hi}(\widehat\omega_s)K_{hi}(\widehat\mu_s) + \frac{A_l K_{hi}(\widehat\omega_l)K_{hi}(\widehat\mu_l)}{n-1} \\
		& \equiv \widehat{Q}_{-l}(\widehat\omega_i, \widehat\mu_i) + \frac{A_l K_{hi}(\widehat\omega_l)K_{hi}(\widehat\mu_l)}{n-1}.
	\end{align*}
	Next, we write
	\begin{align*}
		T_{ij}T_{lk} & = (A_i\widehat\omega_i - 1)\{\widehat{Q}^{-1}_{-l}(\widehat\omega_i, \widehat\mu_i) - Q^{-1}(\widehat\omega_i, \widehat\mu_i)\}K_{hi}(\widehat\omega_j)K_{hi}(\widehat\mu_j)A_j(Y_j - \widehat\mu_j) \\
		& \hphantom{=} \quad\quad \times (A_l\widehat\omega_l - 1)\{\widehat{Q}^{-1}_{-i}(\widehat\omega_l, \widehat\mu_l) - Q^{-1}(\widehat\omega_l, \widehat\mu_l)\}K_{hl}(\widehat\omega_k)K_{hl}(\widehat\mu_k)A_k(Y_k - \widehat\mu_k) \\
		& \hphantom{=} - (A_i\widehat\omega_i - 1)\frac{A_lK_{hi}(\widehat\omega_l)K_{hi}(\widehat\mu_l)K_{hi}(\widehat\omega_j)K_{hi}(\widehat\mu_j)}{(n-1)\widehat{Q}(\widehat\omega_i, \widehat\mu_i)\widehat{Q}_{-l}(\widehat\omega_i, \widehat\mu_l)}A_j(Y_j - \widehat\mu_j) \\
		& \hphantom{=} \quad\quad\times (A_l\widehat\omega_l - 1)\{\widehat{Q}^{-1}(\widehat\omega_l, \widehat\mu_l) - Q^{-1}(\widehat\omega_l, \widehat\mu_l)\}K_{hl}(\widehat\omega_k)K_{hl}(\widehat\mu_k)A_k(Y_k - \widehat\mu_k) \\
		& \hphantom{=} - (A_i\widehat\omega_i - 1)\{\widehat{Q}^{-1}_{-l}(\widehat\omega_i, \widehat\mu_i) - Q^{-1}(\widehat\omega_i, \widehat\mu_i)\}K_{hi}(\widehat\omega_j)K_{hi}(\widehat\mu_j)A_j(Y_j - \widehat\mu_j) \\
		& \hphantom{=} \quad\quad \times (A_l\widehat\omega_l - 1)\frac{A_iK_{hl}(\widehat\omega_i)K_{hl}(\widehat\mu_i)K_{hl}(\widehat\omega_k)K_{hl}(\widehat\mu_k)}{(n-1)\widehat{Q}(\widehat\omega_l, \widehat\mu_l)\widehat{Q}_{-i}(\widehat\omega_l, \widehat\mu_l)}A_k(Y_k - \widehat\mu_k)
	\end{align*}
	When $j = k$, the expectation of the last two terms can be upper bounded as
	\begin{align*}
		& \E \left[ \left|\frac{(A_i\widehat\omega_i - 1)(A_l\widehat\omega_l - 1)}{n-1}\frac{A_lK_{hi}(\widehat\omega_l, \widehat\mu_l)K_{hi}(\widehat\omega_j, \widehat\mu_j)K_{hl}(\widehat\omega_j, \widehat\mu_j) }{\widehat{Q}(\widehat\omega_i, \widehat\mu_i)\widehat{Q}_{-l}(\widehat\omega_i, \widehat\mu_l)Q(\widehat\omega_l, \widehat\mu_l)\widehat{Q}(\widehat\omega_l, \widehat\mu_l)}A_j(Y_j - \widehat\mu_j)^2 \{Q(\widehat\omega_l, \widehat\mu_l) - \widehat{Q}(\widehat\omega_l, \widehat\mu_l)\} \right| \right. \\
		& \left. + \left| \frac{(A_i\widehat\omega_i - 1)(A_l\widehat\omega_l - 1)}{n-1}\frac{A_iK_{hl}(\widehat\omega_i, \widehat\mu_i)K_{hl}(\widehat\omega_j, \widehat\mu_j)K_{hi}(\widehat\omega_j, \widehat\mu_j)}{\widehat{Q}(\widehat\omega_l, \widehat\mu_l)\widehat{Q}_{-i}(\widehat\omega_l, \widehat\mu_l)\widehat{Q}_{-l}(\widehat\omega_i, \widehat\mu_i)Q(\widehat\omega_i, \widehat\mu_i)}A_j(Y_j - \widehat\mu_j)^2\{Q(\widehat\omega_i, \widehat\mu_i) -\widehat{Q}_{-l}(\widehat\omega_i, \widehat\mu_i)\} \right| \mid D^n \right] \\
		& \lesssim \E\left[\frac{A_lA_j K_{hi}(\widehat\omega_l, \widehat\mu_l)K_{hi}(\widehat\omega_j, \widehat\mu_j)}{h^2(n-1)} \E\left\{\left|\widehat{Q}(\widehat\omega_l, \widehat\mu_l) - Q(\widehat\omega_l, \widehat\mu_l) \right| \mid A_l, A_j, X_l, X_j, X_i\right\} \mid D^n \right] \\
		& \hphantom{\lesssim} + \E\left[\frac{A_iA_jK_{hl}(\widehat\omega_i, \widehat\mu_i)K_{hi}(\widehat\omega_j, \widehat\mu_j)}{h^2(n-1)} \E\left\{\left|\widehat{Q}_{-l}(\widehat\omega_i, \widehat\mu_i) - Q(\widehat\omega_i, \widehat\mu_i) \right| \mid A_j, X_j, X_i\right\} \mid D^n \right] \\
		& \lesssim (nh^2)^{-3/2}
	\end{align*}
	The last inequality follows by Lemma \ref{lemma:Q} and because
	\begin{align*}
		\E\left\{A_iA_jK_{hl}(\widehat\omega_i, \widehat\mu_i)K_{hi}(\widehat\omega_j, \widehat\mu_j) \mid D^n\right\} & = \E\left[A_iK_{hl}(\widehat\omega_i, \widehat\mu_i)\E\{A_jK_{hi}(\widehat\omega_j, \widehat\mu_j) \mid X_i, D^n\} \mid D^n\right] \\
		& \lesssim \E\left\{A_iK_{hl}(\widehat\omega_i, \widehat\mu_i) \mid D^n\right\}\\
		& \lesssim 1.
	\end{align*}
	When $j \neq k$, the last two terms can be upper bounded as
	\begin{align*}
		& \E \left[ \left|(A_i\widehat\omega_i - 1)\frac{A_lK_{hi}(\widehat\omega_l, \widehat\mu_l)K_{hi}(\widehat\omega_j, \widehat\mu_j)}{(n-1)\widehat{Q}(\widehat\omega_i, \widehat\mu_i)\widehat{Q}_{-l}(\widehat\omega_i, \widehat\mu_l)}A_j(Y_j - \widehat\mu_j) \right. \right. \\
		& \hphantom{=} \quad\quad \left. \times \vphantom{\frac{A_lK_{hi}(\widehat\omega_l, \widehat\mu_l)K_{hi}(\widehat\omega_j, \widehat\mu_j)}{(n-1)\widehat{Q}(\widehat\omega_i, \widehat\mu_i)\widehat{Q}_{-l}(\widehat\omega_i, \widehat\mu_l)}} (A_l\widehat\omega_l - 1)\{\widehat{Q}^{-1}(\widehat\omega_l, \widehat\mu_l) - Q^{-1}(\widehat\omega_l, \widehat\mu_l)\}K_{hl}(\widehat\omega_k, \widehat\mu_k)A_k(Y_k - \widehat\mu_k) \right| \\
		& \hphantom{=} + \left| \vphantom{\frac{A_lK_{hi}(\widehat\omega_l, \widehat\mu_l)K_{hi}(\widehat\omega_j, \widehat\mu_j)}{(n-1)\widehat{Q}(\widehat\omega_i, \widehat\mu_i)\widehat{Q}_{-l}(\widehat\omega_i, \widehat\mu_l)}} (A_i\widehat\omega_i - 1)\{\widehat{Q}^{-1}_{-l}(\widehat\omega_i, \widehat\mu_i) - Q^{-1}(\widehat\omega_i, \widehat\mu_i)\}K_{hi}(\widehat\omega_j, \widehat\mu_j)A_j(Y_j - \widehat\mu_j)  \right. \\
		& \hphantom{=} \left. \left. \quad\quad \times (A_l\widehat\omega_l - 1)\frac{A_iK_{hl}(\widehat\omega_i, \widehat\mu_i)K_{hl}(\widehat\omega_k, \widehat\mu_k)}{(n-1)\widehat{Q}(\widehat\omega_l, \widehat\mu_l)\widehat{Q}_{-i}(\widehat\omega_l, \widehat\mu_l)}A_k(Y_k - \widehat\mu_k) \right| \mid D^n \right] \\
		& \lesssim \E\left[\frac{A_lA_jA_k K_{hi}(\widehat\omega_l, \widehat\mu_l)K_{hi}(\widehat\omega_j, \widehat\mu_j)K_{hl}(\widehat\omega_k, \widehat\mu_k)}{(n-1)} \E\left\{\left|\widehat{Q}(\widehat\omega_l, \widehat\mu_l) - Q(\widehat\omega_l, \widehat\mu_l) \right| \mid A_j, A_k, X_l, X_j, X_k, X_i\right\} \mid D^n \right] \\
		& \hphantom{\lesssim} + \E\left[\frac{A_iA_jA_k K_{hl}(\widehat\omega_i, \widehat\mu_i)K_{hi}(\widehat\omega_j, \widehat\mu_j)K_{hl}(\widehat\omega_k, \widehat\mu_k)}{(n-1)} \E\left\{\left|\widehat{Q}_{-l}(\widehat\omega_i, \widehat\mu_i) - Q(\widehat\omega_i, \widehat\mu_i) \right| \mid A_j, A_k, X_l, X_j, X_k, X_i\right\} \mid D^n \right] \\
		& \lesssim (n^3h^2)^{-1/2}
	\end{align*}
	Next, we bound the first term. For $j = k$, we have
	\begin{align*}
		& \left| \E\left[(A_i\widehat\omega_i - 1)\{\widehat{Q}^{-1}_{-l}(\widehat\omega_i, \widehat\mu_i) - Q^{-1}(\widehat\omega_i, \widehat\mu_i)\}K_{hi}(\widehat\omega_j, \widehat\mu_j)A_j(Y_j - \widehat\mu_j)^2 \right. \right. \\
		& \left. \left. \quad\quad \times (A_l\widehat\omega_l - 1)\{\widehat{Q}^{-1}_{-i}(\widehat\omega_l, \widehat\mu_l) - Q^{-1}(\widehat\omega_l, \widehat\mu_l)\}K_{hl}(\widehat\omega_j, \widehat\mu_j) \mid D^n \right] \right| \\
		& = \left| \E\left[A_i(\widehat\omega_i - \omega_i)\{\widehat{Q}^{-1}_{-l}(\widehat\omega_i, \widehat\mu_i) - Q^{-1}(\widehat\omega_i, \widehat\mu_i)\}K_{hi}(\widehat\omega_j, \widehat\mu_j)A_j(Y_j - \widehat\mu_j)^2 \right. \right. \\
		& \left. \left. \quad\quad\quad \times A_l(\widehat\omega_l - \omega_l)\{\widehat{Q}^{-1}_{-i}(\widehat\omega_l, \widehat\mu_l) - Q^{-1}(\widehat\omega_l, \widehat\mu_l)\}K_{hl}(\widehat\omega_j, \widehat\mu_j)\mid D^n \right] \right| \\
		& \lesssim \left| \E\left[A_iA_jA_lK_{hi}(\widehat\omega_j, \widehat\mu_j)K_{hl}(\widehat\omega_j, \widehat\mu_j) \right. \right. \\
		& \left. \left. \quad\quad\quad \times  \E\left\{|\widehat{Q}_{-l}(\widehat\omega_i, \widehat\mu_i) - Q(\widehat\omega_i, \widehat\mu_i)| | \widehat{Q}_{-i}(\widehat\omega_l, \widehat\mu_l) - Q(\widehat\omega_l, \widehat\mu_l)| \mid A_j, X_i, X_j, X_l, D^n \right\} \right] \mid D^n \right| \\
		& \lesssim (nh^2)^{-1}
	\end{align*}
	This means that, for $k = j$, we have the bound
	\begin{align*}
		|\E(T_{ij}T_{lj} \mid D^n)| \lesssim (nh^2)^{-1},
	\end{align*}
	so that
	\begin{align*}
		\left| \frac{1}{n^2(n-1)^2} \E\left\{\mathop{\sum\sum\sum}\limits_{1 \leq i \neq j \neq l \leq n} T_{ij} T_{lj} \mid D^n \right\} \right| \lesssim (nh)^{-2}.
	\end{align*}
	Finally, for $k \neq j$, we have
	\begin{align*}
		& \left| \E\left[(A_i\widehat\omega_i - 1)\{\widehat{Q}^{-1}_{-l}(\widehat\omega_i, \widehat\mu_i) - Q^{-1}(\widehat\omega_i, \widehat\mu_i)\}K_{hi}(\widehat\omega_j, \widehat\mu_j)A_j(Y_j - \widehat\mu_j) \right. \right. \\
		& \left. \left. \quad\quad \times (A_l\widehat\omega_l - 1)\{\widehat{Q}^{-1}_{-i}(\widehat\omega_l, \widehat\mu_l) - Q^{-1}(\widehat\omega_l, \widehat\mu_l)\}K_{hl}(\widehat\omega_k, \widehat\mu_k)A_k(Y_k - \widehat\mu_k) \mid D^n \right] \right| \\
		& = \left| \E\left[A_i(\widehat\omega_i - \omega_i)\{\widehat{Q}^{-1}_{-l}(\widehat\omega_i, \widehat\mu_i) - Q^{-1}(\widehat\omega_i, \widehat\mu_i)\}K_{hi}(\widehat\omega_j, \widehat\mu_j)A_j(\mu_j - \widehat\mu_j) \right. \right. \\
		& \left. \left. \quad\quad\quad \times A_l(\widehat\omega_l - \omega_l)\{\widehat{Q}^{-1}_{-i}(\widehat\omega_l, \widehat\mu_l) - Q^{-1}(\widehat\omega_l, \widehat\mu_l)\}K_{hl}(\widehat\omega_k, \widehat\mu_k)A_k(\mu_k - \widehat\mu_k) \mid D^n \right] \right| \\
		& \lesssim \left| \E\left[|\widehat\omega_i - \omega_i||\mu_j - \widehat\mu_j||\widehat\omega_l - \omega_l||\mu_k - \widehat\mu_k|A_iA_jA_kA_lK_{hi}(\widehat\omega_j, \widehat\mu_j)K_{hl}(\widehat\omega_k, \widehat\mu_k) \right. \right. \\
		& \left. \left. \quad\quad\quad \times  \E\left\{|\widehat{Q}_{-l}(\widehat\omega_i, \widehat\mu_i) - Q(\widehat\omega_i, \widehat\mu_i)| | \widehat{Q}_{-i}(\widehat\omega_l, \widehat\mu_l) - Q(\widehat\omega_l, \widehat\mu_l)| \mid A_j, A_k, X_i, X_j, X_l, X_k, D^n \right\} \right] \mid D^n \right| \\
		& \lesssim (nh^2)^{-1}\E\left[|\widehat\omega_i - \omega_i|A_iA_j
		K_{hi}(\widehat\omega_j)K_{hi}(\widehat\mu_j)|\widehat\mu_j - \mu_j| \mid D^n\right] \E\left[A_kA_lK_{hl}(\widehat\omega_k)K_{hl}(\widehat\mu_k)|\widehat\omega_l - \omega_l| |\widehat\mu_k - \mu_k| \mid D^n \right] \\
		& \lesssim (nh^2)^{-1}\|\widehat\omega - \omega\|^2\|\widehat\mu - \mu\|^2,
	\end{align*}
	where the second-to-last inequality relies on Lemma \ref{lemma:Q}. 
	Putting everything together, we have that
	\begin{align*}
		& \left| \E(T_{ij}T_{lk} \mid D^n) \right| \lesssim \|\widehat\omega - \omega\|^2\|\widehat\mu - \mu\|^2(nh)^{-1} + (n^3h^2)^{-1/2}.
	\end{align*}
	so that we have reached
	\begin{align*}
		\left| \frac{1}{n^2(n-1)^2} \E\left(\mathop{\sum\sum\sum\sum}\limits_{1 \leq i \neq j \neq l \neq k \leq n} T_{ij} T_{lk} \mid D^n\right)\right| \lesssim \|\widehat\omega - \omega\|^2 \|\widehat\mu- \mu\|^2(nh^2)^{-1} + (n^3h^2)^{-1/2}
	\end{align*}
	This concludes our proof that
	\begin{align*}
		\E\{(\widetilde{T}_n - T_n)^2 \mid D^n\} \lesssim (n^3h^2)^{-1/2} + \|\widehat\omega - \omega\|^2 \|\widehat\mu- \mu\|^2(nh^2)^{-1}.
	\end{align*}
	\subsubsection{Proof of the bound in Eq. \eqref{eq:statement3}}
	We have
	\begin{align*}
		& \kappa_1(Z_1; D^n) = \E\{\kappa(Z_1, Z_2; D^n) \mid Z_1, D^n\} \\
		& = (A_1\widehat\omega_1 - 1) \E\left\{\frac{K_{h1}(\widehat\omega_2)K_{h1}(\widehat\mu_2)}{Q(\widehat\omega_1, \widehat\mu_1)} A_2(Y_2 - \widehat\mu_2) \mid Z_1, D^n\right\} \\
		& = (A_1\widehat\omega_1 - 1) \E\left\{\frac{K_{h1}(\widehat\omega_2)K_{h1}(\widehat\mu_2)}{Q(\widehat\omega_1, \widehat\mu_1)} A_2\E(\mu_2 - \widehat\mu_2 \mid A_2 = 1, \widehat\omega_2, \widehat\mu_2, D^n) \mid D^n, Z_1\right\} \\
		& = (A_1\widehat\omega_1 - 1)\E(\mu - \widehat\mu \mid A = 1, \widehat\omega_1, \widehat\mu_1, D^n) \\
		& \hphantom{=} + (A_1\widehat\omega_1 - 1) \E\left[\frac{K_{h1}(\widehat\omega_2)K_{h1}(\widehat\mu_2)}{Q(\widehat\omega_1, \widehat\mu_1)} A_2\{\E(\mu - \widehat\mu \mid A = 1, \widehat\omega_2, \widehat\mu_2, D^n) - \E(\mu - \widehat\mu \mid A = 1, \widehat\omega_1, \widehat\mu_1, D^n)\} \mid D^n, Z_1\right]
	\end{align*}
	Similarly, we have
	\begin{align*}
		& \kappa_2(Z_2; D^n) = \E\{\kappa(Z_1, Z_2; D^n) \mid Z_2, D^n\} \\
		& = \E\left\{(A_1\widehat\omega_1 - 1)\frac{K_{h1}(\widehat\omega_2)K_{h1}(\widehat\mu_2)}{Q(\widehat\omega_1, \widehat\mu_1)} \mid Z_2, D^n\right\}A_2(Y_2 - \widehat\mu_2) \\
		& = \E\left\{A_1\E(\widehat\omega - \omega \mid A = 1, \widehat\mu_1, \widehat\omega_1, D^n)\frac{K_{h1}(\widehat\omega_2)K_{h1}(\widehat\mu_2)}{Q(\widehat\omega_1, \widehat\mu_1)} \mid Z_2, D^n\right\}A_2(Y_2 - \widehat\mu_2) \\
		& = \E(\widehat\omega - \omega \mid A = 1, \widehat\mu_2, \widehat\omega_2, D^n)A_2(Y_2 - \widehat\mu_2) \\
		& \hphantom{=} + \E\left[A_1\{\E(\widehat\omega_1 - \omega_1 \mid A_1 = 1, \widehat\mu_1, \widehat\omega_1, D^n) - \E(\widehat\omega_2 - \omega_2 \mid A_2 = 1, \widehat\mu_2, \widehat\omega_2, D^n)\}\frac{K_{h1}(\widehat\omega_2)K_{h1}(\widehat\mu_2)}{Q(\widehat\omega_2, \widehat\mu_2)} \mid Z_2, D^n\right]\\
		& \hphantom{\hphantom{=}} \quad\quad\quad\quad \times A_2(Y_2 - \widehat\mu_2) \\
        & \hphantom{=} + \E\left[A_1\E(\widehat\omega_1 - \omega_1 \mid A_1 = 1, \widehat\mu_1, \widehat\omega_1, D^n)\{Q^{-1}(\widehat\omega_1, \widehat\mu_1) - Q^{-1}(\widehat\omega_2, \widehat\mu_2)\}K_{h1}(\widehat\omega_2)K_{h1}(\widehat\mu_2) \mid Z_2, D^n\right]A_2(Y_2 - \widehat\mu_2)
	\end{align*}
 In virtue of the Lipschitz assumption on the density of $(\widehat\omega(X), \widehat\mu(X))$, the last term can be upper bounded by a constant multiple of $\|\widehat\omega - \omega\|_\infty h$. Next, we have
	\begin{align*}
		\E(\mu_2 - \widehat\mu_2 \mid A_2 = 1, \widehat\omega_2, \widehat\mu_2, D^n) & = \E\{f_\omega(\widehat\omega_2, \omega_2, \widehat\mu_2; D^n) \mid A_2 = 1, \widehat\omega_2, \widehat\mu_2, D^n\} \\
		& = f_\omega(\widehat\omega_2, \widehat\omega_2, \widehat\mu_2; D^n) \\
		& \hphantom{=} + \E\{f_\omega(\widehat\omega_2, \omega_2, \widehat\mu_2; D^n) - f_\omega(\widehat\omega_2, \widehat\omega_2, \widehat\mu_2; D^n) \mid A_2 = 1, \widehat\omega_2, \widehat\mu_2, D^n\} 
	\end{align*}
	Similarly, we have
	\begin{align*}
		\E(\mu_1 - \widehat\mu_1 \mid A_1 = 1, \widehat\omega_1, \widehat\mu_1, D^n) & = \E\{f_\omega(\widehat\omega_1, \omega_1, \widehat\mu_1; D^n) \mid A_1 = 1, \widehat\omega_1, \widehat\mu_1, D^n\} \\
		& = f_\omega(\widehat\omega_1, \widehat\omega_1, \widehat\mu_1; D^n) \\
		& \hphantom{=} + \E\{f_\omega(\widehat\omega_1, \omega_1, \widehat\mu_1; D^n) - f_\omega(\widehat\omega_1, \widehat\omega_1, \widehat\mu_1; D^n) \mid A_2 = 1, \widehat\omega_1, \widehat\mu_1, D^n\} 
	\end{align*}
	In this light, by the smoothness assumption on $f_\omega(t_1, t_2, t_3; D^n)$, we have
	\begin{align*}
		& \left|\E\left[\frac{K_{h1}(\widehat\omega_2)K_{h1}(\widehat\mu_2)}{Q(\widehat\omega_1, \widehat\mu_1)} A_2\{\E(\mu - \widehat\mu \mid A = 1, \widehat\omega_2, \widehat\mu_2, D^n) - \E(\mu - \widehat\mu \mid A = 1, \widehat\omega_1, \widehat\mu_1, D^n)\} \mid D^n, Z_1\right] \right| \\
		& \leq \left| \E\left[\frac{K_{h1}(\widehat\omega_2)K_{h1}(\widehat\mu_2)}{Q(\widehat\omega_1, \widehat\mu_1)} A_2\{f_\omega(\widehat\omega_2, \widehat\omega_2, \widehat\mu_2; D^n) - f_\omega(\widehat\omega_1, \widehat\omega_1, \widehat\mu_1; D^n) \mid D^n, Z_1\right] \right| \\
		& \hphantom{=} + \left| \E\left[\frac{K_{h1}(\widehat\omega_2)K_{h1}(\widehat\mu_2)}{Q(\widehat\omega_1, \widehat\mu_1)} A_2\{f_\omega(\widehat\omega_2, \omega_2, \widehat\mu_2; D^n) - f_\omega(\widehat\omega_2, \widehat\omega_2, \widehat\mu_2; D^n)\} \mid D^n, Z_1\right] \right| \\
		& \hphantom{=} +  \left| \E\{f_1(\widehat\omega_1, \omega_1, \widehat\mu_1; D^n) - f_\omega(\widehat\omega_1, \widehat\omega_1, \widehat\mu_1; D^n) \mid A_1 = 1,\widehat\omega_1, \widehat\mu_1, D^n\} \right| \\
		& \lesssim h^\alpha + \|\widehat\omega - \omega\|^\alpha_\infty
	\end{align*}
	We thus have:
	\begin{align*}
		& \left|\E\left[\frac{K_{h1}(\widehat\omega_2)K_{h1}(\widehat\mu_2)}{Q(\widehat\omega_1, \widehat\mu_1)} A_2\{\E(\mu_2 - \widehat\mu_2 \mid A_2 = 1, \widehat\omega_2, \widehat\mu_2, D^n) - \E(\mu_1 - \widehat\mu_1 \mid A_1 = 1, \widehat\omega_1, \widehat\mu_1, D^n)\} \mid D^n, Z_1\right] \right| \\
		& \lesssim (h^\alpha + \|\widehat\omega - \omega\|^\alpha_\infty) \land \|\widehat\mu - \mu\|_\infty.
	\end{align*}
	An analogous reasoning yields
	\begin{align*}
		& \left| \E\left[A_1\{\E(\widehat\omega_1 - \omega_1 \mid A_1 = 1, \widehat\mu_1, \widehat\omega_1, D^n) - \E(\widehat\omega_2 - \omega_2 \mid A_2 = 1, \widehat\mu_2, \widehat\omega_2, D^n)\}\frac{K_{h1}(\widehat\omega_2)K_{h1}(\widehat\mu_2)}{Q(\widehat\omega_2, \widehat\mu_2)} \mid Z_2, D^n\right] \right| \\
		& \lesssim (h^\beta + \|\widehat\mu - \mu\|_\infty) \land \|\widehat\omega - \omega\|_\infty).
	\end{align*}
	Therefore,
	\begin{align*}
		& \left| \kappa_1(Z; D^n) + \kappa_2(Z; D^n) - \widehat\varphi_{\omega\mu}(Z; D^n) \right| \\
  & \lesssim \{(h^\alpha + \|\widehat\omega - \omega\|_\infty) \land \|\widehat\mu - \mu\|_\infty\} + \{(h^\beta + \|\widehat\mu - \mu\|_\infty) \land \|\widehat\omega - \omega\|_\infty)\} + h.
	\end{align*}
	\subsubsection{Proof of Lemma \ref{lemma:Sn}}
	We have that
	\begin{align*}
		\left|  \kappa_1(Z_1; D^n) \right| & = \left| \int \kappa(Z_1, z_2; D^n) d\Pb(z_2) \right| \\
		& = \left| (A\widehat\omega_1 - 1) \E\left\{ \frac{K_{h1}(\widehat\omega_2, \widehat\mu_2)}{Q(\widehat\omega_1, \widehat\mu_1)}A_2(\mu_2 - \widehat\mu_2) \mid X_1, D^n\right\} \right| \\
		& \lesssim \|\widehat\mu - \mu\|_\infty \left\{A_2\frac{K_{h1}(\widehat\omega_2, \widehat\mu_2)}{Q(\widehat\omega_1, \widehat\mu_1)} \mid X_1, D^n\right\} \\
		& = \|\widehat\mu - \mu\|_\infty
	\end{align*}
	and, by Cauchy-Schwarz, that
	\begin{align*}
		\left|  \kappa_1(Z_1; D^n) \right| & = \left| (A\widehat\omega_1 - 1) \E\left\{ \frac{K_{h1}(\widehat\omega_2, \widehat\mu_2)}{Q(\widehat\omega_1, \widehat\mu_1)}A_2(\mu_2 - \widehat\mu_2) \mid X_1, D^n\right\} \right| \\
		& \leq \left|A\widehat\omega_1 - 1\right| \left[ \E\left\{ A_2\frac{K^2_{h1}(\widehat\omega_2, \widehat\mu_2)}{Q^2(\widehat\omega_1, \widehat\mu_1)} \mid X_1, D^n\right\}\right]^{1/2} \left[\E\left\{(\mu_2 - \widehat\mu_2)^2 \mid X_1, D^n\right\}\right]^{1/2} \\
		& \lesssim h^{-1}\|\widehat\mu - \mu\|.
	\end{align*}
	Similarly, we have
	\begin{align*}
		\left|\kappa_2(Z_2; D^n) \right| & = \left| \int \kappa(z_1, Z_2; D^n) d\Pb(z_1) \right| \\
		& = \left| \E\left\{A_1(\widehat\omega_1 - \omega_1) \frac{K_{h1}(\widehat\omega_2, \widehat\mu_2)}{Q(\widehat\omega_1, \widehat\mu_1)}\mid X_2, D^n\right\} A_2(Y_2 - \widehat\mu_2) \right| \\
		& \lesssim \|\widehat\omega - \omega\|_\infty
	\end{align*}
	and, by Cauchy-Schwarz,
	\begin{align*}
		\left|\kappa_2(Z_2; D^n) \right| & = \left| \E\left\{A_1(\widehat\omega_1 - \omega_1) \frac{K_{h1}(\widehat\omega_2, \widehat\mu_2)}{Q(\widehat\omega_1, \widehat\mu_1)}\mid X_2, D^n\right\} A_2(Y_2 - \widehat\mu_2) \right| \\
		& \leq \E\left[\left\{(\widehat\omega_1 - \omega_1)^2 \mid X_2, D^n\right\} \right]^{1/2} \left[ \E\left\{\frac{K^2_{h1}(\widehat\omega_2, \widehat\mu_2)}{Q^2(\widehat\omega_1, \widehat\mu_1)} \mid X_2, D^n\right\} \right]^{1/2}\left|A_2(Y_2 - \widehat\mu_2)\right| \\
		& \lesssim h^{-1}\|\widehat\omega - \omega\|. 
	\end{align*}
	The bound on $\E(S^2_n \mid D^n)$ follows as in the proof of the bound on $\E(S^2_{n\omega} \mid D^n)$ in Lemma \ref{lemma:Sna}. Recall that
	\begin{align*}
		S_{n} = \frac{1}{n(n-1)}\mathop{\sum\sum}_{1 \leq i \neq j \leq n} \left[ \kappa(Z_i, Z_j; D^n) - \kappa_{1}(Z_i; D^n) - \kappa_{2}(Z_j; D^n)+ \E\left\{ \kappa(Z_1, Z_2; D^n) \mid D^n \right\}\right],
	\end{align*}
	and notice that
	\begin{align*}
		& \E\left[\kappa(Z_i, Z_j; D^n) - \kappa_{1}(Z_i; D^n) - \kappa_{2}(Z_j; D^n)+ \E\left\{ \kappa(Z_1, Z_2; D^n) \mid D^n \right\} \mid Z_i, D^n\right] \\
		& = \E\left[\kappa(Z_i, Z_j; D^n) - \kappa_{1}(Z_i; D^n) - \kappa_{2}(Z_j; D^n)+ \E\left\{ \kappa(Z_1, Z_2; D^n) \mid D^n \right\} \mid Z_j, D^n\right] \\
		& = 0.
	\end{align*}
	Therefore,
	\begin{align*}
		\E(S^2_{n} \mid D^n) & \lesssim \frac{1}{n^2} \E\left\{\kappa^2(Z_1, Z_2; D^n) \mid D^n \right\} \lesssim \frac{1}{(nh)^2} \E\left\{A_2K_{h1}(\widehat\omega_2)K_h(\widehat\mu_2) \mid D^n\right\}  \lesssim (nh)^{-2}.
	\end{align*}
\section{Proof of Proposition \ref{prop:regularity}}
Because $\psi_{n^{-1/2}} - \psi = n^{-1/2}\E(d_\mu + \mu d_f) + o(n^{-1})$, if
\begin{align*}
    \widehat\psi - \psi = (\Pn - \Pb)(\overline\varphi - \overline\varphi_{\omega\mu}) + o_\Pb(n^{-1/2}),
\end{align*}
we have
	\begin{align*}
		n^{1/2}(\widehat\psi - \psi_n) \stackrel{P^n}{\indist} N\left(-\E(d_\mu + \mu d_f), \var(\overline\varphi - \overline\varphi_{\omega\mu})\right),
	\end{align*}
 where $P^n$ is the product measure of the sample. By, for example, Lemma 25.14 in \cite{van2000asymptotic}, paths $t \mapsto P_t$ with score $B(O)$ that are differentiable in quadratic mean, satisfy the equation:
	\begin{align*}
		\log \prod_{i = i}^n \frac{p_{n}(O_i)}{p(O_i)} = n^{-1/2} \sum_{i = 1}^n B(O_i) - \frac{1}{2} \E\{B^2(O)\} + o_P(1).
	\end{align*}
	Therefore, 
 \begin{align*}
     \begin{pmatrix}
     n^{1/2}(\widehat\psi - \psi_n) \\
     \vphantom{\Sigma} \\
     \log \prod_{i = i}^n \frac{p_{n}(O_i)}{p(O_i)}
     \end{pmatrix} & = 
     \begin{pmatrix}
         n^{-1/2}\sum_{i = 1}^n \{\overline\varphi(O_i) - \overline\varphi_{\omega\mu}(O_i) - \psi\} - \E(d_\mu + \mu d_f) \\
         \vphantom{\Sigma} \\
         n^{-1/2}\sum_{i = 1}^n B(d_\omega, d_\mu, d_f)(O_i) - \frac{1}{2}\E\{B^2(d_\omega, d_\mu, d_f)\}
     \end{pmatrix} + o_{P^n}(1) \\
     & \stackrel{P^n}{\indist} N\begin{pmatrix} \begin{bmatrix} -\E(d_\mu + \mu d_f), \\
         \vphantom{\Sigma} \\
     - \frac{1}{2} \E\{B^2(d_\omega, d_\mu, d_f)\} \end{bmatrix}, \Sigma
     \end{pmatrix},
 \end{align*}
 where
 \begin{align*}
   \Sigma = \begin{bmatrix}
         \var(\overline\varphi - \overline\varphi_{\omega\mu}) & \E\{B(d_\omega, d_\mu, d_f)(\overline\varphi - \overline\varphi_{\omega\mu})\} \\
         \E\{B(d_\omega, d_\mu, d_f)(\overline\varphi - \overline\varphi_{\omega\mu})\} & \var\{B(d_\omega, d_\mu, d_f)\}
     \end{bmatrix}
 \end{align*}
By Le Cam's third lemma (e.g., Example 6.7 in \cite{van2000asymptotic}), under sampling from $p_n(o) = p_{\eta_t}$ at $t = n^{-1/2}$, we have
	\begin{align*}
		n^{1/2}(\widehat\psi - \psi_n) \indist N(\theta, \var(\overline\varphi - \overline\varphi_{\omega\mu})), \text{ where } \theta = \E\{B(d_\omega, d_\mu, d_f)(\overline\varphi - \overline\varphi_{\omega\mu})\} - \E(d_\mu + \mu d_f),
	\end{align*}
	We thus compute $\theta$ under the 3 cases considered:
	\begin{enumerate}
		\item $\overline\omega = \omega$ and $\overline\mu = \mu$. In this case, $\overline\varphi - \overline\varphi_{\omega\mu} = \varphi$ and we have
		\begin{align*}
			\E\{B(d_\omega, d_\mu, d_f)(\overline\varphi - \overline\varphi_{\omega\mu})\} & =  -\E\left\{ \frac{A\omega - 1}{\omega(\omega-1)} d_\omega  \varphi \right\} + \E\left\{ \frac{A(Y - \mu)}{\mu(1-\mu)} d_\mu \varphi\right\} + \E(d_f \varphi) \\
			& = 0 + \E\left\{ \frac{A\omega(Y - \mu)^2}{\mu(1-\mu)} d_\mu\right\} + \E(\mu d_f) \\
			& = \E(d_\mu + \mu d_f). 
		\end{align*}
		Therefore, $\theta = 0$.
		\item $\overline\omega = \omega$ and $\overline\mu \neq \mu$. In this case, 
		\begin{align*}
			\overline\varphi - \overline\varphi_{\omega\mu} & = A\omega\{Y - \E(\mu \mid A = 1, \omega, \overline\mu)\} + \E(\mu \mid A = 1, \omega, \overline\mu) \\
			& = A\omega(Y - \mu) + h(\omega, \overline\mu) + A\omega\{\mu - h(\omega, \overline\mu)\}
		\end{align*}
		We have
		\begin{align*}
			& \E\{B(d_\omega, d_\mu, d_f)(\overline\varphi - \overline\varphi_{\omega\mu})\} \\
			& = -\E\left\{ \frac{A\omega - 1}{\omega(\omega-1)} d_\omega(\overline\varphi - \overline\varphi_{\omega\mu})\right\} + \E\left\{ \frac{A(Y - \mu)}{\mu(1-\mu)} d_\mu (\overline\varphi - \overline\varphi_{\omega\mu})\right\} + \E\{d_f(\overline\varphi - \overline\varphi_{\omega\mu})\} \\
			& = -\E\left[\frac{(A\omega - 1)A\omega\{Y - h(\omega, \overline\mu)\}}{\omega(\omega-1)} d_\omega \right] + \E\left\{ \frac{A\omega(Y - \mu)^2}{\mu(1-\mu)} d_\mu\right\} + \E\left[\frac{A(Y - \mu)}{\mu(1-\mu)} \omega \{\mu - h(\omega, \overline\mu)\} d_\mu \right] \\
			& \hphantom{=} + \E\left[d_fA\omega\{Y - h(\omega, \overline\mu)\} + d_f h(\omega, \overline\mu) \right] \\
			& = -\E\left[A\{\mu - h(\omega, \overline\mu)\}d_\omega \right] + \E(d_\mu + \mu d_f).
		\end{align*}
		Therefore $\theta = -\E\left[A\{\mu - h(\omega, \overline\mu)\}d_\omega\right]$.
		\item $\overline\omega \neq \omega$ and $\overline\mu = \mu$. In this case,
		\begin{align*}
			\overline\varphi - \overline\varphi_{\omega\mu} = A\E(\omega \mid A = 1, \overline\omega, \mu) (Y - \mu) + \mu = Ag(\overline\omega, \mu)(Y - \mu) + \mu.
		\end{align*}
		We have
		\begin{align*}
			& \E\{B(d_\omega, d_\mu, d_f)(\overline\varphi - \overline\varphi_{\omega\mu})\} \\
			& = -\E\left\{ \frac{A\omega - 1}{\omega(\omega-1)} d_\omega(\overline\varphi - \overline\varphi_{\omega\mu})\right\} + \E\left\{ \frac{A(Y - \mu)}{\mu(1-\mu)} d_\mu (\overline\varphi - \overline\varphi_{\omega\mu})\right\} + \E\{d_f(\overline\varphi - \overline\varphi_{\omega\mu})\} \\
			& = 0 + \E\{Ag(\overline\omega, \mu) d_\mu \} + \E(\mu d_f).
		\end{align*}
		In this case, $\theta = \E\left[\{Ag(\overline\omega, \mu) - 1\} d_\mu\right]$.
	\end{enumerate}
\end{document}